\documentclass[11pt]{article}
\usepackage[margin=1in,bottom=1in]{geometry}
\usepackage[utf8]{inputenc}
\usepackage[T1]{fontenc}
\usepackage[english]{babel}

\usepackage{tabularx}
\usepackage{mdframed}
\usepackage{libertine}
\usepackage{amsthm}
\usepackage{amsmath}

\usepackage{amssymb}
\usepackage{stmaryrd}
\usepackage{mathrsfs} 
\usepackage{mathtools}
\usepackage{thm-restate}
\usepackage[colorlinks = true,linkcolor = blue,urlcolor = blue, citecolor = blue]{hyperref}
\usepackage[capitalize, nameinlink]{cleveref}
\usepackage{url}
\usepackage{todonotes}
\usepackage{mathrsfs} 
\usepackage{soul}
\usepackage{relsize}
\usepackage{enumitem}
\usepackage[all,defaultlines=3]{nowidow}
\usepackage[mathlines]{lineno}
\usepackage{multicol}
\usepackage{enumitem}
\usepackage{xspace}
\usepackage{lineno}
\usepackage{nicefrac}
\usepackage{wasysym}
\usepackage{comment}
\usepackage[linesnumbered,ruled]{algorithm2e}
\usepackage{textpos}
\setlist[itemize]{topsep=4pt,itemsep=3pt,parsep=0pt} 
\setlist[enumerate]{topsep=4pt,itemsep=3pt,parsep=0pt} 

\usetikzlibrary{decorations.markings}

\crefname{algocf}{Algorithm}{Algorithms}
\crefname{claim}{Claim}{Claims}
\crefname{figure}{Figure}{Figures}
\renewcommand{\preceq}{\preccurlyeq}

\newtheorem{theorem}{Theorem}[section]
\newtheorem*{theorem*}{Theorem}

\newtheorem{corollary}[theorem]{Corollary}

\newtheorem{lemma}[theorem]{Lemma}

\newtheorem{claim}[theorem]{Claim}

\newtheorem{proposition}[theorem]{Proposition}
\theoremstyle{definition}
\newtheorem{definition}[theorem]{Definition}
\theoremstyle{plain}

\theoremstyle{definition}

\newtheorem*{example*}{Example}

\numberwithin{equation}{section}

\renewcommand{\cal}{\mathcal}

\newenvironment{claimproof}[1][Proof of the claim.]{%
  \begin{proof}[#1]%
}{%
  \end{proof}%
}

\newcommand{\drds}[1]{\textsc{Distance-}$#1$ \textsc{Dominating Set}}
\newcommand{\dris}[1]{\textsc{Distance-}$#1$ \textsc{Independent Set}}
\newcommand{\notDominated}[1]{\mathtt{witnessNotDom}(#1)}
\newcommand{\candidate}[1]{\mathtt{candidateDom}(#1)}
\newcommand{\witnessInd}[1]{\mathtt{witnessNotInd}(#1)}
\newcommand{\candidateInd}[1]{\mathtt{candidateInd}(#1)}
\newcommand{\add}[1]{\mathtt{insert}(#1)}
\newcommand{\remove}[1]{\mathtt{remove}(#1)}
\newcommand{\near}[1]{\mathtt{nearVertex}(#1)}
\newcommand{\far}[1]{\mathtt{farVertex}(#1)}
\newcommand{\init}[1]{\mathtt{init}(#1)}
\newcommand{\farVertexDS}[2]{\mathsf{FarVertex}_{#2}[#1]}
\newcommand{\farVertexDSS}[3]{\mathsf{FarVertex}_{#2}^{#3}[#1]}
\newcommand{\nearVertexDS}[2]{\mathsf{NearVertex}_{#2}[#1]}
\newcommand{\nearVertexDSS}[3]{\mathsf{NearVertex}_{#2}^{#3}[#1]}
\newcommand{\dominatingSetDS}[2]{\mathsf{DominatingSet}_{#2}[#1]}
\newcommand{\independentSetDS}[2]{\mathsf{IndependentSet}_{#2}[#1]}
\newcommand{\ds}[1]{\mathtt{domSet}(#1)}
\newcommand{\is}[1]{\mathtt{indSet}(#1)}
\newcommand{\FO}{\mathsf{FO}}

\newcommand{\avgdeg}{\mathsf{avgdeg}}

\newcommand{\Minors}{\mathsf{Minors}}
\newcommand{\Fra}{\mathsf{Fraternal}}
\newcommand{\fra}{\mathsf{fraternal}}

\newcommand{\ar}{\mathsf{ar}}
\newcommand{\Af}{\mathbb{A}}
\newcommand{\Bf}{\mathbb{B}}

\newcommand{\wei}{\mathbf{w}}

\newcommand{\Pp}{\mathcal{P}}

\newcommand{\tup}[1]{{\bar{#1}}}

\newcommand{\Prob}{\mathbb{P}}
\renewcommand{\Pr}{\Prob}

\newcommand{\rt}[2]{#1\langle #2\rangle}

\def\phi{\varphi}
\newcommand{\N}{\mathbb{N}}

\newcommand{\Z}{\mathbb{Z}}

\newcommand{\dist}{\mathrm{dist}}

\newcommand{\sli}{\mathsf{sli}}

\newcommand{\Ff}{\mathcal{F}}
\newcommand{\Ss}{\mathcal{S}}

\newcommand{\F}{\mathcal{F}}
\newcommand{\G}{\Cc}
\newcommand{\val}{\mathrm{val}}

\newcommand{\Hom}{\mathrm{Hom}}
\newcommand{\Sub}{\mathrm{Sub}}
\newcommand{\ISub}{\mathrm{ISub}}

\newcommand{\Cc}{\mathscr C}

\def\epsilon{\varepsilon}
\def\eps{\varepsilon}
\newcommand{\Oh}{\mathcal{O}}

\newcommand{\id}{\mathrm{id}}
\renewcommand{\emptyset}{\varnothing}

\newcommand{\michalin}[1]{\todo[size=\normalsize,inline,color=pink!40]{Micha\l{}: #1}}

\newcommand{\bbinline}[2][inline]{\todo[color=blue!50,#1]{{\textbf{B:} #2}}}
\newcommand{\azinline}[2][inline]{\todo[color=yellow!50,#1]{{\textbf{A:} #2}}}

\newcommand{\WReach}{\textrm{WReach}}
\newcommand{\wcol}{\textrm{wcol}}

\newcommand{\UnMaps}{\mathsf{UnweightedMappings}}
\newcommand{\WeiMaps}{\mathsf{WeightedMappings}}
\newcommand{\MapEx}{\mathsf{MappingExample}}
\newcommand{\RetVer}{\mathsf{RetrieveVertex}}
\newcommand{\WeiSum}{\mathsf{WeightedSum}}
\newcommand{\Ver}{\mathsf{Verify}}
\newcommand{\MapVer}{\mathsf{MapVertex}}

\renewcommand{\leq}{\leqslant}
\renewcommand{\geq}{\geqslant}
\renewcommand{\le}{\leq}
\renewcommand{\ge}{\geq}

\renewcommand{\setminus}{-}
\newcommand{\set}[1]{\{#1\}}
\newcommand{\setof}[2]{\{#1 \colon#2\}}

\usepackage{tikz}

\SetKwBlock{Repeat}{repeat}{}
\let\oldnl\nl
\newcommand{\nonl}{\renewcommand{\nl}{\let\nl\oldnl}}

\newcommand{\affi}[1]{\textcolor{black!50}{#1}}

\tikzstyle{vertex}=[circle,draw=black,fill=gray,minimum size=0.15cm,inner sep=0pt]

\tikzstyle{svertex}=[circle,draw=black,fill=black,minimum size=0.11cm,inner sep=0pt]

\begin{document}

\newcommand{\funding}{
	\begin{minipage}[t]{0.30\textwidth}
		\includegraphics[scale=0.15]{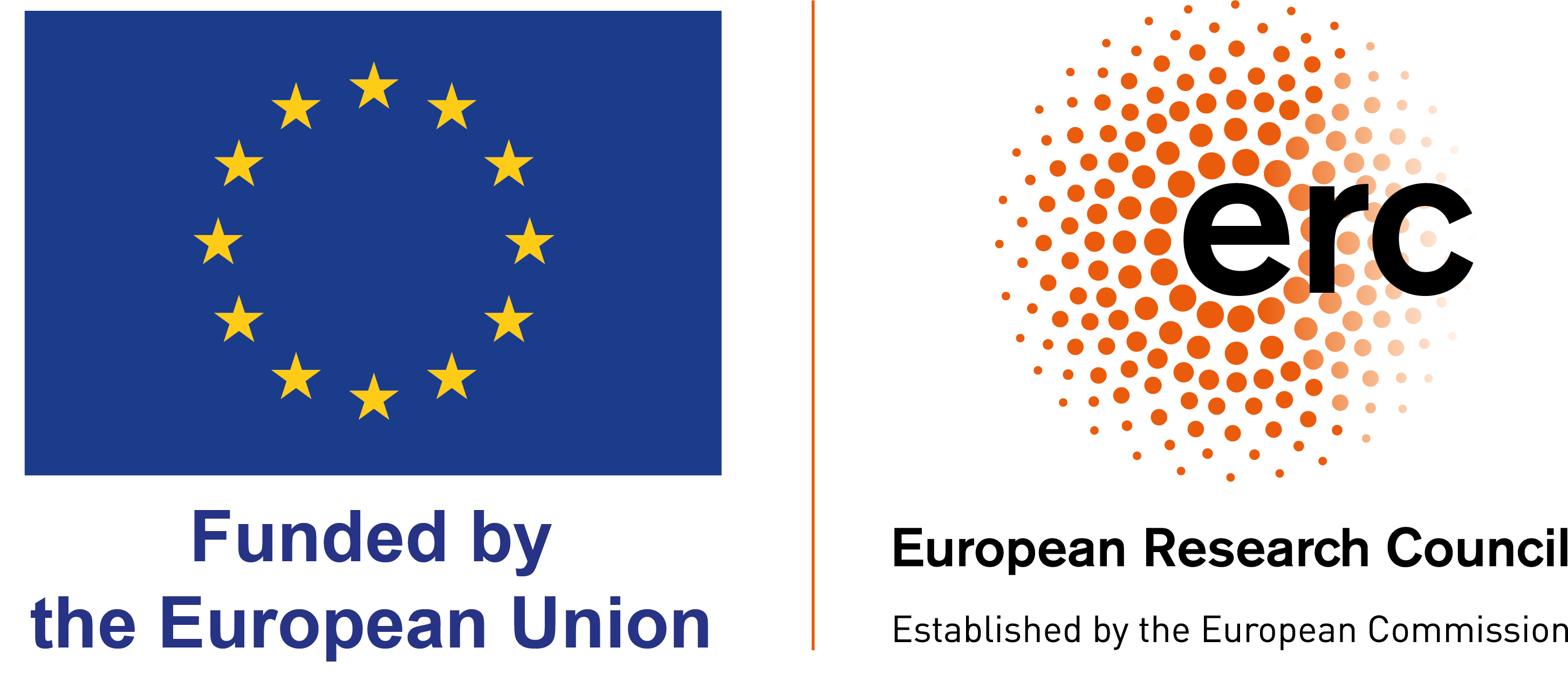}
	\end{minipage}\hfill
	\begin{minipage}[b]{0.65\textwidth}\footnotesize
		All the authors were supported by the project BOBR that received funding from the European Research Council (ERC) under the European Union’s Horizon 2020 research and innovation programme, grant agreement No. 948057. In particular, a majority of the work on this project was done while B. Bosek was also affiliated with the University of Warsaw.
	\end{minipage}}

\renewcommand\footnotemark{}
\title{Dynamic domination and independence in sparse graphs\footnote{\funding}}
\date{}
\author{
	Bart\l{}omiej Bosek\\{\small \affi{Jagiellonian University, Kraków}}\\\href{mailto:bartlomiej.bosek@uj.edu.pl}{\small bartlomiej.bosek@uj.edu.pl}\\
	\and
	Wojciech Nadara\\{\small \affi{University of Warsaw}}\\\href{mailto:w.nadara@mimuw.edu.pl}{\small w.nadara@mimuw.edu.pl}
	\and
	Michał Pilipczuk\\{\small \affi{University of Warsaw}}\\\href{mailto:michal.pilipczuk@mimuw.edu.pl}{\small michal.pilipczuk@mimuw.edu.pl}
	\and
	Anna Zych-Pawlewicz\\{\small \affi{University of Warsaw}}\\\href{mailto:anka@mimuw.edu.pl}{\small anka@mimuw.edu.pl}
}

\maketitle

\begin{abstract}
	Let $\Cc$ be a class of graphs of bounded expansion and $r,k\in \N$ be fixed. We give a dynamic data structure that for a given dynamic graph $G$, updated by edge insertions and deletions subject to the promise that $G\in \Cc$ at all times, maintains the answer to the following two queries:
	\begin{itemize}[nosep]
		\item Does $G$ contain a distance-$r$ dominating set of size $k$?
		\item Does $G$ contain a distance-$r$ independent set of size $k$?
	\end{itemize}
	The data structure is randomized with error probability bounded by $\eps$, for a parameter $\eps>0$ fixed upon the initialization. The amortized update time is $\log^c n\cdot \log \frac{1}{\eps}$, where $n$ is the vertex count of $G$ and $c$ is a constant that depends only on $r$, $k$, and $\Cc$. In the case of the first query, the data structure can also output a distance-$r$ dominating set of size $k$, if existent.
	
	We also prove that when $r=1$, our data structure for the dominating set query can be implemented even if we only assume that the maintained graph $G$ has degeneracy bounded by a constant $d$, yielding a simpler data structure with an improved amortized update time of $2^{k^{\Oh(d)}}\cdot \log^3 n\cdot \log \tfrac{1}{\eps}$. Finally, we prove that in graphs of degeneracy at most $d$, one can maintain an $\Oh(d^2)$-approximation of the minimum size of a (distance-$1$) dominating set with amortized expected update time $d^{\Oh(1)}\cdot \log n$.
\end{abstract}


\thispagestyle{empty}
\newpage

\tableofcontents

\thispagestyle{empty}
\newpage

\clearpage
\setcounter{page}{1}

\section{Introduction}\label{sec:intro}

\emph{Sparsity} is a research area in structural graph theory that studies the structure in graphs that exclude dense local obstructions, formalized through the notion of shallow minors. More precisely, we say that a graph $H$ is a \emph{depth-$r$ minor} of a graph $G$ if a supergraph of $H$ can be obtained from $G$ by contracting mutually disjoint connected subgraphs of radius at most $r$. The two main concepts are the following:
\begin{itemize}
	\item A graph class $\Cc$ has \emph{bounded expansion} if for every $r\in \N$ there is a finite upper bound $c_r\in \N$ on the average degree of depth-$r$ minors of graphs from $\Cc$.
	\item A graph class $\Cc$ is \emph{nowhere dense} if for every $r\in \N$ there is a finite upper bound $t_r\in \N$ on the sizes of complete graphs that can be found as depth-$r$ minors of graphs from $\Cc$.
\end{itemize}
Clearly, if $\Cc$ has bounded expansion then $\Cc$ is nowhere dense as well, but the reverse implication does not hold in general. Many well-studied classes of sparse graphs, such as graphs with bounded maximum degree, graphs excluding a fixed (topological) minor (in particular planar graphs), or some natural classes of sparse geometric intersection graphs, do have bounded expansion. Therefore, any methods designed for bounded expansion classes will apply to them as well. We refer the reader to the monograph of Ne\v{s}et\v{r}il and Ossona de Mendez~\cite{sparsity} or to the lecture notes of Pilipczuk, Pilipczuk, and Siebertz~\cite{sparsityNotes} for a broad introduction to the vast toolbox of techniques for classes of bounded expansion and nowhere dense classes.

The methods of Sparsity are particularly useful in the context of parameterized algorithms, as witnessed by the award of the Nerode Prize 2025 to Ne\v{s}et\v{r}il and Ossona de Mendez for laying foundations of the field. Not surprisingly, these methods apply predominantly to parameterized problems of local nature, such as the following:
\begin{itemize}
	\item \textsc{$H$-(Induced)-Subgraph Isomorphism}: Decide whether the fixed graph $H$ is an (induced) subgraph of the given graph $G$.
	\item \drds{r}: Decide whether the given graph $G$ has a \emph{distance-$r$ dominating set} of size $k$, that is, a set $D$ of $k$ vertices such that every vertex of $G$ is at distance at most $r$ from some vertex of $D$.
	\item \dris{r}: Decide whether the given graph $G$ has a \emph{distance-$r$ independent set} of size $k$, that is, a set $I$ of $k$ vertices that are pairwise at distance more than $r$ from each other.
\end{itemize}
All these problems are $\mathsf{W}[1]$-hard on general graphs, but become fixed-parameter tractable on any nowhere dense class of graphs. This is a manifestation of a more general phenomenon. Namely, each of the problems above can be expressed by a sentence of first-order logic $\FO$, of length depending on the relevant parameter(s): the vertex count of $H$, or $r$ and $k$. As proved by Grohe, Kreutzer, and Siebertz~\cite{GroheKS17}, the \emph{model-checking problem} for~$\FO$ --- given a graph $G$ and an $\FO$ sentence $\varphi$, decide whether $\varphi$ holds in $G$ ---  can be solved in almost linear fixed-parameter time on any nowhere dense class $\Cc$; precisely in time $\Oh_{\Cc,\varphi,\eps}(|G|^{1+\eps})$, for any fixed~$\eps>0$, where $|G|$ denotes the number of vertices of $G$. As proved earlier by Dvo\v{r}\'ak, Kr\'al', and Thomas~\cite{DvorakKT13}, on bounded expansion classes one can even obtain a linear fixed-parameter time of $\Oh_{\Cc,\varphi}(|G|)$. While fixed-parameter tractability of \textsc{$H$-(Induced)-Subgraph Isomorphism}, \drds{r}, \dris{r} on nowhere dense classes follows from these meta-theorems, dedicated algorithms for those problems were known already earlier~\cite{DawarK09,NesetrilM08a}. In fact, their development was vital in establishing the toolbox needed~for~\cite{DvorakKT13,GroheKS17}.

Since the fundamental work of Dvo\v{r}\'ak et al.~\cite{DvorakKT13} and of Grohe et al.~\cite{GroheKS17}, the approach to model-checking on classes of sparse graphs has been thoroughly understood, particularly for classes of bounded expansion. The developed framework is robust and has been successfully extended to the settings of enumeration, counting, and other aggregate queries~\cite{GroheS18,KazanaS19,SchweikardtSV22,Torunczyk20}, circuit complexity~\cite{PilipczukST18}, approximation~\cite{Dvorak22}, and distributed computing~\cite{BlinFFGGMRT26}. One setting that has so far resisted progress, despite efforts, is that of dynamic data structures. 

Precisely, consider the following setting. Fix a class of bounded expansion $\Cc$ and an $\FO$ sentence $\varphi$. Let $G$ be a fully dynamic graph that is updated over time by edge insertions and edge deletions (the vertex set stays fixed), subject to a guarantee that $G\in \Cc$ at all times. The goal is to design a data structure that efficiently maintains whether $\varphi$ is satisfied in $G$.

The question about the existence of such a data structure was asked in 2013 by Dvo\v{r}\'ak and T\r{u}ma~\cite{DvorakT13}, who showed that this is indeed possible for the \textsc{$H$-Induced-Subgraph Isomorphism} problem. Precisely, they gave a dynamic data structure that is able to maintain the number of induced copies of $H$ in $G$ (as well as the number of $H$-subgraphs in $G$ and the number of homomorphisms from $H$ to~$G$) with amortized update time $\Oh_{\Cc,h}(\log^{\binom{h}{2}-1} n)$, where $h$ and $n$ denote the vertex counts of $H$ and $G$, respectively. Somewhat curiously, their approach is based on the Inclusion-Exclusion Principle, so even though their data structure can count the number of solutions, it is not capable of actually providing an example solution.

\textsc{$H$-Induced-Subgraph Isomorphism} is a specific example of an $\FO$-definable problem, because the sentence expressing it is purely existential: After quantifying the vertices of $H$   existentially, the adjacencies and nonadjacencies can be verified using a quantifier-free subformula. Consequently, algorithms for this problem on classes of sparse graphs are typically much simpler; see e.g.~\cite{NesetrilM08a}. Since the work of Dvo\v{r}\'ak and T\r{u}ma in 2013, no progress on their question has been reported, even for the more complicated $\exists\forall$ problems such as \drds{r} and \dris{r}.

We remark that some previous works~\cite{DvorakKT13,Torunczyk20} considered the partially dynamic setting where (essentially) the graph stays fixed and its vertices and edges gets recolored, but this setting is much simpler than the fully dynamic one that we study here.

\paragraph*{Our contribution.} In this work we propose fully dynamic data structures for the \drds{r} and \dris{r} problems on classes of bounded expansion, with polylogarithmic amortized update time.\footnote{For a tuple of parameters $\tup p$, the $\Oh_{\tup p}(\cdot)$ notation hides multiplicative factors that may depend on $\tup p$.} 

\begin{restatable}{theorem}{mainDomBE}\label{thm:main-domset-be}
	Fix a graph class $\Cc$ of bounded expansion, $r,k\in \N$, and $\eps>0$. Then there is a randomized dynamic data structure that given a dynamic $n$-vertex graph $G$, updated by edge insertions and removals subject to guarantee that it always belongs to $\Cc$, supports the following query:
	\begin{itemize}
		\item $\ds{}$: return a distance-$r$ dominating set of size at most $k$ in $G$, or $\bot$ if no such set exists.
	\end{itemize}
	Every answer to the query is correct with probability at least $1-\eps$ against an oblivious adversary.
	The amortized time complexity of an update or a query is $\log^{\Oh_{\Cc,r,k}(1)} n\cdot \log \tfrac{1}{\eps}$. The data structure can be initialized on an edgeless graph in time $\Oh_{\Cc,r,k}(n\log n\log \tfrac{1}{\eps})$ and occupies  $\Oh_{\Cc,r,k}(n\log n\log \tfrac{1}{\eps})$ space at all times.
\end{restatable}

\vspace{-0.3cm}

\begin{restatable}{theorem}{mainIndBE}\label{thm:main-indset-be}
	Fix a graph class $\Cc$ of bounded expansion, $r,k\in \N$, and $\eps>0$. Then there is a randomized dynamic data structure that given a dynamic $n$-vertex graph $G$, updated by edge insertions and removals subject to guarantee that it always belongs to $\Cc$, supports the following query:
	\begin{itemize}
		\item $\is{}$: decide whether in $G$ there exists a distance-$r$ independent set of size $k$.
	\end{itemize}
	Every answer to the query is correct with probability at least $1-\eps$ against an oblivious adversary.
	The amortized time complexity of an update or a query is $\log^{\Oh_{\Cc,r,k}(1)} n\cdot \log \tfrac{1}{\eps}$. The data structure can be initialized on an edgeless graph in time $\Oh_{\Cc,r,k}(n\log n\log \tfrac{1}{\eps})$ and occupies  $\Oh_{\Cc,r,k}(n\log n\log \tfrac{1}{\eps})$ space at all times.
\end{restatable}

Note that the data structure of \cref{thm:main-domset-be} is capable of providing an example solution, while that of \cref{thm:main-indset-be} only reports the existence of a solution. Both data structures are randomized against an oblivious adversary, which means that the guarantees on the probability of correctness hold only assuming the adversary cannot choose the next update/query based on the data structure's answers to the previous queries. We stress that the data structure's answers to the consecutive queries are \emph{not} independent random variables; they all depend on some common random bits fixed upon the initialization. 

At this point, the proofs of \cref{thm:main-domset-be,thm:main-indset-be} do not extend to $\FO$-definable problems beyond \drds{r} and \dris{r}. This is because our approach is based on dynamizing two dedicated algorithms for these problems due to Fabia\'nski, Pilipczuk, Siebertz, and Toru\'nczyk~\cite{FabianskiPST19}, obtained via their framework of \emph{progressive exploration}. We believe that this might indicate that in order to resolve the question of Dvo\v{r}\'ak and T\r{u}ma, one has to design a new $\FO$-model-checking algorithm on classes of sparse graphs that would be more in the spirit of progressive exploration.

We remark that on the way to \cref{thm:main-domset-be,thm:main-indset-be}, we resolve the question posed by Dvo\v{r}\'ak and T\r{u}ma~\cite{DvorakT13} by extending their data structure so that it can also report examples of relevant mappings. See \cref{thm:dvorak-tuma-examples} for a precise formulation.

\medskip

Next, we observe that for the specific case $r=1$ (so the classic \textsc{Dominating Set} problem), the data structure of \cref{thm:main-domset-be} can be simplified and works already under the assumption of bounded degeneracy (equivalently, of bounded arboricity or of bounded maximum average degree among subgraphs).

\begin{restatable}{theorem}{mainDomDeg}\label{thm:main-domset-deg}
	Fix $d,k\in \N$ and $\eps>0$. Then there is a randomized dynamic data structure that given a dynamic $n$-vertex graph $G$, updated by edge insertions and removals subject to guarantee that $G$ is $d$-degenerate at all times, supports the following query:
	\begin{itemize}
		\item $\ds{}$: return a dominating set of size at most $k$ in $G$, or $\bot$ if no such set exists.
	\end{itemize} 
	Every answer to the query is correct with probability at least $1-\eps$ against an oblivious adversary.
	The amortized time complexity of an update or a query is $2^{k^{\Oh(d)}}\cdot \log^3 n \log \tfrac{1}{\eps}$. The data structure can be initialized on an edgeless graph in time $k^{\Oh(d)}\cdot n\log n\log \tfrac{1}{\eps}$ and uses  $k^{\Oh(d)}\cdot n\log n\log \tfrac{1}{\eps}$ space at all times.
\end{restatable}

Finally, it is also known that on bounded expansion classes, both \drds{r} and \dris{r} admit constant-factor approximation algorithms~\cite{Dvorak13}. We prove that for $r=1$, this can be lifted to the dynamic setting, even assuming only bounded~degeneracy.

\begin{restatable}{theorem}{mainApx}\label{thm:apx-dynamic}
For every $d\in \N$ there is a randomized dynamic data structure (against an oblivious adversary) that given a dynamic $n$-vertex graph $G$, updated by edge insertions and removals, subject to guarantee that $G$ is $d$-degenerate at all times, maintains a dominating set $D$ in $G$ satisfying $|D|\leq (4d+1)^2\cdot \mathsf{dom}_1(G)$, where $\mathsf{dom}_1(G)$ denotes the minimum size of a dominating set in $G$. The access to $D$ is provided by queries that report its size in time $\Oh(1)$ and enumerate all its vertices in time $\Oh(|D|)$. The data structure can be initialized on an edgeless graph in time $\Oh(n)$, uses $\Oh(dn)$ space at all times, and processes every edge insertion and deletion in expected amortized time $d^{\Oh(1)}\cdot \log n$.
\end{restatable}

Whether \cref{thm:apx-dynamic} can be lifted to $r>1$ under the assumption that the maintained graph belongs to a fixed bounded-expansion class $\Cc$ at all times, remains an interesting open question. Note here that while the setting $r=1$ is trivial for the \dris{r} problem (an $n$-vertex $d$-degenerate graph contains an independent set of size $\frac{n}{d+1}$), it becomes again interesting for $r>1$. In fact, it is known~\cite{Dvorak13} that in bounded-expansion classes, there is a constant-factor multiplicative gap between the minimum size of a distance-$r$ dominating set and the maximum size of a distance-$(2r+1)$ independent set. So if only an approximate value of the optimum needs to be reported, the questions of dynamic approximation for \drds{r} and \dris{(2r+1)} in bounded expansion classes are~equivalent.

\paragraph*{Acknowledgements.} We thank \L{}ukasz Kowalik, Piotr Sankowski, and Marek Soko\l{}owski for many discussions over several years that, perhaps indirectly, led to this work.

\section{Overview}\label{sec:overview}

In this section we provide a high-level exposition of the technical ideas leading to our findings, particularly to the main results: \cref{thm:main-domset-be,thm:main-indset-be}.

\subsection{Progressive exploration}

The key idea behind \cref{thm:main-domset-be,thm:main-indset-be} is to develop dynamic counterparts of the \emph{progressive exploration} algorithms for \drds{r} and \dris{r} in sparse graphs, proposed by Fabia\'nski, Pilipczuk, Siebertz, and Toru\'nczyk~\cite{FabianskiPST19}. For the sake of concreteness, let us focus our discussion on the \drds{r} problem.

Suppose that we are given a graph $G$ and positive integers $r$ and $k$. We would like to find a distance-$r$ dominating set of size $k$ in $G$, or conclude that no such set exists. A progressive exploration algorithm for this problem --- called by Fabia\'nski et al. the \emph{semi-ladder algorithm} --- proceeds in rounds, and constructs a sequence of \emph{candidates} $D_1,D_2,D_3,\ldots$, which are $k$-tuples of vertices, and \emph{witnesses} $w_1,w_2,w_3,\ldots$, which are single vertices. With $D_1,\ldots,D_{i-1}$ and $w_1,\ldots,w_{i-1}$ having been constructed in rounds $1,2,\ldots,i-1$, the $i$th round proceeds as follows:
\begin{description}
	\item[Step 1:] Determine whether there exists a $k$-tuple of vertices $D_i$ that distance-$r$ dominates all the witnesses found so far, that is, $\{w_1,\ldots,w_{i-1}\}$. If not, then $G$ does not admit a distance-$r$ dominating set of size $k$ and we can terminate the algorithm.
	\item[Step 2:] Verify whether $D_i$ distance-$r$ dominates the whole graph. If so, then $D_i$ is a valid solution that can be reported. Otherwise, select any vertex $w_i$ that is not distance-$r$ dominated by $D_i$, and proceed with the new candidate $D_i$ and witness $w_i$ to the next round.
\end{description}
It is clear that when this algorithm reports an outcome, then this outcome is correct. However, a priori it is not at all clear that the algorithm will reach an outcome within a small number of rounds. This seems particularly problematic, because the task  in Step~1 --- finding a $k$-tuple of vertices that distance-$r$ dominates all the $i-1$ witnesses gathered so far --- gets more and more computationally expensive with every round. (We have not discussed its implementation so far, we will do it shortly.)

Somewhat surprisingly, Fabia\'nski et al. proved that on sparse graphs, more precisely on any nowhere dense class of graphs, the semi-ladder algorithm always terminates within a bounded number of rounds.

\begin{theorem}[follows from~\cite{FabianskiPST19}]\label{thm:overview-semiladder}
	For every $r\in \N$ and every nowhere dense class of graphs $\Cc$, there exists a constant $c$ such that the semi-ladder algorithm deployed on any graph $G\in \Cc$ with parameters $r$ and $k$, terminates after performing at most $k^c$ rounds.
\end{theorem}

See also \cref{thm:ladder-bound,thm:termination} for a more precise formulation, and recall that bounded expansion classes are in particular nowhere dense. We remark that the proof of \cref{thm:overview-semiladder} requires deep tools from the theory of Sparsity, in particular the equivalent characterization of nowhere denseness through the notion of \emph{flatness} (also known as \emph{uniform quasi-wideness}).

\cref{thm:overview-semiladder} makes the semi-ladder algorithm perfect for lifting to the dynamic setting. Namely, since the number of rounds executed by the algorithm is bounded by a constant depending only on the class $\Cc$ and the parameters $r$ and $k$, we can just run the semi-ladder algorithm after every update to the graph, provided we are able to maintain data structures that allow efficient implementation of every~round.

Consider then the $i$th round of the algorithm. The task in Step~2 boils down to a single application of the following query in $G$, with $S=D_i$ (so $|S|\leq k$):
\begin{itemize}
	\item $\far{S,r}$: Given a set $S\subseteq V(G)$, find a vertex $v$ such that $\dist(v,S)>r$ for all $s\in S$, or return $\bot$ if no such vertex exists.
\end{itemize}
On the other hand, Step~1 can be implemented using at most $k^i$ queries of the following kind, with $|S|<i$:
\begin{itemize}
	\item $\near{S,r}$: Given a set $S\subseteq V(G)$, find a vertex $v$ such that $\dist(v,S)\leq r$ for all $s\in S$, or return $\bot$ if no such vertex exists.
\end{itemize}
Indeed, we may iterate through all the partitions of the set of witnesses $\{w_1,\ldots,w_{i-1}\}$ into $k$ sets, and for each partition $\Pp$, invoke $\near{S,r}$ on each part $S$ of $\Pp$ to verify whether $S$ can be \mbox{distance-$r$} dominated by a single vertex $v_S$. Note that if for some partition $\Pp$ all these checks go through, then $D_i\coloneqq \{v_S\colon S\in \Pp\}$ is a $k$-vertex set that distance-$r$ dominates $\{w_1,\ldots,w_{i-1}\}$; and otherwise, if the checks for all the partitions fail, then $\{w_1,\ldots,w_{i-1}\}$ cannot be distance-$r$ dominated by $k$ vertices. There are at most $k^{i-1}$ partitions to check, and checking each requires $k$ applications of the $\near{S,r}$~query.

The discussion above combined with \cref{thm:overview-semiladder} amounts to the proof of the following statement.

\begin{proposition}\label{prop:overview-reduction}
	Let $r,k\in \N$ and $\Cc$ be a nowhere dense class of graphs. Then the semi-ladder algorithm on any graph $G\in \Cc$ with parameters $r$ and $k$ can be implemented using $k^{\Oh_{\Cc,r}(1)}$ queries of the form $\near{S,r}$ with $|S|\leq k^{\Oh_{\Cc,r}(1)}$, and $k^{\Oh_{\Cc,r}(1)}$ queries of the form $\far{S,r}$ with $|S|\leq k$.
\end{proposition}

So from now on our task is to design data structures that can efficiently answer queries $\near{S,r}$ and $\far{S,r}$ with sets $S$ of bounded size and $r$ being a fixed constant. As we will see shortly, our data structure for $\near{S,r}$ is a non-trivial extension of the result of Dvo\v{r}\'ak and T\r{u}ma~\cite{DvorakT13}, which in particular answers one of their open questions. On the other hand, the data structure for $\far{S,r}$ is more involved and requires new insights.

We remark that our data structures for $\near{S,r}$ and $\far{S,r}$ queries heavily rely on the assumption that the fixed class $\Cc$, to which the graph $G$ is guaranteed to belong at all times, has bounded expansion, even though \cref{prop:overview-reduction} actually works in the larger generality of nowhere dense classes. In \cref{sec:conclusions}, we comment on the possibility of lifting our results to nowhere dense classes.

\medskip

Finally, let us discuss the case of \dris{r}. For this problem, Fabia\'nski et al. gave a different algorithm, called the \emph{ladder algorithm}, which is somewhat more complicated than the semi-ladder algorithm and whose proof of correctness requires more insight. However, the bottom line is that on any nowhere dense class $\Cc$ and for fixed parameters $r$ and $k$, the ladder algorithm again terminates within $\Oh_{\Cc,r,k}(1)$ rounds and each of those rounds can be implemented using $\Oh_{\Cc,r,k}(1)$ calls to (slight generalizations of) the queries $\near{S,r}$ and $\far{S,r}$ with $|S|\leq \Oh_{\Cc,r,k}(1)$. So designing  data structures implementing these queries efficiently is sufficient to obtain the data structure for the \dris{s} promised in \cref{thm:main-indset-be}. We remark that in the case of a positive outcome, the ladder algorithm actually does not provide a distance-$r$ independent set of size $k$, but only a proof of its existence; this is why this restraint is also present in \cref{thm:main-indset-be}.

\subsection{Finding near vertices}

We start by designing a data structure for the $\near{S,r}$ queries, which is the easier part of the argument. For this, we first recall the result of Dvo\v{r}\'ak and T\r{u}ma~\cite{DvorakT13}, which requires some terminology.

For two graphs $H$ and $G$, a \emph{homomorphism} from $H$ to $G$ is a mapping $\varphi\colon V(H)\to V(G)$ such that for any edge $uv$ of $H$, $\varphi(u)\varphi(v)$ is an edge in $G$. A \emph{subgraph isomorphism} is an injective homomorphism, and an \emph{induced subgraph isomorphism} is one where the implication above is an equivalence: $uv$ is an edge in $H$ if and only if $\varphi(u)\varphi(v)$ is an edge in $G$. We denote by $\Hom(H,G),\Sub(H,G),\ISub(H,G)$ the sets of homomorphisms, subgraph isomorphisms, and induced subgraph isomorphisms from $H$ to $G$, respectively.

With this terminology, the result of Dvo\v{r}\'ak and T\r{u}ma reads as follows.

\begin{theorem}[\cite{DvorakT13}]\label{thm:overview-DT}
	Fix a bounded expansion class $\Cc$, graph $H$, and $\Ff\in \{\Hom,\Sub,\ISub\}$. Then there is a data structure that for a dynamic graph $G$ on $n$ vertices, guaranteed to belong to $\Cc$ at all times, maintains the value $|\Ff(H,G)|$ with amortized update time $\Oh_{\Cc,H}(\log^{\binom{h}{2}-1} n)$, where $h$ is the vertex count of $H$. The data structure can be initialized on an edgeless $G$ in $\Oh_{\Cc,H}(n)$ time and uses $\Oh_{\Cc,H}(n)$ space at all~times.	
\end{theorem}

As already remarked in \cite{DvorakT13}, \cref{thm:overview-DT} is quite robust and can be easily extended to the setting of vertex- and edge-colored graphs, where vertices and edges bear colors (that can be updated), and these colors have to be preserved under homomorphisms and (induced) subgraph isomorphisms. Note that by placing unique colors on some vertices, we can extend \cref{thm:overview-DT} to the setting of \emph{rooted mappings}: we can query for the number of homomorphisms/(induced) subgraph isomorphisms that map some fixed tuple $\tup u$ of vertices of $H$ to a given tuple $\tup v$ of vertices of $G$.

Coming back to the implementation of the query $\near{S,r}$, the following claim is straightforward and translates the query to the language of homomorphisms.

\begin{claim}\label{cl:homs}
	Let $r\in \N$, $G$ be a graph, and $S=\{s_1,\ldots,s_\ell\}$ be a set of vertices of $G$. Then the following conditions are equivalent:
	\begin{itemize}
		\item There is a vertex $v$ in $G$ such that $\dist(v,s_i)\leq r$ for all $i\in \{1,\ldots,\ell\}$.
		\item There is a function $\rho\colon \{1,\ldots,\ell\}\to \{0,1,\ldots,r\}$ such that there is a homomorphism $\varphi$ from the graph $H_\rho$ in \cref{fig:homs} to $G$ satisfying $\varphi(z_i)=s_i$ for all $i\in \{1,\ldots,\ell\}$.
	\end{itemize}
\end{claim}

\begin{figure}
	\centering
	\begin{tikzpicture}
		
		\node[vertex] (x) at (5,2) {};
		\node at (5,2.5) {$x$};
		
		\foreach \i/\j in {1/1,2/2,3/3,4/4,5/5,8/\ell-1,9/\ell} {
			\node[vertex] (z\i) at (\i,0) {};
			\node at (\i,-0.5) {$z_{\j}$};
			\draw[
			postaction={decorate, decoration={
					markings,
					mark=at position 0.25 with {\node[svertex] {};},
					mark=at position 0.5 with {\node[svertex] {};},
					mark=at position 0.75 with {\node[svertex] {};}
			}}
			] (x) -- (z\i);
 		}
 	
 		\draw[<->,dashed] (1-0.1,0+0.3) -- (5-0.22,2+0.22); 
		\node at (6.5,0) {$\dots$};
		\node at (3-0.5,1+0.5) {$\rho(1)$};
		
	\end{tikzpicture}
	\caption{Graph $H_\rho$ in \cref{cl:homs}. It is constructed from  vertex $x$ and vertices $z_1,\ldots,z_\ell$ by connecting $x$ with $z_i$ by a path of length $\rho(i)$, for each $i\in \{1,\ldots,\ell\}$.}\label{fig:homs}
\end{figure}
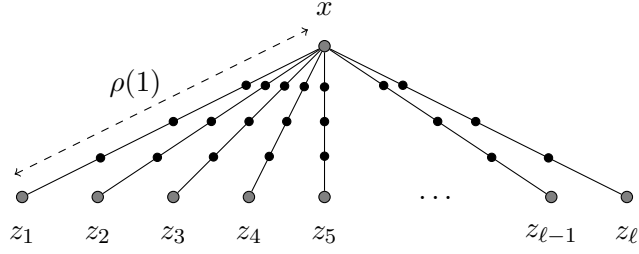

Therefore, in order to obtain a data structure for just the existential query $\near{S,r}$ --- whether a suitable vertex $v$ exists --- it suffices to maintain the data structure of \cref{thm:overview-DT} for every possible graph $H_\rho$ (of which there are $(r+1)^\ell$ many, where $\ell$ is an upper bound on $|S|$), where the vertices $z_1,\ldots,z_\ell$ are marked with unique colors. Upon query $\near{S,r}$, we enumerate the vertices of $S$ as $s_1,\ldots,s_\ell$, mark them using the colors of $z_1,\ldots,z_\ell$, and determine whether any of the maintained data structures indicates that there exists at least one homomorphism from some $H_\rho$ to $G$ that maps each $z_i$ to the respective $s_i$. Once this is determined, vertices $s_i$ can be unmarked.

However, this method only determines the \emph{existence} of a vertex $v$ that is close to all the vertices of $S$, and does not provide an \emph{example} of such a vertex. To implement the semi-ladder algorithm, we actually have to be able to retrieve a suitable vertex $v$, because it should serve as a witness for the next rounds. While we believe that this could be done for homomorphisms by a careful inspection of the arguments in~\cite{DvorakT13}, how to do this for (induced) subgraph isomorphisms is not clear at all, due to an application of the Inclusion--Exclusion Principle; see also the discussion in~\cite[Section~9]{DvorakT13}.

We now present a robust way to mitigate this caveat and turn a counting data structure into an example-reporting data structure with the help of randomization. While for the $\near{S,r}$ query this is probably not necessary, because the approach presented above relies only on homomorphisms, the method that we are going to introduce now will be reused in the implementation of $\far{S,r}$. More precisely, we will use an elegant fingerprint retrieval technique that was recently used in the context of data structures for parameterized problems~\cite{Majewski24,Nadara22}; let us describe it now.

The first observation is that \cref{thm:overview-DT} can be easily lifted to count \emph{weighted} mappings. Precisely, we assume that we additionally have a fixed weight function $\wei\colon V(H)\times V(G)\to \Z$, and the contribution of a mapping $\varphi$ to the count is $\val_\wei(\varphi)\coloneqq \prod_{u\in V(H)} \wei(u,\varphi(u))$, instead of just $1$. Precisely, instead of $|\Ff(H,G)|$, the data structure maintains the value
\[\val_\wei(\Ff(H,G))\coloneqq \sum_{\varphi\in \Ff(H,G)} \val_\wei(\varphi).\]
Suppose now that we have a fixed vertex $x$ of $H$ and we would like to find, if existent, any vertex $v$ of $G$ such that there exists $\varphi\in \Ff(H,G)$ with $v=\varphi(x)$. Suppose further, for a moment, that there is at most one such vertex $v$. By assuming without loss of generality that the vertex set of $G$ consists of numbers $\{1,\ldots,n\}$, we consider two weight functions $\wei,\wei'\colon V(H)\times V(G)\to \Z$ defined as follows:
\begin{align*}
	\wei(y,w) &\coloneqq 1\qquad \textrm{for all }(y,w)\in V(H)\times V(G);\\
	\wei'(y,w) &\coloneqq \begin{cases}w & \textrm{if }y=x,\\ 1 & \textrm{otherwise.}\end{cases}
\end{align*}
The key observation is the following: if indeed the sought vertex $v$ is unique, then it is equal to the ratio
\[
\frac{\val_{\wei'}(\Ff(H,G))}{\val_{\wei}(\Ff(H,G))}.
\]
Note that this ratio can be easily computed by maintaining the (weight-extended) data structures of Dvo\v{r}\'ak and T\r{u}ma for weight functions $\wei$ and $\wei'$. Observe also that if the denominator of this ratio --- $\val_\wei(\Ff(H,G))$ --- is equal to $0$, then this means that there is no $\Ff$-mapping from $H$ to $G$ and consequently no vertex $v$ with the required property.

The remaining question is how to lift the assumption that there is only at most one vertex $v$ that can be the image of $x$ in a mapping $\varphi\in \Ff(H,G)$. We do this with the help of randomization. Let us fix a desired bound $\eps>0$ on the error probability. Upon initialization of the data structure, for each $i\in \{0,1,\ldots,\lceil \log n\rceil\}$ we sample $\xi=\Theta(\log \tfrac{1}{\eps})$ vertex subsets $S_{i,1},\ldots,S_{i,\xi}\subseteq V(G)$, where each set $S_{i,j}$ is constructed by including every vertex independently with probability $\frac{1}{2^{i+1}}$, so that the expected size of $S_{i,j}$ is $\tfrac{n}{2^{i+1}}$. It is not hard to see that if the set $X\coloneqq \{v\mid \textrm{there is }\varphi\in \Ff(H,G)\textrm{ such that }\varphi(x)=v\}$ is non-empty, then with probability at least $1-\eps$ at least one of the sets $S_{i,j}$ will have intersection of size exactly $1$ with $X$. Therefore, we may apply the same trick as in the previous paragraph, but for every set  $S_{i,j}$ simultaneously, where in the data structures constructed for $S_{i,j}$ we nullify all the weights of the pairs $(y,w)$ with $y=x$ and $w\notin S_{i,j}$. Thus, with probability at least $1-\eps$ at least one of the data~structures points to an example vertex $v\in X$, whose feasibility can be verified using the data structure of~\cref{thm:overview-DT}.

The fingerprinting retrieval technique presented above can be applied iteratively to retrieve example images of the vertices of $H$ under a homomorphism or an (induced) subgraph isomorphism one by one, until a full example mapping is retrieved. In this way, we answer the open question of Dvo\v{r}\'ak and T\r{u}ma~\cite{DvorakT13} about extending their data structure with a possibility to report example (induced) copies of a fixed graph $H$ in a dynamic graph $G$ that always belongs to a fixed class of bounded expansion. Precisely, we prove the following result. (A \emph{$k$-colored graph} is one where each edge is colored using one of $k$ colors.)

\begin{restatable}{theorem}{reporting} \label{thm:dvorak-tuma-examples}
	Let $\eps>0$, $H$ be a fixed graph, $\G$ be a graph class of bounded expansion, $k\in \N$, and $\F$ be either $\Hom$, $\Sub$, or $\ISub$. 
	Let $G$ be a dynamic $k$-colored
	graph on $n$ vertices that belongs to $\G$ at all times.
	Then, there is a randomized data structure $\MapEx_{\F, H, k,\eps}[G]$
	which after every update is able to report that either $\F(H, G)$ is empty, or provide some $\phi \in \F(H, G)$. The amortized update time is $\Oh_{\G, H, k}(\log^{\Oh_H(1)} n \log{\frac{1}{\eps}})$, the initialization time is $\Oh_{\G, H, k}(n \log n \log {\tfrac{1}{\eps}})$, and the space complexity is $\Oh_{\G, H, k}(n \log n \log {\tfrac{1}{\eps}})$.  The data structure never provides false positives, but may fail to provide an example mapping $\phi$ with probability at most $\eps$, against an oblivious adversary.
\end{restatable}

\newcommand{\fr}{\mathsf{far}}

\subsection{Finding far vertices}

We now proceed to the description of the data structure supporting the query $\far{S,r}$. For simplicity, we shall focus on a simpler task: We would like to be able to report the quantity $|V_\fr|$, where
\[V_\fr\coloneqq \{v\in V(G)\mid \dist(v,s)>r\textrm{ for all }s\in S\}.\]
Once we achieve this, a data structure that can actually report an example vertex $v\in V_\fr$ can be obtained using the fingerprint retrieval technique described in the previous section, as follows:
\begin{itemize}
	\item First, we generalize the data structure so that it reports the quantity $\sum_{v\in V_\fr} \wei(v)$, for a weight function $\wei\colon V(G)\to \Z$ fixed upon initialization.
	\item Then, we apply the fingerprint retrieval technique so that by maintaining the data structures from the first point for $\Oh(\log n\log \tfrac{1}{\eps})$ different weight functions $\wei$ allows us to retrieve the index of an example vertex from $V_\fr$, provided $V_\fr\neq \emptyset$.
\end{itemize}

Furthermore, we will assume that in all the queries, $S$ consists of a single fixed vertex $s$, $S=\{s\}$. This can be achieved by adding a fresh isolated vertex $s$ to the graph that at the time of a query, is made adjacent to all the vertices of $S$. Then the original query $\far{S,r}$ is equivalent to the query $\far{\{s\},r+1}$ after the modification.

At this point we reach the crucial novel part of this work. Namely, counting $|V_\fr|$ cannot be directly formulated as counting homomorphisms, because roughly speaking, $V_\fr$ is defined by the \emph{non-existence} of a homomorphism. Instead, we will compute $|V_\fr|$ using the Inclusion-Exclusion Principle, by adding and subtracting vertices that are, in various ways, close to $s$. As there may be an unbounded number of short paths connecting any pair of vertices, it is a priori unclear how to apply the Inclusion-Exclusion Principle over them. We overcome this issue by exploiting the  properties of bounded expansion graph classes again.
Precisely, we use the same Sparsity tool as Dvo\v{r}\'ak and T\r{u}ma~\cite{DvorakT13}: \emph{fraternal augmentations}.



\paragraph*{Fraternal augmentations.}
Let $G$ be a graph belonging to a fixed bounded expansion class $\Cc$. The assumption that~$G\in \Cc$ implies that for some constant $d_1=\Oh_\Cc(1)$, $G$ is \emph{$d_1$-degenerate}: every subgraph of $G$ has a vertex of degree at most $d_1$. Hence, by iteratively deleting vertices of the smallest degree from $G$ and orienting incident edges outwards, we obtain an orientation $\vec G_1$ with maximum outdegree at most~$d_1$.

Now, we add \emph{fraternal edges} to $\vec G_1$: for every pair of (directed) edges $(u,v)$ and $(u,v')$ with a common tail $u$, we add an (undirected) edge $vv'$, unless it was already present. It turns out that the new graph $G_2$ --- consisting of the oriented edges of $G$ and of the added fraternal edges --- is still sparse: the underlying undirected graph belongs to some graph class $\Cc_2$, depending only on $\Cc$, that still has bounded expansion. Consequently, the new fraternal edges can be oriented so that we obtain an oriented supergraph $\vec G_2$ of $\vec G_1$ whose maximum outdegree is at most $d_1+d_2$, for some constant $d_2$ depending only on $\Cc_2$. By performing this operation $r$ times, we eventually obtain an orientation $\vec G_r$ of a supergraph of $G$ whose maximum outdegree is at most $d\coloneqq d_1+d_2+\ldots+d_r$. 

With every edge of $\vec G_r$ we may naturally associate its \emph{length}: the original edges of $G$ have length $1$, and a fraternal edge added to edges of lengths $a$ and $b$ is assigned length $a+b$. Edges of length larger than $r$ will have no significance for us, so we may just not create them in the process.

An important idea in the work of Dvo\v{r}\'ak and T\r{u}ma~\cite{DvorakT13} is that the sequence of fraternal augmentations described above can be maintained in the dynamic setting with amortized update time $\Oh_{\Cc,r}(\log^r n)$. The key component here is the classic data structure of Brodal and Fagerberg~\cite{BrodalF99}, which can be used to maintain the orientation at each level of the construction. In this ``tower'' of Brodal--Fagerberg data structures, every update to $G$ triggers $\Oh_{\Cc}(\log n)$ updates to $\vec G_1$ in the amortized sense, which in turn triggers $\Oh_{\Cc}(\log^2 n)$ updates to $\vec G_2$, and so on. 

So from now on we assume that we maintain the augmentation $\vec G_r$ in our data structure, where every edge is decorated with its length. Moreover, since the outdegree of every vertex is at most $d$, we may additionally maintain a labelling of the edges of $\vec G_r$ with labels $\{1,\ldots,d\}$ so that the edges with the same tail are assigned pairwise different labels.

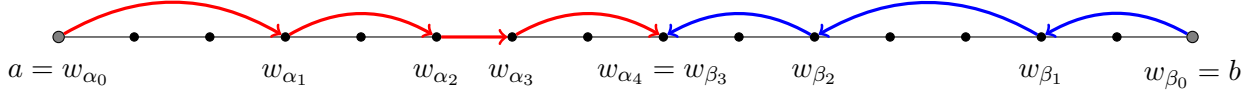
\begin{figure}
	\centering
	\begin{tikzpicture}
		
		\node[vertex] (z0) at (0,0) {};
		\node[vertex] (z15) at (15,0) {};
		\draw (z0) -- (z15);
		
		\foreach \i in {1,2,3,4,5,6,7,8,9,10,11,12,13,14} {
			\node[svertex] (z\i) at (\i,0) {};
		}
		
		\draw[->,very thick, red] (z0) edge[bend left] (z3);
		\draw[->,very thick, red] (z3) edge[bend left] (z5);
		
		\draw[->, very thick, red] (z5) edge (z6);
		\draw[->, very thick, red] (z6) edge[bend left] (z8);
		
		\draw[->, very thick, blue] (z15) edge[bend right] (z13);
		\draw[->, very thick, blue] (z13) edge[bend right] (z10);
		\draw[->, very thick, blue] (z10) edge[bend right] (z8);
		
		\node at (0,-0.5) {$a=w_{\alpha_0}$};
		\node at (3,-0.5) {$w_{\alpha_1}$};
		\node at (5,-0.5) {$w_{\alpha_2}$};
		\node at (6,-0.5) {$w_{\alpha_3}$};
		\node at (8,-0.5) {$w_{\alpha_4}=w_{\beta_3}$};
		\node at (10,-0.5) {$w_{\beta_2}$};
		\node at (13,-0.5) {$w_{\beta_1}$};
		\node at (15,-0.5) {$w_{\beta_0}=b$};
		
	\end{tikzpicture}
	\caption{An $a$-$b$-path with a shortcut. The directed paths $Q,R$ of the shortcut are resp. red and blue.}\label{fig:shortcut}
\end{figure}

\paragraph*{Shortcuts.}
Suppose now that in $G$ we have two vertices $a$ and $b$ that can be connected by a path $P=(a=w_0,w_1,\ldots,w_{r'}=b)$ of some length $r'\leq r$. If we trace what happens with $P$ during the consecutive augmentations, then it is not hard to see that in $\vec G_r$, $P$ will have an oriented ``shortcut'' consisting of:
\begin{itemize}
	\item a directed path $Q$ of the form $w_{\alpha_0}\to w_{\alpha_1}\to \ldots \to w_{\alpha_p}$, where $0=\alpha_1<\alpha_2<\ldots<\alpha_p$, so that each edge $(w_{\alpha_j},w_{\alpha_{j+1}})$ has length at most $\alpha_{j+1}-\alpha_j$; and
	\item a directed path $R$ of the form $w_{\beta_0}\to w_{\beta_1}\to \ldots \to w_{\beta_q}$, where $r'=\beta_1>\beta_2>\ldots>\beta_q=\alpha_p$, so that each edge $(w_{\beta_j},w_{\beta_{j+1}})$ has length at most $\beta_j-\beta_{j+1}$.
\end{itemize}
See \cref{fig:shortcut}. If we now record the lengths and the labels of the consecutive edges of the shortcut paths $Q$ and $R$, we obtain two sequences of pairs from $\{1,\ldots,d\}\times \{1,\ldots,r\}$ such that the sum of all the second coordinates is at most $r$; we call such a pair of sequences the \emph{pattern} of a shortcut. At this point, two observations are crucial:
\begin{enumerate}[label=(O\arabic*)]
	\item\label{o:numpat} There are only at most $(dr)^{2r}=\Oh_{\Cc,r}(1)$ different patterns. Let $\Pi$ be their set.
	\item\label{o:unique} For every pair of vertices $s,u\in V(G)$ and pattern $\pi=(\pi_Q,\pi_R)$, there is at most one shortcut of an $a$-$b$-path with pattern $\pi$. This is because in $\vec G_r$, the edges outgoing from a single vertex have pairwise different labels, so knowing $a$, $b$, and sequences $\pi_Q,\pi_R$, we can uniquely reconstruct the paths $Q$ and $R$. A shortcut is present if the reconstructed paths $Q$ and $R$ end at the same vertex.
\end{enumerate}
Based on the above discussion, the following is now clear.

\begin{claim}\label{cl:overview-shortcuts}
	For any two vertices $a,b\in V(G)$, the following conditions are equivalent:
	\begin{itemize}
		\item $\dist(a,b)\leq r$.
		\item In $\vec G^r$, there exists a shortcut between $a$ and $b$ with some pattern $\pi\in \Pi$.
	\end{itemize}
\end{claim}

We remark that the concept of path shortcuts in the context of fraternal augmentations that we describe above actually dates back to the 2006 work of Kowalik and Kurowski~\cite{KowalikK06}.

\paragraph*{Counting far vertices.} We may now come back to our initial goal of maintaining the cardinality of the set $V_\fr$, which consists of vertices $u$ with $\dist(s,u)>r$. For a vertex $u$ and a pattern $\pi\in \Pi$, call $\pi$ \emph{realized} at $u$ if in $\vec G_r$ there is a shortcut between $s$ and $u$ with pattern $\pi$. By \cref{cl:overview-shortcuts}, $V_\fr$ consists of those vertices $u$ at which no pattern of $\Pi$ is realized. Now, if for $\Gamma\subseteq \Pi$ we define $V_\Gamma$ as the set of those $u\in V(G)$ at which all the patterns of $\Gamma$ are realized, then by the Inclusion-Exclusion principle, we have
\[|V_\fr|=\sum_{\Gamma\subseteq \Pi} (-1)^{|\Gamma|}\cdot |V_\Gamma|.\]
(Note here that $V_\emptyset=V(G)$.) Since $\Pi$ is of size $\Oh_{\Cc,r}(1)$ by \ref{o:numpat}, to compute $|V_\fr|$ it suffices to compute $|V_\Gamma|$ for each $\Gamma\subseteq \Pi$. However, it is easy to argue using \ref{o:unique} that $|V_\Gamma|$ is equal to the number of homomorphisms $\varphi$ from the oriented graph $H_\Gamma$ depicted in \cref{fig:HGamma} to $\vec G_r$ that satisfy $\phi(z)=s$. (Here, we consider both $H_\Gamma$ and $\vec G_r$ as directed graphs with edges decorated with lengths and labels, and the considered homomorphisms have to respect those decorations.) It now remains to note that this homomorphism count can be maintained using the data structure of \cref{thm:overview-DT}. So all in all, we maintain the data structures of \cref{thm:overview-DT} for all graphs $H_\Gamma$ for $\Gamma\subseteq \Pi$, and from the maintained homomorphism counts we piece together the value of $|V_\fr|$ using Inclusion-Exclusion. This concludes the description of the data structure supporting query $\far{S,r}$ in classes of bounded expansion.

\subsection{Other results}\label{sec:overview-other}

Finally, let us comment on the two side results, \cref{thm:main-domset-deg,thm:apx-dynamic}.

\medskip

The proof of \cref{thm:main-domset-deg} follows the same path as that of \cref{thm:main-domset-be}, except that for $r=1$, queries $\near{S,1}$ and $\far{S,1}$ can be implemented in a much simpler way, and relying only on the boundedness of degeneracy. Let us consider the query $\near{S,1}$, which boils down to finding a vertex $v$ that is simultaneously adjacent to all the vertices of $S$. Since the maintained graph $G$ is $d$-degenerate, we may use the data structure of Brodal and Fagerberg~\cite{BrodalF99} to maintain its orientation $\vec G$ with maximum outdegree $\Oh(d)$. Now, when searching for a suitable vertex $v$, we may first test every vertex in the outneighborhood of $S$ in $\vec G$; this amounts to testing $\Oh(d|S|)$ vertices, each in time $\Oh(d|S|)$. If we do not find a suitable $v$ in this way, we know that we are looking for $v$ such that all the edges between $v$ and $S$ are oriented away from $v$ in $\vec G$. The idea now is that together with $\vec G$, we can efficiently maintain the hypergraph of the outneighborhoods of all the vertices (which are of size $\Oh(d)$), together with all their subsets (which amounts to $2^{\Oh(d)}\cdot n$ subsets in total). Then verifying whether there exists a vertex $v$ with $S$ contained in its neighborhood amounts to checking whether $S$ belongs to the maintained hypergraph.

The argument presented above easily extends to reporting a suitable vertex $v$ adjacent to all the vertices of $S$, as well as counting the number of such vertices $v$. Similarly, for every subset $S'\subseteq S$, we may count the number of vertices $v$ that are adjacent to all the vertices of $S'$. Thus, using the Inclusion-Exclusion Principle, we may also count the number of vertices $v$ that are adjacent to none of the vertices in $S$. This can be easily extended to weighted counting (under any weight function $\wei\colon V(G)\to \Z$ fixed upon initialization), so we may again use the fingerprint retrieval technique to design a randomized data structure that can report a vertex non-adjacent to all the vertices of $S$, provided there exists one. This gives an implementation of $\far{S,1}$ and completes the proof of \cref{thm:main-domset-deg}.

\medskip

Finally, the proof of \cref{thm:apx-dynamic} mostly relies on completely different ideas. Similarly as above, using the data structure of Brodal and Fagerberg~\cite{BrodalF99} we may maintain an orientation $\vec G$ of the maintained graph $G$ with maximum outdegree $\Oh(d)$. The key idea is that if we consider the set system of closed outneighborhoods $\Ss\coloneqq \{\{v\}\cup N^+_{\vec G}(v)\colon v\in V(G)\}$, then for any inclusion-wise maximal packing of disjoint sets ${\cal M}\subseteq \Ss$, $\bigcup {\cal M}$ is an $\Oh(d^2)$-approximation of the minimum dominating set. This observation is in essence already present in an old distributed $\Oh(d^2)$-approximation algorithm for \textsc{Dominating Set} on degenerate graphs, due to Lenzen and Wattenhofer~\cite{LenzenW10}. It is not hard to see that the set system $\Ss$ can be maintained efficiently under updates to $\vec G$, hence we can use the recent data structure of Assadi and Solomon~\cite{AssadiSolomon21} to efficiently maintain an inclusion-wise maximal packing in $\Ss$ as well.

In our argumentation we take an extra mile to clarify that the combinatorial argument behind the approximation guarantee is in fact more general, as it extends to larger $r$. Namely, if we assume that $G$ belongs to a fixed class of bounded expansion $\Cc$, then there exists a vertex ordering $\sigma$ of $G$ whose \emph{weak $r$-coloring number} $\wcol_r(G,\sigma)$ is bounded by a constant $c$ depending only on $\Cc$ and $r$. This means that all the \emph{weak $r$-reachability sets} $\WReach_r[G,\sigma,v]$, for $v\in V(G)$, are of size bounded by $c$; these are analogues of the closed outneighborhoods in the $r=1$ case. We observe that again, the union of any inclusion-wise maximal packing in the set system $\Ff_r(G,\sigma)\coloneqq \{\WReach_r[G,\sigma,v]\colon v\in V(G)\}$ is a $c^2$-approximation of the minimum distance-$r$ dominating set in $G$. Unfortunately, for $r>1$ we do not know how to efficiently maintain $\Ff_r(G,\sigma)$ under updates to $G$, as a single update to $G$ may necessitate an unbounded number of changes to $\Ff_r(G,\sigma)$.
\begin{figure}
	\centering
	\begin{tikzpicture}
		
		\node[vertex] (z) at (-4,0) {};
		\node at (-4.5,0) {$z$};
		
		\node[vertex] (x) at (4,0) {};
		\node at (4.5,0) {$x$};
		
		\foreach \i in {1,2,3,4,5} {
			\node[vertex] (u\i) at (0,3-\i*1) {};
		}
		
		\foreach \t in {1,2,3} {
			\node[svertex] (z1\t) at (-\t,2-0.5*\t) {};
		}
		\draw[very thick, red,->] (z) -- (z13);
		\draw[very thick, red,->] (z13) -- (z12);
		\draw[very thick, red,->] (z12) -- (z11);
		\draw[very thick, red,->] (z11) -- (u1);

		\foreach \t in {1,2,3} {
			\node[svertex] (z4\t) at (-\t,-1+0.25*\t) {};
		}
		\draw[very thick, red,->] (z) -- (z43);
		\draw[very thick, red,->] (z43) -- (z42);
		\draw[very thick, red,->] (z42) -- (z41);
		\draw[very thick, red,->] (z41) -- (u4);
		
		\foreach \t in {2} {
			\node[svertex] (z2\t) at (-\t,1-0.25*\t) {};
		}
		\draw[very thick, red,->] (z) -- (z22);
		\draw[very thick, red,->] (z22) -- (u2);
		
		\foreach \t in {2} {
			\node[svertex] (z5\t) at (-\t,-2+0.5*\t) {};
		}
		\draw[very thick, red,->] (z) -- (z52);
		\draw[very thick, red,->] (z52) -- (u5);

		\foreach \t in {1,2} {
			\node[svertex] (z3\t) at (-\t*1.33,0) {};
		}
		\draw[very thick, red,->] (z) -- (z32);
		\draw[very thick, red,->] (z32) -- (z31);
		\draw[very thick, red,->] (z31) -- (u3);

		\foreach \t in {1,2} {
			\node[svertex] (x1\t) at (\t*1.33,2-0.66*\t) {};
		}
		\draw[very thick, blue,->] (x) -- (x12);
		\draw[very thick, blue,->] (x12) -- (x11);
		\draw[very thick, blue,->] (x11) -- (u1);

		\foreach \t in {1,2} {
			\node[svertex] (x2\t) at (\t*1.33,1-0.33*\t) {};
		}
		\draw[very thick, blue,->] (x) -- (x22);
		\draw[very thick, blue,->] (x22) -- (x21);
		\draw[very thick, blue,->] (x21) -- (u2);

		\foreach \t in {1,2,3} {
			\node[svertex] (x3\t) at (\t,0) {};
		}
		\draw[very thick, blue,->] (x) -- (x33);
		\draw[very thick, blue,->] (x33) -- (x32);
		\draw[very thick, blue,->] (x32) -- (x31);
		\draw[very thick, blue,->] (x31) -- (u3);

		\foreach \t in {1,2,3} {
			\node[svertex] (x5\t) at (\t,-2+\t*0.5) {};
		}
		\draw[very thick, blue,->] (x) -- (x53);
		\draw[very thick, blue,->] (x53) -- (x52);
		\draw[very thick, blue,->] (x52) -- (x51);
		\draw[very thick, blue,->] (x51) -- (u5);
		
		\foreach \t in {2} {
			\node[svertex] (x4\t) at (\t,-1+\t*0.25) {};
		}
		\draw[very thick, blue,->] (x) -- (x42);
		\draw[very thick, blue,->] (x42) -- (u4);

	\end{tikzpicture}
	\caption{Example graph $H_\Gamma$ constructed for a set $\Gamma\subseteq \Pi$ of size $5$. Every pair (red path, blue path) represents a shortcut with some pattern $\pi\in \Gamma$. Note that both the red edges and the blue edges are decorated with $\{1,\ldots,d\}\times \{1,\ldots,r\}$ (not depicted).}\label{fig:HGamma}
\end{figure}
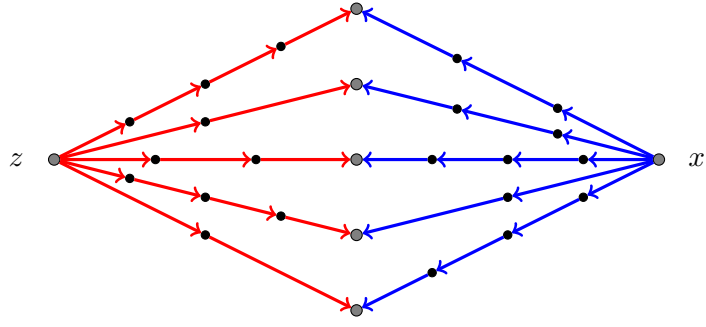
\section{Preliminaries}\label{sec:prelims}

We write $\N$ for the set of nonnegative integers.
For a positive integer $k$, we denote $[k]\coloneqq \{1,\ldots,k\}$ and $[0,k]\coloneqq \{0,1,\ldots,k\}$.
We follow the convention that if $\tup x$ is a tuple of objects, then the $i$th element of $\tup x$ is denoted by $x_i$. For a tuple of parameters $\tup p$, the $\Oh_{\tup p}(\cdot)$ notation hides multiplicative factors that may depend on $\tup p$.

\subsection{Graphs}

\paragraph{Basics.} We use standard graph terminology and notation. All graphs considered in this paper are finite, undirected, and simple (without loops or parallel edges), unless explicitly stated.

For a graph $G$, by $V(G)$ and $E(G)$ we denote the vertex set and the edge set of $G$, respectively. We also write $|G| \coloneqq |V(G)|$ and $\|G\| \coloneqq |E(G)|$. For vertices $u,v$ of $G$, by $\dist_G(u,v)$ we denote the distance between $u$ and $v$, defined as the smallest possible length of a path in $G$ with endpoints $u$ and $v$. This notation is extended to subsets naturally, e.g. $\dist_G(v,S)=\min_{s\in S}\dist_G(v,s)$. The \emph{closed neighborhood} of a vertex $u$ in $G$ is the set $N_G[u]$ that consists of $u$ and all the neighbors of $u$. By $\avgdeg(G)$ we denote the average degree in $G$; note that it is equal to $\frac{2\|G\|}{|G|}$.

A \emph{rooted graph} is a graph $G$ together with a tuple $\tup u$ of vertices of $G$ (not necessarily distinct), called the \emph{roots}. We denote it by $\rt{G}{\tup u}$. 



We use the standard notion of directed graphs, again disallowing loops and parallel edges (edges with same head and tail). An \emph{oriented graph} is a directed graph $D$ where any pair of vertices can be the endpoints of at most one edge; that is, for distinct vertices $u,v$ we disallow that the edges $(u,v)$ and $(v,u)$ are simultaneously present. For an undirected graph $G$, an \emph{orientation} of $G$ is an oriented graph obtained from $G$ by choosing an orientation of every edge of $G$. Moreover, if an orientation of $G$ satisfies that the maximum outdegree is at most $d$, we will call it a \emph{$d$-orientation}. 

A directed graph is \textit{connected} if its underlying undirected graph is connected.

A graph $G$ is \emph{$d$-degenerate} if every subgraph of $G$ contains a vertex of degree at most $d$. The \emph{degeneracy} of $G$ is the smallest $d$ for which $G$ is $d$-degenerate. Note that if $G$ is $d$-degenerate, then it has a $d$-orientation. Indeed, such an orientation can be obtained by iteratively removing from $G$ a vertex of the smallest degree (which is always at most~$d$) and orienting the edges incident to it away from it.

The \emph{arboricity} of a graph $G$ is the smallest $\alpha\in \N$ such that the edge of $G$ can be partitioned into $\alpha$ sets each of which forms a forest. It is well-known that if $d$ and $\alpha$ are the degeneracy and the arboricity of~$G$, respectively, then we have
\[\alpha\leq d\leq 2\alpha.\]
In other words, the degeneracy and the arboricity of a graph are within a multiplicative factor of $2$, so these are essentially equivalent parameters from the point of view of our purposes. In this paper we choose to use degeneracy as the base parameter to state our results.




\paragraph*{Mappings between graphs.} For graphs $H$ and $G$, a \emph{homomorphism} from $H$ to $G$ is a map $\varphi\colon V(H)\to V(G)$ such that for every pair of distinct vertices $u,v\in V(H)$, if $uv$ is an edge in $H$ then $\varphi(u)\varphi(v)$ is an edge in $G$. A \emph{subgraph isomorphism} is a homomorphism that is injective (i.e. $\varphi(u)\neq \varphi(v)$ for all distinct $u,v\in V(H)$), and an \emph{induced subgraph isomorphism} is a subgraph isomorphism where the implication stated above is in fact an equivalence: $uv$ is an edge in $H$ if and only if $\varphi(u)\varphi(v)$ is an edge in $G$. By $\Hom(H,G)$, $\Sub(H,G)$, and $\ISub(H,G)$ we denote the sets of all homomorphisms, subgraph isomorphisms, and induced subgraph isomorphisms from $H$ to $G$, respectively.

The definitions of homomorphisms and (induced) subgraphs isomorphisms can be naturally extended to rooted graphs, by requiring that the source graph $\rt{H}{\tup u}$ and the target graph $\rt{G}{\tup v}$ have the same number of roots ($|\tup u|=|\tup v|$) and that the $i$th root of the source graph is mapped to the $i$th root of the target graph ($\varphi(u_i)=v_i$, for all $i\in [|\tup u|]$). In particular, if $\rt{H}{\tup u}$ and $\rt{G}{\tup v}$ are rooted graphs with $|\tup u|=|\tup v|$, then by $\Hom(\rt{H}{\tup u},\rt{G}{\tup v})$ we denote the set of all homomorphisms from $\rt{H}{\tup u}$ to $\rt{G}{\tup v}$; and similarly for (induced) subgraph isomorphisms.



\paragraph*{Relational structures.} We use the standard notion of relational structures considered in finite model theory. A \emph{signature} is a finite set $\Sigma$ of relation names, with each relation name $R\in \Sigma$ having a prescribed \emph{arity} $\ar(R)\in \N$. A \emph{$\Sigma$-structure} $\Af$ consists of a finite universe $U=U(\Af)$ and, for every relation name $R\in \Sigma$, its \emph{interpretation} $R^\Af\subseteq U^{\ar(R)}$. 

Here are some examples of modeling (variations of) graphs as relational structures, which will be used in this paper.
\begin{itemize}
	\item An undirected graph $G$ is modelled as a $\Sigma$-structure whose universe is the vertex set $U=V(G)$, and $\Sigma=\{ R \}$ consists of a single binary (arity-$2$) relation $R$ that is always irreflexive and symmetric. We let $(u,v),(v,u) \in R$ iff $\{u,v\} \in E(G)$.  To model directed graphs, we drop the symmetricity requirement.
	\item For a finite set of colors $C$, a \emph{$C$-colored graph} is an undirected graph $G$ together with a function $\mathrm{col}\colon E(G)\to C$ assigning every edge its color. We model $C$-colored graph $G$ by letting the universe be the vertex set, $U=V(G)$, and making $\Sigma$ consist of $|C|$ relational symbols, one for each color, $\Sigma = \{ R_c  \colon c \in C\}$.
	 All these relations are binary, and each relation $R_c$ for $c \in C$ contains tuples $(u,v)$ and $(v,u)$ for each edge $\{ u,v \} \in E(G)$ with $\mathrm{col}(\{ u,v \})=c$.
	 We may similarly define and model $C$-colored directed graphs.
	 For $k\in \N$, we use a shorthand: a \emph{$k$-colored graph} is a $[k]$-colored graph, that is, a graph whose edges are colored with colors $\{1,\ldots,k\}$.
\end{itemize}

The \emph{Gaifman graph} of a $\Sigma$-structure $\Af$ is the graph on vertex set $U(\Af)$ where two distinct elements $u,v\in U(\Af)$ are adjacent if and only if they appear together in some tuple of some relation $R^\Af$, $R\in \Sigma$.

Homomorphisms of relational structures are defined naturally: for two $\Sigma$-structures $\Af$ and $\Bf$, a homomorphism from $\Af$ to $\Bf$ is a mapping $\varphi\colon U(\Af)\to U(\Bf)$ such that for every $R\in \Sigma$ and $\tup u\in U(\Af)^{\ar(R)}$, it holds that $\tup u\in R^{\Af}$ implies $\varphi(\tup u)\in R^{\Bf}$ (where $\varphi(\tup u)$ denotes coordinate-wise application). The notions of (induced) substructure isomorphisms then follow in the same way as for graphs, and so we may extend the $\Hom(\cdot,\cdot),\Sub(\cdot,\cdot),\ISub(\cdot,\cdot)$ notation to relational structures as well.

Similarly to graphs, we may consider \emph{rooted relational structures} by equipping a structure $\Af$ with a tuple $\tup u$ of roots, which are (not necessarily distinct) elements of the universe; we denote it by $\rt{\Af}{\tup u}$. Homomorphisms and (induced) substructure isomorphisms can be again extended to rooted structures in the expected way. Thus, we may speak about homomorphisms of rooted, colored directed graphs etc.

\paragraph*{Sparsity.} For $r\in \N$, we say that a \emph{depth-$r$ model} of a graph $H$ in a graph $G$ is a mapping $\eta$ from the vertices of $H$ to subsets of the vertices of $G$ satisfying the following two conditions:
\begin{itemize}
	\item The sets $\{\eta(u)\colon V(H)\}$ are pairwise disjoint and each of them induces a connected subgraph of $G$ of radius at most $r$.
	\item Whenever $uv$ is an edge in $H$, in $G$ there is an edge with one endpoint in $\eta(u)$ and the other in $\eta(v)$.
\end{itemize}
The sets $\{\eta(u)\colon u\in V(H)\}$ are called \emph{branch sets} of the model. We say that $G$ contains $H$ as a \emph{depth-$r$ minor} if $G$ contains a depth-$r$ model of $H$.

For a graph $G$ we define\footnote{Classic literature, e.g.~\cite{sparsity}, use the ratio $\frac{\|H\|}{|H|}$ instead of $\avgdeg(H)$, which is exactly twice smaller. This difference is immaterial for our results, as we are interested only in the boundedness of the $\nabla_r(\cdot)$ parameters.}
\[\nabla_r(G)\coloneqq \max\{\,\avgdeg(H)~\colon~H\textrm{ is a depth-}r\textrm{ minor of }G\,\}.\]
For a graph class $\Cc$, we define
\[\nabla_r(\Cc)\coloneqq \sup_{G\in \Cc} \nabla_r(G).\]
Note that $\nabla_r(\Cc)$ may be equal to $+\infty$; this happens if graphs $G\in \Cc$ attain arbitrarily large values of $\nabla_r(G)$. This finite/infinite distinction underlies the central definition considered in this work.


\begin{definition}
	A class of graphs $\Cc$ has \emph{bounded expansion} if $\nabla_r(\Cc)$ is finite for every $r\in \N$. 
\end{definition}

Unpacking the definitions, this means that there exists a function $f\colon \N\to \N$ such that for every $G\in \Cc$ and a depth-$r$ minor $H$ of $G$, we have $\avgdeg(H)\leq f(r)$. Note that depth-$0$ minors are just subgraphs, hence every graph $G$ belonging to a class of bounded expansion $\Cc$ is $\nabla_0(\Cc)$-degenerate; this is a constant depending only on $\Cc$.

In our proofs, we will use stability of the notion of bounded expansion under two basic operations. The first is adding pendants.
For a graph $G$ and a vertex $v$ of $G$, the operation of \emph{adding a pendant} to $v$ produces a new graph $G'$ obtained from $G$ by adding a fresh vertex $v'$ and an edge $vv'$.
	
\begin{lemma} \label{obs:pendants-bnd-exp}
	Let $\Cc$ be a class of graphs of bounded expansion and let $\Cc'$ be the class of all the graphs that can be obtained from a graph from $\Cc$ by repeatedly adding pendants (an arbitrary number of times). Then $\Cc'$ is also of bounded expansion.
\end{lemma}
\begin{proof}
	Fix $r\in \N$.
	We prove that if $G'$ is obtained from $G$ by adding a pendant $v'$ to $v\in V(G)$, then
	\begin{equation}\label{eq:zubr}
		\nabla_r(G')\leq \max\left(1,\nabla_r(G)\right).
	\end{equation}
	By applying \eqref{eq:zubr} repeatedly to every graph in $G'\in \Cc'$, we conclude that
	\[\nabla_r(\Cc')\leq \max\left(1,\nabla_r(\Cc)\right),\]
	so $\Cc$ having bounded expansion implies that $\Cc'$ has bounded expansion as well.
	
	Towards \eqref{eq:zubr}, let $H'$ be a depth-$r$ minor of some $G'$, and let $\eta'$ be the witnessing model. We consider two cases. First, if $\eta'$ contains no branch set equal to $\{v'\}$, then it is easy to see that $\eta$ obtained by removing $v'$ from the branch set it is contained in (if any) is a depth-$r$ model of $H'$ in $G$. Hence $\avgdeg(H')\leq \nabla_r(G)$. Second, if there exists $x\in V(H')$ such that $\eta'(x)=\{v'\}$, then $x$ must have degree at most $1$ in $H'$ and $\eta$ obtained from $\eta'$ by removing the branch set $\{v'\}$ is a depth-$r$ model of $H\coloneqq H'-x$ in $G$. It is easy to see that adding a pendant to a graph does not increase its average degree unless it stays below $1$, hence $\avgdeg(H')\leq \max(1,\avgdeg(H))\leq \max(1,\nabla_r(\Cc))$. This proves \eqref{eq:zubr} 
\end{proof}

The next operation is taking congested shallow minors. For $c,r\in \N$, a \emph{congestion-$c$ depth-$r$ model} $\eta$ of a graph $H$ in a graph $G$ is defined similarly to a (standard) depth-$r$ model, except for the following amendments:
\begin{itemize}
	\item In the first condition, we allow the branch sets $\{\eta(u)\colon u\in V(H)\}$ to intersect, but every vertex $v$ of $G$ may belong to at most $c$ of them.
	\item In the second condition, we demand that whenever $uv$ is an edge of $H$, the branch sets $\eta(u)$ and $\eta(v)$ intersect or there is an edge connecting them in $G$.
\end{itemize}
Note that for $c=1$, these are just standard depth-$r$ models.

For a class of graphs $\Cc$, by $\Minors^{c,r}(\Cc)$ we denote the class of all congestion-$c$ depth-$r$ minors of graphs from $\Cc$. We will use the following standard statement.

\begin{theorem}[see e.g.~{\cite[Proposition 4.6]{sparsity}} or {\cite[Chapter 1, Lemma 2.27]{sparsityNotes}}]\label{thm:cong-minors}
	For every graph class $\Cc$ of bounded expansion and $c,r\in \N$, the class $\Minors^{c,r}(\Cc)$ also has bounded expansion.
\end{theorem}

We use \cref{thm:cong-minors} to argue a classic result from Sparsity 
about classes of bounded expansion admitting \emph{fraternal augmentations} of bounded maximum outdegree. Let $G$ be a graph and $\vec G$ be an orientation of $G$. The \emph{fraternal augmentation} of $\vec G$, denoted as $\fra(\vec G)$, is the undirected graph with vertex set $V(\fra(\vec G))=V(\vec G)=V(G)$, obtained from $\vec G$ as follows:
\begin{itemize}
    \item $E(\fra(\vec G))$ contains all edges $E(G)$, plus
	\item for every pair of edges $(w, u),(w, v)\in E(\vec G)$ with a common tail $w$ (such pairs will be said to form a \emph{fork}), add an undirected edge $uv$, if not already present.
\end{itemize}
Note that thus, $\fra(\vec G)$ is a supergraph of $G$, obtained by adding all the \emph{fraternal} edges $uv$ as above.
Next, for an undirected graph $G$ and $d \in \N$ we define $\Fra(G,d)$
to be the graph class constructed as follows: for every $d$-orientation $\vec G$ of $G$, include $\fra(\vec G)$ and all its subgraphs in $\Fra(G,d)$.
Finally, for a class $\Cc$ of graphs, we let $\Fra(\Cc,d)=\bigcup_{G \in \Cc} \Fra(G,d)$.
The following lemma is the key observation about the behavior of fraternal augmentations on classes of bounded expansion.
\begin{lemma}\label{lem:fraternal}
	Let $\Cc$ be a graph class of bounded expansion and $d\in \N$. Then $\Fra(\Cc,d)$ also has bounded expansion.
\end{lemma}
\begin{proof}
	We observe that if $G'=\fra(\vec G)$ is the fraternal augmentation of an orientation $\vec G$ of $G$ of maximum outdegree~$d$, then $G'$ is a congestion-$(d+1)$ depth-$1$ minor of $G$. Indeed, if for $u\in V(G)$ we define
	\[\eta(u)\coloneqq \{u\}\cup \{v~\mid~(u, v)\in E(\vec G)\},\]
	then $\eta$ is a congestion-$(d+1)$ depth-$1$ model of $G'$ in $G$. It follows that $\Fra(\Cc,d)\subseteq \Minors^{d+1,1}(\Cc)$, so $\Fra(\Cc,d)$ has bounded expansion by \cref{thm:cong-minors}.
\end{proof}

\subsection{Computation}

\paragraph*{Computation model.} When working with a graph on $n$ vertices in the algorithmic context, we use the standard word RAM model with words of length $\Oh(\log n)$. In particular, space complexity of data structures is measured in the number of words occupied. We remark that all the numbers appearing in the computation, e.g. vertex weights or homomorphism counts, will be always bounded polynomially in $n$, and hence they fit within a constant number of words and arithmetic operations on them can be executed in constant time.

\paragraph*{Dynamic graphs and data structures.} Let us introduce some terminology to facilitate speaking about fully dynamic data structures for graphs that are updated over time.

By a \emph{dynamic graph} we mean a graph $G$ whose vertex set stays fixed, but whose edge set is modified over time by updates of the following types:
\begin{itemize}
	\item \emph{insert} an edge $uv$, if not already present; and
	\item \emph{remove} the edge $uv$, if present.
\end{itemize}
By a \emph{data structure} for a dynamic graph $G$ we mean a data structure that maintains $G$ and supports a superset of the following basic methods:
\begin{itemize}
	\item $\init{n}$: Initialize the data structure on an edgeless graph $G$ on $n$ vertices.
	\item $\add{uv}$: Insert the edge $uv$.
	\item $\remove{uv}$: Remove the edge $uv$.
\end{itemize}
These methods will be typically augmented by some further queries, specified in the description of the data structure. By the \emph{update time} of a data structure we mean the time complexity of methods $\add{uv}$ and $\remove{uv}$, and by the \emph{initialization time} we mean the time complexity of method $\init{n}$. While for simplicity we assume that the initialization works for an edgeless graph, note that one can initialize any graph by first initializing an edgeless graph and then adding the edges one by one.

In this work, we deal with classes of graphs of bounded expansion, which are by definition $d$-degenerate for some constant $d$. Hence, in all our data structures, we use the data structure of Brodal and Fagerberg (see \cref{thm:bf}) to answer adjacency queries. This introduces an additional factor of $d$ in the query times of our data structures, which is submerged within the asymptotic notation.

Some of our data structures are randomized \emph{against an oblivious adversary}. By this we mean that the data structure is run on some stream of updates and queries and the answer to every query is correct with probability lower bounded by a specified value (typically, $1-\eps$ for a parameter $\eps>0$ fixed in the context), assuming that the consecutive updates/queries \emph{do not} depend on the answers previously returned by that data structure. We remark that while the answer to every query is correct with probability at least, say, $1-\eps$, the answers to the queries are \emph{not} independent random variables. In fact, in all the randomized data structures proposed in this work, the answers to the queries depend on some common random bits fixed upon the initialization of the data structure. We also remark that if we are asked only about deciding the existence of distance-$r$ dominating or independent set (i.e., if we do not need to return a certifying distance-$r$ dominating set), then assuming we answer correctly, we do not leak any random bits chosen by the data structure. So in the restricted variant with binary answers, the data structures guarantees of both \cref{thm:main-domset-be} and \cref{thm:main-indset-be} may be strengthened to answer correctly any series of $q$ queries with probability at least $1-q\eps$ against an adaptive~adversary.

\paragraph*{Brodal--Fagerberg data structure.}
We recall a classic result by Brodal and Fagerberg that we will use multiple times throughout.
\begin{theorem}[Brodal--Fagerberg, \cite{BrodalF99}]\label{thm:bf}
	Let $d \in \N$. There exists a deterministic data structure that maintains a dynamic graph $G$, guaranteed to be $d$-degenerate at all times, under edge insertions and deletions.
	The data structure stores an explicit $4d$-orientation $\vec{G}$ of $G$.
	An edge insertion has amortized cost $\Oh(1)$, an edge deletion has amortized cost $\Oh(d + \log n)$, and the amortized number of edges that change orientation in a single operation is $\Oh(\log n)$. Moreover, the reorientations can be reported within the same time bounds.
	The data structure uses $\Oh(n + m)$ memory, where $n$ and $m$ are the number of vertices of $G$ and the current number of edges of $G$, respectively. The data structure requires the knowledge of $d$.
	In addition to that, the data structure can also handle adjacency queries in $\Oh(d)$ worst case time per query.
\end{theorem}


\section{Progressive exploration}\label{sec:ladders}

In this section we recall the framework of \emph{progressive exploration}, proposed by Fabia\'nski et al.~\cite{FabianskiPST19,Fabianski2018arxiv} for  designing parameterized algorithms for \drds{r} and \dris{r}, among other problems of similar kind. While \cite{FabianskiPST19} is the conference version, we refer to the arxiv version \cite{Fabianski2018arxiv} for all the proofs. Our main focus here is to explain that the progressive exploration algorithms for \drds{r} and \dris{r} can be executed using a bounded number of calls to the following two basic queries, each working on a graph $G$ and stipulated by a distance parameter $r\in \N$:
\begin{itemize}
	\item $\far{S,r}$: Given a set $S \subseteq V(G)$, return a vertex $v$ of $G$ such that $\dist_G(v,S)>r$; or $\bot$ if no such vertex exists.
	\item $\near{S,r}$: Given a set $S \subseteq V(G)$, return a vertex $v$ of $G$ such that $\dist_G(v,s)\leq r$ for all $s\in S$; or $\bot$ if no such vertex exists.
\end{itemize}
For the \dris{r} problem, we will need a slight generalization of these.
This will allow us to concentrate in subsequent sections on implementing those queries efficiently in the setting of dynamic data structures.

\subsection{Domination}
For \drds{r}, Fabia\'nski et al. proposed a \emph{semi-ladder algorithm} that proceeds as follows (see \cref{alg:1} for a pseudo-code). The algorithm is employed on a graph $G$ and searches for a distance-$r$ dominating set of size $k$, for given parameters $r,k\in \N$. It assumes access to two queries:
\begin{itemize}
	\item $\notDominated{D,r}$:
	Given a set $D$ of at most $k$ vertices, return a vertex $v$ that is not distance-$r$ dominated by~$D$, that is, such that $\dist_G(v,D)>r$. In case no such vertex $v$ exists, return $\bot$.
	\item $\candidate{W,r,k}$: Given a vertex subset $W$, return a set $D$ of at most $k$ vertices that distance-$r$ dominates~$W$, that is, $\dist_G(w,D)\leq r$ for all $w\in W$. In case no such set $D$ exists, return $\bot$.
\end{itemize}
Note that the $\notDominated{D,r}$ query is just equivalent to the $\far{D,r}$ query. We later~show that also the $\candidate{W,r,k}$ query can be implemented using a bounded (in terms of $k$ and~$|W|$) number of $\near{S,r}$ queries with $S\subseteq W$.

The algorithm iteratively constructs two sequences:
\begin{itemize}
	\item a sequence of \emph{candidates} $D_1,D_2,\ldots$, each being a set of at most $k$ vertices; and 
	\item a sequence of \emph{witnesses} $w_1,w_2,\ldots$, each being a single vertex of $G$.
\end{itemize}
After the $(i-1)$st iteration, candidates $D_1,\ldots,D_{i-1}$ and witnesses $w_1,\ldots,w_{i-1}$ are already constructed. Then the $i$th iteration proceeds as follows:
\begin{itemize}
	\item First, we call $\candidate{W,r,k}$ where $W$ comprises of all the  witnesses gathered so far, that is, $W\coloneqq \{w_1,\ldots,w_{i-1}\}$. This either yields the next candidate $D_i$ that distance-$r$ dominates $W$, or a conclusion that no such $D_i$ exists. In the latter case, we may conclude that $G$ has no distance-$r$ dominating set of size at most $k$.
	\item Second, we call $\notDominated{D_i,r}$ to verify whether $D_i$ already distance-$r$ dominates the whole graph. If so, then we have found a solution, and otherwise the call provides the next witness $w_i$, with which we may proceed to the next iteration.
\end{itemize}
It is clear that when this algorithm returns an answer, then this answer is always correct. The main insight of Fabia\'nski et al. is that since every candidate $D_i$ distance-$r$ dominates all the witnesses $w_1,\ldots,w_{i-1}$ but not the witness $w_i$, the constructed candidates and witnesses form a pattern called a \emph{semi-ladder}. Since semi-ladders cannot be too long in classes of bounded expansion, this provides a bound on the number of iterations that the algorithm executes before breaking the loop and reporting an outcome.

\begin{algorithm}
	\SetKwInOut{Input}{Input}
	\SetKwInOut{Output}{Output}
	
	\nonl
	\underline{procedure DomSet} $(G,r,k)$\\
	\Input{A graph $G$ and parameters $r,k\in \N$}
	\Output{A distance-$r$ dominating set of size $k$, or $\bot$ if no such set exists.}
	\medskip
	$W \coloneqq \emptyset$\;
	\Repeat
	{
		$D\coloneqq \candidate{W,r,k}$\;
		\If{$D = \bot$}
		{
			\KwRet{$\bot$} \;
		}
		$w \coloneqq \notDominated{D,r}$ \;
		\If{$w = \bot$}
		{
			\KwRet{$D$} \;
		}
		$W \coloneqq W \cup \{ w \}$ \;
	}
	\caption{Semi-ladder algorithm for distance-$r$ dominating set of size $k$ in a graph $G$}\label{alg:1}
\end{algorithm}

To be more precise, a \emph{distance-$r$ semi-ladder} of order $\ell$ in a graph $G$ is a pair of sequences of vertices $a_1,\ldots,a_\ell$ and $b_1,\ldots,b_\ell$ such that
\begin{itemize}
	\item $\dist_G(a_i,b_i)>r$ for each $i\in [\ell]$; and
	\item $\dist_G(a_i,b_j)\leq r$ for all $1\leq j<i\leq \ell$.
\end{itemize}
The \emph{distance-$r$ semi-ladder index} of a graph $G$, denoted $\sli_r(G)$, is the largest order of a distance-$r$ semi-ladder that can be found in $G$. For a graph class $\Cc$, we define its distance-$r$ semi-ladder index as $\sli_r(\Cc)\coloneqq \sup_{G\in \Cc} \sli_r(G)$. Note that this value may be infinite if graphs from $\Cc$ contain distance-$r$ semi-ladders of arbitrarily large order. But as proved by Fabia\'nski et al., this does not happen in bounded expansion~classes.

\begin{theorem}[{\cite[Lemma~29]{Fabianski2018arxiv}}]\label{thm:sli-be}
	For every class $\Cc$ of bounded expansion and $r\in \N$, $\sli_r(\Cc)$ is finite.
\end{theorem}

We remark that \cref{thm:sli-be} holds even in a larger generality when $\Cc$ is nowhere dense, but we will not use this here. We also note that for $r=1$, excluding a biclique suffices to bound the semi-ladder index.

\begin{theorem}[{\cite[Lemma~33]{Fabianski2018arxiv}}]\label{thm:sli-biclique}
	For every graph $G$ that does not contain the biclique $K_{t,t}$ as a subgraph, for some $t\in \N$, we have $\sli_1(G)<3t$.
\end{theorem}

The following statement summarizes how a bound on the semi-ladder index influences the number of iterations executed by the semi-ladder algorithm.

\begin{theorem}[consequence of {\cite[Lemma~5 and Corollary~13]{Fabianski2018arxiv}}]\label{thm:termination}
	Let $\Cc$ be a graph class and $r\in \N$ be such that $\sli_r(\Cc)$ is finite. Then the semi-ladder algorithm (\cref{alg:1}) run on any graph from $\Cc$ terminates after executing less than $k^{\sli_r(\Cc)}$ iterations.
\end{theorem}
\begin{proof}[Proof sketch]
Let us briefly sketch how this result follows from the discussion in~\cite{Fabianski2018arxiv}.

In~\cite{Fabianski2018arxiv}, Fabia\'nski et al. define the semi-ladder index of first-order formulas on graph classes as follows. For a graph $G$ and a first-order formula $\phi(\tup x,\tup y)$, where $\tup x,\tup y$ are tuples of variables, they define $\phi(G)$ be a bipartite graph $(A \uplus B, F)$ where $A=V(G)^{\tup x}$ (evaluations of variables of $\tup x$ in $V(G)$), $B=V(G)^{\tup y}$ (same for $\tup y$), and for $\tup a\in A$ and $\tup b\in B$, we put $\tup a\tup b \in F$ iff $\phi(\tup a,\tup b)$ holds in $G$.
A \emph{semi-ladder} of order $\ell$ in $\varphi(G)$ is a pair of sequences $\tup a_1,\ldots,\tup a_\ell\in A$ and $\tup b_1,\ldots,\tup b_\ell\in B$ such that $a_ib_i\notin F$ for all $i\in [\ell]$ and $a_ib_j\in F$ for all $i,j\in [\ell]$ with $i<j$. The \emph{semi-ladder index} of $\phi(G)$ is the largest order of a semi-ladder in $\phi(G)$, and the \emph{semi-ladder index} of a formula $\phi(\tup x,\tup y)$ on a graph class $\Cc$ is defined as the supremum of the semi-ladder indices of graphs $\phi(G)$, for $G\in\Cc$.

In this notation, the quantity $\sli_r(\Cc)$ is equal to the semi-ladder index on $\Cc$ of the formula $\delta_r(x,y)$ that verifies whether $x$ and $y$ are at distance at most $r$. \cite[Lemma~5]{Fabianski2018arxiv} then says that the semi-ladder index of the formula $\delta^k_r(\tup x,y)=\bigvee_{i=1}^k \delta_r(x_i,y)$, which checks whether the $k$-tuple of vertices $\tup x$ distance-$r$ dominates $y$, has semi-ladder index smaller than $k^\ell$, where $\ell=\sli_r(\Cc)$ is the semi-ladder index of $\delta_r(x,y)$. With this, \cite[Corollary~13]{Fabianski2018arxiv} asserts that $k^\ell-1$ is an upper bound on the number of iterations executed by the semi-ladder algorithm.
\end{proof}


Note here that if the algorithm performs at most $\ell\coloneqq k^{\sli_r(\Cc)}$ iterations in total, then all the calls to the $\candidate{W,r,k}$ query are applied only to sets $W$ of size at most $\ell$.

Our goal now is to reduce implementation of the $\notDominated{D,r}$ and $\candidate{W,r,k}$ queries to the more basic queries $\far{S,r}$ and $\near{S,r}$. For $\notDominated{D,r}$ this is trivial: queries $\notDominated{D,r}$ and $\far{D,r}$ are just equivalent. For $\candidate{W,r,k}$, this is a bit more complex.

\begin{lemma}\label{lem:candidate-to-near}
	Let $G$ be a graph, $W$ be a subset of vertices of $G$, and $k,r\in \N$ be parameters. Then the query $\candidate{W,r,k}$ can be answered in time $\Oh(k^{|W|+1})$ by performing at most $k^{|W|+1}$ queries of the form $\near{S,r}$, where $S$ is a subset of $W$.
\end{lemma}
\begin{proof}
	We enumerate all
	the partitions of $W$ into at most $k$ subsets; there are at most $k^{|W|}$ such partitions. For every partition $\cal P$ and every part $S\in \cal P$, we issue the query $\near{S,r}$ to verify whether there exists a vertex $v_S$ such that $\dist_G(v_S,s)\leq r$ for all $s\in S$. If for some partition $\cal P$ we can find such a vertex $v_S$ for every part $S\in \cal P$, then the set $D\coloneqq \{v_S\colon S\in \cal P\}$ is a valid answer to the query $\candidate{W,r,k}$. And if this check fails for every considered partition $\cal P$, then there is no set $D$ of size at most $k$ that distance-$r$ dominates $W$ and we can safely return $\bot$.
	Thus, the total
	number of queries $\near{S,r}$ issued is bounded by
	$k^{|W|+1}$.
\end{proof}

We conclude this section by a statement summarizing the discussion and implying the following: designing an efficient data structure for the $\near{S,r}$ and $\far{S,r}$ queries suffices to implement the semi-ladder algorithm in the dynamic setting, thereby providing a dynamic data structure for the \drds{r} problem.

\begin{lemma}\label{lem:summary-dom}
	Let $\Cc$ be a class of graphs, $k,r\in \N$ be such that $\sli_r(\Cc)$ is finite, and $\ell\coloneqq k^{\sli_r(\Cc)}$. Let $G$ be a dynamic graph that belongs to $\Cc$ at all times. Suppose that there are data structures $\nearVertexDS{G}{r,\ell}$, supporting queries $\near{S,r}$ in $G$ with $|S|\leq \ell$, and $\farVertexDS{G}{r,k}$, supporting queries $\far{S,r}$ in $G$ with $|S|\leq k$, so that $\nearVertexDS{G}{r,\ell}$ and $\farVertexDS{G}{r,k}$ have amortized update/query time $T$, initialization time $I$, and space complexity $M$. Then there is a data structure $\dominatingSetDS{G}{r,k}$ for $G$ that supports the~query
	$\ds{}$, which returns a distance-$r$ dominating set of size at most $k$ in $G$, or $\bot$ if no such set exists.
	The amortized update time of $\dominatingSetDS{G}{r,k}$ is $2^{k^{\Oh(\sli_r(\Cc))}}\cdot T$, the query time is $\Oh(k)$, the initialization time is $\Oh(I)$, and the space complexity is $\Oh(M)$.
	
	Moreover, we may allow that the data structures $\nearVertexDS{G}{r,\ell}$ and $\farVertexDS{G}{r,k}$ are randomized with error probability at most $\eps$, for any fixed parameter $\eps>0$ and against an oblivious adversary, with amortized update/query time, initialization time, and space complexity becoming $T\cdot \log \tfrac{1}{\eps}$, $I\cdot \log \tfrac{1}{\eps}$, and $M\cdot \log \tfrac{1}{\eps}$, respectively. In this case,  $\dominatingSetDS{G}{r,k}$ is also randomized with error probability at most $\eps$ against an oblivious adversary. The amortized update time becomes $2^{k^{\Oh(\sli_r(\Cc))}}\cdot T\cdot \log \tfrac{1}{\eps}$, the query time remains $\Oh(k)$, the initialization time becomes $k^{\Oh(\sli_r(\Cc))}\cdot I\cdot \log \tfrac{1}{\eps}$, and the space complexity becomes  $k^{\Oh(\sli_r(\Cc))}\cdot M\cdot \log \tfrac{1}{\eps}$.
\end{lemma}
\begin{proof}
	For simplicity, we first assume that data structures $\nearVertexDS{G}{r,\ell}$ and $\farVertexDS{G}{r,k}$ are deterministic. We will then explain how to strengthen the proof if they are randomized.
	
	The data structure $\dominatingSetDS{G}{r,k}$ simply maintains the assumed data structures $\nearVertexDS{G}{r,\ell}$ and $\farVertexDS{G}{r,k}$. Upon every update to $G$, we relay the update to $\nearVertexDS{G}{r,\ell}$ and $\farVertexDS{G}{r,k}$ and we run the semi-ladder algorithm on $G$ to find a distance-$r$ dominating set of size at most $k$ (or detect a lack thereof). By \cref{thm:termination}, this algorithm can be implemented using at most $\ell$ queries $\notDominated{D,r}$ with $|D|\leq k$ and at most $\ell$ queries $\candidate{W,r,k}$ with $|W|\leq \ell$. Every query $\notDominated{D,r}$ can be answered by a single call to the query $\far{D,r}$ offered by $\farVertexDS{G}{r,k}$. By \cref{lem:candidate-to-near}, every query $\candidate{W,r,k}$ can be answered by performing $k^{|W|+1}\leq 2^{k^{\Oh(\sli_r(\Cc))}}$ calls to queries of the form $\near{S,r}$ with $S\subseteq W$, offered by $\nearVertexDS{G}{r,\ell}$. Once a distance-$r$ dominating set of size $k$ (or lack thereof) is computed, it can be provided in time $\Oh(k)$ upon any $\ds{}$ query.
	The claimed bounds on the amortized update time of $\dominatingSetDS{G}{r,k}$ follow immediately from the assumed guarantees about $\nearVertexDS{G}{r,\ell}$ and $\farVertexDS{G}{r,k}$; and similarly for the initialization time and the space complexity.
	
	In case $\nearVertexDS{G}{r,\ell}$ and $\farVertexDS{G}{r,k}$ are randomized, we first assume that they work against an adaptive adversary and we will then explain how to improve the argument if they only work against an oblivious adversary. The data structure $\dominatingSetDS{G}{r,k}$ maintains their instances with a rescaled error parameter $\eps'\coloneqq \tfrac{\eps}{q}$, where
	\[q\coloneqq k^{\sli_r(\Cc)}+k^{k^{\sli_r(\Cc)}}\leq 2^{k^{\Oh(\sli_r(\Cc))}}\]
	is an upper bound on the total number of queries to $\nearVertexDS{G}{r,\ell}$ or $\farVertexDS{G}{r,k}$ that may occur in a single run of the semi-ladder algorithm. Therefore, by the union bound, we conclude that with probability at least $1-\eps$ none of these queries returns an incorrect answer and the semi-ladder algorithm runs correctly. It is straightforward to verify that rescaling the error probability in $\nearVertexDS{G}{r,\ell}$ and $\farVertexDS{G}{r,k}$ influences the complexity guarantees for $\dominatingSetDS{G}{r,k}$ as claimed.

	We remark that even if the adversary issuing the updates to our graph is oblivious, this argument still required $\nearVertexDS{G}{r,\ell}$ and $\farVertexDS{G}{r,k}$ to work against an adaptive adversary, because the semi-ladder algorithm acts as an internal adaptive adversary --- it issues a number of $\mathtt{nearVertex}/\mathtt{farVertex}$ queries that depend on the answers to previous $\mathtt{nearVertex}/\mathtt{farVertex}$ queries in that execution of the \cref{alg:1}. To remedy this, we initialize multiple independent instances of our helper data structures: $\nearVertexDSS{G}{r,\ell}{1}, \ldots, \nearVertexDSS{G}{r,\ell}{\ell}, \farVertexDSS{G}{r,k}{1}, \ldots, \farVertexDSS{G}{r,k}{\ell}$. In the $i$th iteration of the main loop of the \cref{alg:1} we use  $\nearVertexDSS{G}{r,\ell}{i}$ and $\farVertexDSS{G}{r,k}{i}$ data structures to answer any $\mathtt{nearVertex}$ and $\mathtt{farVertex}$ queries asked during that iteration. The $\mathtt{nearVertex}$ queries asked within the $i$th iteration of different runs do not depend on answers to each other, hence they can all be handled by a single instance $\nearVertexDSS{G}{r,\ell}{i}$ working only against an oblivious adversary; and similarly for the $\farVertexDSS{G}{r,k}{i}$ data structure. This guarantees the correctness in the setting with the weaker assumptions, but it incurs an additional multiplicative $\ell$ factor to update/query time, initialization time, and memory. However, this additional $\ell$ factor does not change the complexities as stated.
\end{proof}

\subsection{Independence}

For the \dris{r} problem, Fabia\'nski et al.~\cite{FabianskiPST19,Fabianski2018arxiv} proposed a different procedure, called the \emph{ladder algorithm}. Similarly to the semi-ladder algorithm, it is based on alternately finding candidates for the sought distance-$r$ independent set and witnesses that the current candidate is not yet a solution. The notion of witnessing is, however, more delicate. For the further discussion, we fix the distance parameter $r\in \N$.

\begin{definition}
	Let $G$ be a graph  and $\tup a$ be a tuple of vertices $G$. We say that a vertex $w\in V(G)$ is a \emph{distance-$r$ dependence witness} for $\tup a$ if there are distinct $i,j\in [|\tup a|]$ such that $\dist_G(a_i,w)+\dist_G(a_j,w)\leq r$. More generally, a set $P\subseteq V(G)$ is a distance-$r$ dependence witness for $\tup a$ if $P$ contains some $w\in P$ that is a distance-$r$ dependence witness for $\tup a$.
\end{definition}

Note that a $k$-tuple $\tup a$ of vertices in a graph $G$ forms a distance-$r$ independent set of size $k$ if and only if there is no distance-$r$ dependence witness for $\tup a$.

The ladder algorithm will be parameterized by one more parameter $p\in \N$, which stipulates the size of dependence witnesses investigated by the algorithm. Finding candidates and witnesses is delegated to the following two queries, which will eventually be implemented using (generalizations of) $\near{S,r}$ and $\far{S,r}$.
\begin{itemize}
	\item $\witnessInd{\cal L,r,p}$:
	Given a family $\cal L$ of tuples of vertices, return a set $P\subseteq V(G)$ of size at most~$p$ such that $P$ is a distance-$r$ dependence witness for all $\tup a\in \cal L$. In case no such vertex set $P$ exists, return $\bot$.
	\item $\candidateInd{W,r,k}$: Given a vertex subset $W$, return $k$-tuple $\tup a$ of vertices such that $W$ is not a distance-$r$ dependence witness for $\tup a$. In case no such tuple $\tup a$ exists, return $\bot$.
\end{itemize}

\begin{algorithm}
	\SetKwInOut{Input}{Input}
	\SetKwInOut{Output}{Output}
	
	\nonl
	\underline{procedure IndSet} $(G,r,k,p)$\\
	\Input{A graph $G$ and parameters $r,k,p\in \N$}
	\Output{Conclusion whether there exists a distance-$r$ independent set of size $k$ in $G$}
	\medskip
	${\cal L}\coloneqq \emptyset$\;
	$W \coloneqq \emptyset$\;
	\Repeat
	{
		$\tup a\coloneqq \candidateInd{W}$\;
		\If{$\tup a = \bot$}
		{
			\KwRet{``There is no distance-$r$ independent set of size $k$''} \;
		}
		${\cal L}\coloneqq {\cal L}\cup \{\tup a\}$ \;
		$P \coloneqq \witnessInd{\cal L,p}$ \;
		\If{$P = \bot$}
		{
			\KwRet{``There exists a distance-$r$ independent set of size $k$''} \;
		}
		$W \coloneqq W \cup P$ \;
	}
	\caption{Ladder algorithm for distance-$r$ independent set of size $k$ in a graph $G$}\label{alg:2}
\end{algorithm}


The ladder algorithm for parameters $r,k,p$ proceeds iteratively in rounds as follows (see \cref{alg:2} for a pseudocode). We construct two sequences: candidates $\tup a_1,\tup a_2,\tup a_3,\ldots$, each being a $k$-tuple of vertices, and witnesses $P_1,P_2,P_3,\ldots$, each being a vertex subset of size at most $p$. After the $(i-1)$st iteration, the candidates $\tup a_1,\ldots,\tup a_{i-1}$ and the witnesses $P_1,\ldots,P_{i-1}$ are already constructed. Then, the $i$th iteration is as follows:
\begin{itemize}
	\item First, we call $\candidateInd{W,r,k}$ with $W=P_1\cup P_2\cup \ldots \cup P_{i-1}$ consisting of the union of all the witnesses found so far. If this call returns $\bot$, then we can safely report that there is no distance-$r$ independent set of size $k$ in $G$. Otherwise, the output of the call becomes the next candidate $\tup a_i$.
	\item Second, we call $\witnessInd{\cal L,r,p}$, where $\cal L\coloneqq \{\tup a_1,\ldots,\tup a_i\}$ is the set of all the candidates found so far. If this call returns $\bot$, we report that there exists a distance-$r$ independent set of size $k$ in $G$. Otherwise, the witness output by the call becomes the next witness $P_i$, with which we can proceed to the next iteration.
\end{itemize}
Note that the correctness of the algorithm is not obvious: When the algorithm reports in the second point that there exists a distance-$r$ independent set of size $k$, then a priori there is no reason to assume that this is correct. In particular, the algorithm only provides this conclusion without providing any actual independent set. However, Fabia\'nski et al.~\cite{FabianskiPST19,Fabianski2018arxiv} proved that in classes of bounded expansion, the algorithm always finishes within a bounded number of rounds and provides a correct answer, provided the parameter $p$ is set large enough.

\begin{theorem}[consequence of {\cite[Theorems~5, Theorem~12, and Corollary~16]{Fabianski2018arxiv}}]\label{thm:ladder-bound}
	Let $\Cc$ be a graph class of bounded expansion and $r,k\in \N$. Then, there exist constants $p,\ell\in \Oh_{\Cc,r,k}(1)$ such that the ladder algorithm (\cref{alg:2}) run on any graph $G\in \Cc$ for parameters $r,k,p$, terminates within at most $\ell$ rounds and always outputs the correct answer. 
\end{theorem}

Thus, within a single run of the ladder algorithm as above, every $\witnessInd{\cal L,r,p}$ query involves a family $\cal L$ consisting of at most $\ell$ candidates, each being a $k$-tuple of vertices, and every $\candidateInd{W,r,k}$ query involves a vertex set $W$ with $|W|\leq p\ell$. We remark that similarly to \cref{thm:termination}, \cref{thm:ladder-bound} also holds in the larger generality where $\Cc$ is only assumed to be nowhere dense.

We now show how to implement the $\witnessInd{\cal L,r,p}$ and $\candidateInd{W,r,k}$ queries assuming access to the following generalizations of $\far{S,r}$ and $\near{S,r}$:
\begin{itemize}
	\item $\far{S,\rho}$: Given a set $S$ and a mapping $\rho\colon S\to \N$, return a vertex $v$ of $G$ such that $\dist_G(v,s)\geq\rho(s)$ for all $s\in S$; or $\bot$ if no such vertex exists.
	\item $\near{S,\rho}$: Given a set $S$ and a mapping $\rho\colon S\to \N$, return a vertex $v$ of $G$ such that $\dist_G(v,s)\leq \rho(s)$ for all $s\in S$; or $\bot$ if no such vertex exists.
\end{itemize}
In the following, for a function $\rho\colon S\to \N$ we denote $\max \rho\coloneqq \max_{s\in S} \rho(s)$. The reader should think that in the queries above we always have $\max \rho \leq r+1$ for the distance parameter $r$ we are working with. The generalization allows us to assign a different relevant distance $\rho(s)\leq r+1$ to every vertex $s$ of $S$.

\begin{lemma}\label{lem:candidateInd-far}
	Let $G$ be a graph, $W$ be a subset of vertices of $G$, and $k,r\in \N$ be parameters. Then the query $\candidateInd{W,r,k}$ can be answered in time $(r+2)^{\Oh(k|W|)}$ by performing at most $(r+2)^{k|W|}$ queries of the form $\far{S,\rho}$, where $S$ is a subset of $W$ and $\max \rho\leq r+1$.
\end{lemma}
\begin{proof}
	For a vertex $u$ of $G$, we define the \emph{profile} of $u$ as the function $\mathsf{prof}[u]\colon W\to \{0,1,\ldots,r,r+1\}$ such that for every $s\in W$,
	\[\mathsf{prof}[u](s)\coloneqq \begin{cases} \dist_G(u,s) &\textrm{if }\dist_G(u,s)\leq r,\\
		r+1 &\textrm{otherwise.}\end{cases}\]
	Let $\cal P\coloneqq \{0,1,\ldots,r,r+1\}^W$ be the set of all possible profiles. Note that $|\cal P|=(r+2)^{|W|}$. 
	
	Observe that for a $k$-tuple $\tup a$ of vertices, the $k$-tuple of profiles $(\mathsf{prof}[a_i]\colon i\in [k])$ determines whether $W$ is a distance-$r$ dependence witness for $\tup a$. Therefore, there is a set $\cal I\subseteq (\cal P)^k$ of $k$-tuples of profiles such that $W$ is a distance-$r$ dependence witness for $\tup a$ if and only if $(\mathsf{prof}[a_i]\colon i\in [k])\in \cal I$. Note that $\cal I$ can be computed in time $(r+2)^{\Oh(k|W|)}$ by investigating every $k$-tuple of profiles in $(\cal P)^k$ and deciding in time $(k|W|)^{\Oh(1)}$ whether it should be included in $\cal I$.
	
	Now, for each $\tup \pi=(\pi_1,\ldots,\pi_k)\in \cal I$ and each $i\in [k]$, we call $\far{S,\pi_i}$ to find, if existent, any vertex $a_i$ whose profile $\mathsf{prof}[a_i]$ is coordinate-wise not smaller than $\pi_i$. If for any $\tup \pi\in \cal I$ we manage to find all such vertices $a_i$, then we have found a tuple $\tup a=(a_1,\ldots,a_k)$ for which $W$ is not a distance-$r$ dependence witness; so $\tup a$ can be reported. And if this check fails for every $\tup \pi\in \cal I$, then we may safely conclude that $W$ is a distance-$r$ dependence witness for every $k$-tuple of vertices of $G$.
\end{proof}

\begin{lemma}\label{lem:witnessInd-near}
	Let $G$ be a graph, $p,k,r\in \N$, and $\cal L$ be a set of $k$-tuples of vertices of $G$. Then the query $\witnessInd{W,r,p}$ can be answered in time $(r+2)^{\Oh(pk|\cal L|)}$ by performing at most $(r+2)^{pk|\cal L|}$ queries of the form $\near{S,\rho}$, where $|S|\leq k|\cal L|$ and $\max \rho\leq r$.
\end{lemma}
\begin{proof}
	Let $L$ be the set of all the vertices featured in the tuples of $\cal L$. Note that $|L|\leq k|\cal L|$. We will use a slightly different notion of a profile than in the proof of \cref{lem:candidateInd-far}. For a vertex $u$ of $G$, its \emph{profile} is the function $\mathsf{prof}[u]\colon L\to \{0,1,\ldots,r,+\infty\}$ defined as follows: for $s\in L$,
	\[\mathsf{prof}[u](s)\coloneqq \begin{cases} \dist_G(u,s) &\textrm{if }\dist_G(u,s)\leq r,\\
		+\infty &\textrm{otherwise.}\end{cases}\]
	Let $\cal Q\coloneqq \{0,1,\ldots,r,+\infty\}^L$ be the set of all possible profiles. Note that $|\cal Q|=(r+2)^{|L|}\leq (r+2)^{k|\cal L|}$. 
	
	Observe that for a $p$-tuple of vertices $\tup b=(b_1,\ldots,b_p)$, the $p$-tuple of profiles $(\mathsf{prof}[b_i]\colon i\in [p])$ determines whether $\{b_1,\ldots,b_p\}$ is a distance-$r$ dependence witness for each of the $k$-tuples in $\cal L$. Therefore, there is a set $\cal J\subseteq (\cal Q)^p$ of $p$-tuples of profiles such that $\{b_1,\ldots,b_p\}$ is a distance-$r$ dependence witness for every $k$-tuple in $\cal L$ if and only if $(\mathsf{prof}[b_i]\colon i\in [p])\in \cal J$. Note that $\cal J$ can be computed in time $(r+2)^{\Oh(pk|\cal L|)}$ by investigating every $p$-tuple of profiles in $(\cal Q)^p$ and deciding in time $(pk|\cal L|)^{\Oh(1)}$ whether it should be included in $\cal J$.
	
	Now, for each $\tup \pi=(\pi_1,\ldots,\pi_p)\in \cal J$ and each $i\in [p]$, we let $L'\coloneqq \{s\in L~\mid~\pi_i(s)\textrm{ is finite}\}$ and call $\near{L',\pi_i|_{L'}}$ to find, if existent, any vertex $b_i$ whose profile $\mathsf{prof}[b_i]$ is coordinate-wise not larger than $\pi_i$. If for any $\tup \pi\in \cal J$ we manage to find all such vertices $b_i$, then $\{b_1,\ldots,b_p\}$ is a distance-$r$ dependence witness for all the tuples in $\cal L$; so $\{b_1,\ldots,b_p\}$ can be reported. And if this check fails for every $\tup \pi\in \cal J$, then we may safely conclude that there is no set of size at most $p$ that is a distance-$r$ dependence witness for all the tuples in $\cal L$.
\end{proof}

We conclude this section with a statement analogous to \cref{lem:summary-dom} that summarizes the ingredients needed to design a data structure for the \dris{r} problem.

\begin{lemma}\label{lem:summary-ind}
	Let $r,k\in \N$, $\Cc$ be a graph class of bounded expansion, and $p,\ell\in \N$ be the constants provided by \cref{thm:ladder-bound} for $\Cc,r,k$. Let $G$ be a dynamic graph that belongs to $\Cc$ at all times. Suppose that there are data structures $\nearVertexDS{G}{r,k\ell}$, supporting queries $\near{S,\rho}$ in $G$ with $|S|\leq k\ell$ and $\max \rho\leq r$, and $\farVertexDS{G}{r+1,p\ell}$, supporting queries $\far{S,\rho}$ in $G$ with $|S|\leq p\ell$ and $\max \rho\leq r+1$, so that $\nearVertexDS{G}{r,k\ell}$ and $\farVertexDS{G}{r+1,p\ell}$ have amortized update/query time $T$, initialization time $I$, and space complexity $M$. Then there is a data structure $\independentSetDS{G}{r,k}$ for $G$ that supports the~query
	\begin{itemize}
		\item $\is{}$: Decide whether in $G$ there exists a distance-$r$ independent set of size $k$. 
	\end{itemize}
	The amortized update time of $\independentSetDS{G}{r,k}$ is $(r+2)^{\Oh(pk\ell)}\cdot T$, the query time is $\Oh(1)$, the initialization time is $\Oh(I)$, and the space complexity is $\Oh(M)$.
	
	Moreover, we may allow the data structures $\nearVertexDS{G}{r,k\ell}$ and $\farVertexDS{G}{r+1,p\ell}$ to be randomized with error probability at most $\eps$, for any fixed parameter $\eps>0$ and against an oblivious adversary, with amortized update/query time, initialization time, and space complexity becoming $T\cdot \log \tfrac{1}{\eps}$, $I\cdot \log \tfrac{1}{\eps}$, and $M\cdot \log \tfrac{1}{\eps}$, respectively. In this case,  $\independentSetDS{G}{r,k}$ is also randomized with error probability at most $\eps$ against an oblivious adversary. The amortized update time becomes $(r+2)^{\Oh(pk\ell)}\cdot T\cdot \log\tfrac{1}{\eps}$, the initialization time becomes $(r+2)^{\Oh(pk\ell)}\cdot I\cdot \log \tfrac{1}{\eps}$, and the space complexity becomes  $(r+2)^{\Oh(pk\ell)}\cdot M\cdot \log \tfrac{1}{\eps}$.
\end{lemma}
\begin{proof}
	The proof is completely analogous to that of \cref{lem:summary-dom}, with \cref{thm:ladder-bound} used to argue that the ladder algorithm employed with parameter $p$ is correct and performs at most $\ell$ iterations, and \cref{lem:candidateInd-far,lem:witnessInd-near} used to reduce the at most $2\ell$ queries of the form $\witnessInd{\cal L,r,p}$ and $\candidateInd{W,r,k}$ to $(r+2)^{\Oh(pk\ell)}$ queries of the form $\near{S,\rho}$ and $\far{S,\rho}$ that are relayed to the data structures $\nearVertexDS{G}{r,k\ell}$ and $\farVertexDS{G}{r+1,p\ell}$. 
\end{proof}

\section{Counting and finding small patterns in sparse graphs}\label{sec:dvorak}



In this section we first recall the main results of Dvo\v{r}\'ak and T\r{u}ma~\cite{DvorakT13} about dynamic data structures for counting homomorphisms and (induced) subgraphs in classes of bounded expansion, and then we generalize them in various ways for the sake of using them in the next section.
The main result Dvo\v{r}\'ak and T\r{u}ma can be expressed as follows (this is a slightly more general formulation of \cref{thm:overview-DT}).

\begin{theorem}[\cite{DvorakT13}] \label{thm:dvorak-tuma}
	Let $H$ be a fixed graph, $\G$ be a graph class of bounded expansion, $k\in \N$, and $\F$ be either $\Hom$, $\Sub$, or $\ISub$. Suppose $G$ is a dynamic $k$-colored graph on $n$ vertices that belongs to $\G$ at all times. Then, there is a~data structure $\UnMaps_{\F, H}[G]$ that is able to determine the quantity $|\F(H, G)|$ after every update. The amortized update time is $\Oh_{\G, H, k}(\log^h n)$, where $h = {|H| \choose 2} - 1$, the initialization time is $\Oh_{\G,H,k}(n)$, and the space complexity is $\Oh_{\G, H, k}(n)$. 
\end{theorem}

\subsection{Finding examples of mappings}


While the data structure of \cref{thm:dvorak-tuma} is able to \emph{count} the appearances of $H$ in $G$ as an induced subgraph (that is, determine $|\ISub(H, G)|$), it is not easy to restore any example of such a mapping from the data structure, because the counting is performed using the Inclusion--Exclusion Principle. In fact, Dvo\v{r}\'ak and T\r{u}ma explicitly ask the question whether examples of induced subgraph isomorphisms can be also reported efficiently~\cite[Section~5]{DvorakT13}. While the same issue also appears when counting subgraphs, restoring examples of homomorphisms can be done by carefully tracing transitions with nonzero contributions in the designed dynamic programming, as this part of the argumentation in \cite{DvorakT13} does not involve subtractions or counting using the Inclusion-Exclusion Principle.


In this section, we resolve all these open questions and prove that example mappings can be efficiently reported both for homomorphisms and for (induced) subgraphs, at the cost of making the data structure randomized.
Formally, we prove the following extension of \cref{thm:dvorak-tuma}, already announced in \cref{sec:overview}.


\reporting*

For the algorithms in the following sections we will only need the case $\F = \Hom$, which is the only case that is easy to handle for the argumentation provided in the work of Dvo\v{r}\'ak and T\r{u}ma~\cite{DvorakT13} (despite not being stated and proven explicitly there). However, the toolbox that we need to introduce for resolving the harder cases of $\F = \Sub$ and $\F = \ISub$ will turn out to be useful later on; and besides, these cases can be also viewed as results of independent interest. In particular, even though the case $\F=\Hom$ does not actually require randomization, in the proofs of \cref{thm:main-domset-be,thm:main-indset-be} we will reuse some randomized parts of the toolbox, and the resulting data structures are randomized.


Let us first introduce a handy definition. A \emph{vertex function} is a function $f$ that takes a graph $G$ and its vertex $v$ and returns a nonnegative integer $f(G,v)$. If we always have $f(G,v)\in \{0,1\}$, we call the vertex function $f$ \emph{binary}. 

The following lemma, which draws inspiration from the works of Majewski, Pilipczuk, and Zych-Pawlewicz~\cite{Majewski24} and of Nadara, Pilipczuk, and Smulewicz~\cite{Nadara22}, is our main tool for recovering examples while given access only to counting. We call this technique \emph{fingerprint retrieval}.
\begin{lemma}(Fingerprint Retrieval Lemma) \label{lem:restore-by-sampling}
	Let $f$ be a vertex function and let $G$ be a dynamic graph on $n$ vertices.
	Suppose that for any fixed weight function $\wei \colon V(G) \to \Z$, there is a data structure $\WeiSum_{f, \wei}[G]$ that reports the quantity $\sum_{v \in V(G)} f(G, v)\cdot \wei(v)$ after every update to $G$.
	Also, suppose that
	\begin{enumerate}[label={(\arabic*)}]
		\item \label{it:binary-f} $f$ is binary, or
		\item \label{it:general-f} there is a data structure $\Ver_{f}[G]$ that supports the following queries: for a given vertex $v$, decide whether $f(G, v)>0$.
	\end{enumerate}
	Assume further that for the data structure $\WeiSum_{f, \wei}[G]$ and, in case \ref{it:general-f}, also for $\Ver_f[G]$, we have the following complexity guarantees: amortized update and query time at most $T$, initialization time at most~$I$, and space complexity at most $M$.

	Then, for every~$\eps>0$ there exists a randomized data structure $\RetVer_{f,\eps}[G]$ that support the following query: return any $v$ such that $f(G, v) > 0$, or $\bot$ if no such vertex exists. The amortized update and query time is $\Oh(T\log n\log \tfrac{1}{\eps})$, the initialization time is $\Oh((I+n)\log n\log \tfrac{1}{\eps})$, and the space complexity is $\Oh((M+n)\log n\log \tfrac{1}{\eps})$. The data structure never provides false positive, but may fail to provide an example vertex $v$ with probability at most $\eps$, against an oblivious adversary.
\end{lemma}


\begin{proof}
Throughout the proof we assume that the vertex set of $G$ is $V(G)=[n]=\{1,2,\dots,n\}$ with $n>0$. Recall that every update preserves this vertex set.
Thus, every vertex is permanently identified with its integer label.

Assume without loss of generality that $\eps\leq 1$. Set 
\[L\coloneqq \lfloor \log_2 n\rfloor + 1\qquad \textrm{and} \qquad\xi\coloneqq\left\lceil 8\ln\tfrac{1}{\eps}\right\rceil.\]
Note that $\xi\ge 1$.



\paragraph*{Sampling sets $S_{i,j}$.} Upon initialization of the data structure, for every pair $(i,j)\in [0,L]\times [\xi]$ we construct a set $S_{i,j}\subseteq V(G)$ by including every $v\in V(G)$ in $S_{i,j}$ with probability $p_i\coloneqq 2^{-(i+1)}$.
All these random choices are independent over all vertices and all pairs $(i,j)$. After initialization, the sets $S_{i,j}$ remain fixed forever.

We store each set $S_{i,j}$ explicitly as an $n$-bit array indexed by $[n]$.
Hence, the whole family $\{S_{i,j}\colon (i,j)\in [0,L]\times [\xi]\}$ is stored once and uses $\Oh(n(L+1)\xi)$ additional space.
Note that for every fixed subset $X\subseteq V(G)$ and every fixed pair $(i,j)$, the random variable $|S_{i,j}\cap X|$ has distribution $\mathrm{Bin}(|X|,p_i)$.

\paragraph*{Weighted sums.}
For every pair $(i,j)\in [0,L]\times [\xi]$, define functions $\wei^{(0)}_{i,j},\wei^{(1)}_{i,j}\colon V(G)\to \Z$ by
\[
\wei^{(0)}_{i,j}(v) \coloneqq
\begin{cases}
1 & \text{if } v\in S_{i,j},\\
0 & \text{if } v\notin S_{i,j},
\end{cases}
\qquad
\wei^{(1)}_{i,j}(v) \coloneqq
\begin{cases}
v & \text{if } v\in S_{i,j},\\
0 & \text{if } v\notin S_{i,j}.
\end{cases}
\]
These weight functions are fixed after the sampling upon initialization.
As $V(G)=[n]\subseteq \Z$, both $\wei^{(0)}_{i,j}$ and $\wei^{(1)}_{i,j}$ indeed map $V(G)$ to~$\Z$.

Let us give some intuition behind this definition.
By setting $\wei^{(1)}_{i,j}(v)$ to $v$ on $S_{i,j}$ and to $0$ outside, we make sure that for any $X\subseteq V(G)$ with $|X\cap S_{i,j}|=1$, the sum $\sum_{v\in X} \wei^{(1)}_{i,j}(v)$ is exactly the (identifier of the) unique vertex of $X\cap S_{i,j}$. The choice of the probabilities in the definition of the sets $S_{i,j}$ makes sure for any non-empty $X$, the probability that $|X\cap S_{i,j}|=1$ for at least one pair $(i,j)\in [0,L]\times [\xi]$ is at least $1-\eps$. Therefore, if we manage to maintain the sums $\sum_{v\in X} \wei^{(1)}_{i,j}(v)$ for all $(i,j)\in [0,L]\times [\xi]$, where $X\coloneqq \{v\mid f(G,v)>0\}$, then with high probability one of these sums is actually equal to the identifier of some vertex belonging to $X$ --- so the example we seek.


\paragraph*{Data structure.}
We now give the description of the data structure $\RetVer_{f,\eps}[G]$. For each pair $(i,j)\in [0,L]\times [\xi]$ we maintain two instances of the $\WeiSum$ data structure: $\WeiSum_{f,\wei^{(0)}_{i,j}}[G]$ and $\WeiSum_{f,\wei^{(1)}_{i,j}}[G]$.
In case \ref{it:general-f}, we additionally maintain one instance of $\Ver_f[G]$, which will be used to filter out candidates for examples.
Every update of the dynamic graph $G$ is relayed to all the maintained instances.

\paragraph*{Query algorithm.}
When a query is issued to the current $G$, we inspect all the pairs $(i,j)\in [0,L]\times [\xi]$ in lexicographic order.
For every pair $(i,j)$ we compute the values 
\[\alpha_{i,j}\coloneqq \sum_{v\in V(G)} f(G,v)\cdot \wei^{(0)}_{i,j}(v)\qquad\textrm{and}\qquad \beta_{i,j}\coloneqq \sum_{v\in V(G)} f(G,v)\cdot \wei^{(1)}_{i,j}(v)\] by querying the maintained instances of $\WeiSum$.
If $\alpha_{i,j}=0$, then we continue to the next pair.
Since $f(G,v)\ge 0$ for all $v\in V(G)$ and $\wei^{(0)}_{i,j}(v)\in\{0,1\}$, we always have $\alpha_{i,j}\ge 0$.
Consequently, whenever $\alpha_{i,j}\neq 0$, in fact $\alpha_{i,j}>0$.
\begin{itemize}
\item Case \ref{it:binary-f}: If $\alpha_{i,j}=1$, then we return the vertex $\beta_{i,j}$, because then exactly one vertex of $S_{i,j}$ contributes a nonzero value to the sum defining $\beta_{i,j}$. Otherwise, if $\alpha_{i,j}>1$, we continue.
\item Case \ref{it:general-f}: Since $\alpha_{i,j}>0$ and therefore is nonzero, we form the rational number \[q\coloneqq \frac{\beta_{i,j}}{\alpha_{i,j}}.\]
When several positive summands contribute to the sums, this quotient need not identify any unique witness; this is why the algorithm needs to perform an explicit verification step.
If $q$ is not an integer, or if $q\notin[n]$, then we continue.
Otherwise, we query $\Ver_f(G)$ on the vertex $q$.
If $\Ver_f(G)$ confirms that $f(G,q)>0$, then we return $q$; otherwise we continue.
\end{itemize}
If no pair $(i,j)$ leads to returning a vertex, then we output $\bot$:  there is no vertex $v$ with $f(G,v)>0$.

\paragraph*{No false positives.} We now argue that when the query algorithm returns a vertex $v$, then we for sure have $f(G,v)>0$.
Suppose first that we are in the case \ref{it:binary-f} and the algorithm returns a vertex while processing some pair $(i,j)\in [0,L]\times [\xi]$.
Then $\alpha_{i,j}=1$, that is, $\sum_{v\in V(G)} f(G,v)\cdot \wei^{(0)}_{i,j}(v)=1$.
Since $\wei^{(0)}_{i,j}$ is the indicator function of $S_{i,j}$, this is equivalent to $\sum_{v\in S_{i,j}} f(G,v)=1$.
Now, every summand belongs to $\{0,1\}$ due to $f$ being binary.
Therefore, there exists exactly one vertex $x\in S_{i,j}$ such that $f(G,x)=1$, and every other vertex $v\in S_{i,j}\setminus\{x\}$ satisfies $f(G,v)=0$.
Consequently, $\beta_{i,j}=\sum_{v\in S_{i,j}} f(G,v)\cdot v=x$.
We conclude that the query algorithm returns the vertex $x$, and this vertex satisfies $f(G,x)=1>0$.
Since $x\in V(G)=[n]$, the returned value is indeed a valid vertex label.

Now suppose that we are in the case \ref{it:general-f}.
Then the algorithm returns a vertex only after first checking that the candidate $q$ is an integer in $[n]$, and then asking $\Ver_f[G]$, which confirms that the returned vertex $q$ indeed satisfies $f(G,q)>0$.
Hence every returned vertex is correct.
Note that such a returned vertex $q$ need not belong to the currently processed set $S_{i,j}$: when several positive vertices contribute, $q=\beta_{i,j}/\alpha_{i,j}$ is a weighted average of their labels and may accidentally equal the label of a vertex outside $S_{i,j}$.
This is irrelevant, because the required output is any vertex with positive $f$-value and $\Ver_f[G]$ confirms $f(G,q)>0$ before returning~$q$.

Therefore, the query algorithm returns no false positives in both cases.
If the algorithm outputs $\bot$, then it makes no positive claim, so again no false positive occurs.

\paragraph*{Probability of a false negative.} We now bound the probability that the algorithm returns a false negative.
Fix an arbitrary time step of the execution and suppose that a query is made at this time step.
Let $G$ be the graph at that moment.
Because the adversary is oblivious, $G$ depends only on the fixed sequence of updates, and not on the random family $\{S_{i,j}\colon (i,j)\in [0,L]\times [\xi]\}$ sampled at the initialization.

Define \[X\coloneqq \{v\in V(G)\mid f(G,v)>0\}\qquad\textrm{and}\qquad m\coloneqq |X|.\]
Thus, $X$ is a fixed subset of $V(G)$ from the point of view of the probability space governing the initial sampling of the family $\{S_{i,j}\colon (i,j)\in [0,L]\times [\xi]\}$.

If $m=0$, then $f(G,v)=0$ for every vertex $v\in V(G)$.
Therefore, for every pair $(i,j)\in [0,L]\times [\xi]$ we have $\alpha_{i,j}=\sum_{v\in V(G)} f(G,v)\cdot \wei^{(0)}_{i,j}(v)=0$, so the algorithm returns the answer $\bot$, and this answer is correct.
Hence, only the case $m\ge 1$ requires further analysis.


Set \[i^\star\coloneqq \lfloor \log_2 m\rfloor + 1\]
Since $1\le m\le n$, we have $1=\lfloor \log_2 1\rfloor + 1 \le i^\star \le \lfloor \log_2 n\rfloor + 1 = L$.
Let $p\coloneqq p_{i^\star}=2^{-(i^\star+1)}$.
Exponentiating base~$2$ yields $2^{i^\star-1}\le m < 2^{i^\star}$ and therefore \[1/4 = 2^{i^\star-1}\cdot 2^{-(i^\star+1)} \le mp < 2^{i^\star}\cdot 2^{-(i^\star+1)} = 1/2.\]
In particular, $p < 1/(2m)$.

Fix any $j\in [\xi]$ and consider the random variable $Y\coloneqq |S_{i^\star,j}\cap X|$.
Since every vertex of $X$ is included in $S_{i^\star,j}$ independently with probability $p$, we have $Y\sim\mathrm{Bin}(m,p)$.
Therefore \[\Pr(Y=1)=mp(1-p)^{m-1}.\]
Since $p<1/(2m)\le 1/2<1$, Bernoulli's inequality gives $(1-p)^{m-1}\ge 1-(m-1)p$.
As $p<1/(2m)$, we have $(m-1)p < (m-1)/(2m) < 1/2$, so $1-(m-1)p > 1/2$.
Consequently, \[(1-p)^{m-1} > 1/2.\]
Combining this with $mp\ge 1/4$, we get \[\Pr(Y=1)=mp(1-p)^{m-1} > 1/8.\]
Thus, we conclude that \[\Pr\bigl(|S_{i^\star,j}\cap X|=1\bigr) > 1/8\qquad\textrm{for every fixed }j\in [\xi].\]

The sets $S_{i^\star,1},\dots,S_{i^\star,\xi}$ are sampled independently, so the events $\bigl\{|S_{i^\star,j}\cap X|=1\bigr\}$, for $j=1,\dots,\xi$, are independent.
Hence \[\Pr\Bigl(|S_{i^\star,j}\cap X|\neq 1\textrm{ for all }j\in [\xi]\Bigr) \le \left(1-1/8\right)^\xi \le \exp\!\left(-\xi/8\right) \le \eps,\]
where the last inequality holds by the definition of $\xi$.

It remains to prove that if $|S_{i^\star,j}\cap X|=1$ for some $j\in[\xi]$, then the query algorithm returns a correct vertex no later than when it processes the pair $(i^\star,j)$.
Let $x$ be the unique element of $S_{i^\star,j}\cap X$.
Then $f(G,x)>0$ by the definition of $X$.
If the algorithm has already returned a vertex for an earlier pair, then that earlier returned vertex is correct by the no-false-positives property.
So suppose that the pair $(i^\star,j)$ is actually processed.

\begin{itemize}
\item Case \ref{it:binary-f}: Every vertex $v\in S_{i^\star,j}\setminus\{x\}$ satisfies $v\notin X$, hence $f(G,v)=0$, whereas $f(G,x)=1$ because $f$ is binary.
Therefore \[\alpha_{i^\star,j}=\sum_{v\in S_{i^\star,j}} f(G,v)=1\qquad \textrm{and}\qquad \beta_{i^\star,j}=\sum_{v\in S_{i^\star,j}} f(G,v)\cdot v=x.\]
Thus the algorithm returns a correct vertex $x$ satisfying $f(G,x)>0$.
\item Case \ref{it:general-f}: Again, every vertex $v\in S_{i^\star,j}\setminus\{x\}$ satisfies $v\notin X$, hence $f(G,v)=0$.
Therefore,
\[\alpha_{i^\star,j}=\sum_{v\in S_{i^\star,j}} f(G,v)=f(G,x)>0\qquad\textrm{and}\qquad \beta_{i^\star,j}=\sum_{v\in S_{i^\star,j}} f(G,v)\cdot v=f(G,x)\cdot x,\]
so $q=\beta_{i^\star,j}/\alpha_{i^\star,j}=x$.
Since $x\in[n]$, the algorithm queries $\Ver_f[G]$ on $x$, receives confirmation that $f(G,x)>0$, and returns $x$.
\end{itemize}

We conclude that a false negative can occur only if $|S_{i^\star,j}\cap X|\neq 1$ for every $j\in\{1,\dots,\xi\}$, and this happens with probability at most $\eps$.
Therefore, for every fixed query asked at any fixed time step, the probability that the data structure outputs $\bot$ despite $X\neq\emptyset$ is at most $\eps$.
This is exactly the false-negative guarantee claimed in the lemma statement.

\paragraph*{Complexity.} Upon initialization of the data structure, we have to sample and store all the sets $S_{i,j}$ for $(i,j)\in [0,L]\times [\xi]$; this takes $\Oh(n\log n\log \tfrac{1}{\eps})$ time. Also, we have to initialize all the $2(L+1)\xi=\Oh(\log n\log \tfrac{1}{\eps})$ maintained instances of $\WeiSum$ and, in the case \ref{it:general-f}, also the maintained instance of $\Ver$; this takes $\Oh(I\log n\log \tfrac{1}{\eps})$ time.
Therefore, the initialization time is $\Oh((n+I)\log n\log \tfrac{1}{\eps})$, as promised. Note also that the stored sets $S_{i,j}$ take $\Oh(n\log n\log \tfrac{1}{\eps})$ space, the maintained instances of $\WeiSum$ and $\Ver$ take $\Oh(M\log n\log \tfrac{1}{\eps})$ space, and no more space is used by the data structure throughout its lifetime. Therefore, the space complexity is $\Oh((M+n)\log n\log \tfrac{1}{\eps})$.

We are left with analyzing the amortized update and query complexity.
Note that every update to the maintained graph $G$ is relayed to all the $2(L+1)\xi=\Oh(\log n\log \tfrac{1}{\eps})$ maintained instances of $\WeiSum$ and, in the case \ref{it:general-f}, also to the single maintained instance of $\Ver$.
Therefore, the amortized cost of executing an update to $G$ is $\Oh(T\log n\log \tfrac{1}{\eps})$.
For every processed pair it performs two queries to maintained instances of $\WeiSum$ and, in the case \ref{it:general-f}, at most one query to $\Ver_f[G]$.
As for the amortized query complexity, the query algorithm inspects at most $(L+1)\xi=\Oh(\log n\log \tfrac{1}{\eps})$ pairs $(i,j)$, and for each inspected pair it issues two queries to the maintained instances of $\WeiSum$ and, in the case \ref{it:general-f}, at most one query to the maintained instance of $\Ver$.
Besides these calls, the algorithms performs only basic control logic of the loop and the arithmetic tests based on $\alpha_{i,j}$ and $\beta_{i,j}$, which amount to $\Oh(1)$ time per processed pair. Therefore, the amortized query time is $\Oh(T\log n\log \tfrac{1}{\eps})$ as well.
\end{proof}

We remark that in the proof of \cref{lem:restore-by-sampling}, it is possible to replace explicit sampling of the sets $\{S_{i,j}\colon i\in [0,L]\times [\xi]\}$ by sampling $\xi=\Theta(\log n)$ pairwise-independent hash functions $h_1,\ldots,h_\xi\colon V(G)\to [p]$, where $p$ is a prime of magnitude $\Theta(n)$. Notably, storing each such hash function $h_i$ requires only $\Oh(\log n)$ bits, instead of the $\Oh(n\log n)$ bits needed to store the sets $S_{i,1},\ldots,S_{i,\xi}$. So at the cost of a somewhat more complicated probabilistic analysis of the query algorithm, one can reduce the contribution of the sets $S_{i,j}$ to the initialization time and the space complexity from $\Oh(n\log n\log \tfrac{1}{\eps})$ to $\Oh(\log n\log \tfrac{1}{\eps})$. However, regardless of this improvement, we still need to initialize and store $\Oh(\log n\log \tfrac{1}{\eps})$ instances of $\WeiSum$, which contributes terms $\Oh(I\log n\log \tfrac{1}{\eps})$ and $\Oh(M\log n\log \tfrac{1}{\eps})$ to the initialization time and the space complexity, respectively. As we expect $I,M\geq n$ from the implementation of $\WeiSum$, these contributions anyway dominate the potential savings and the improvement has no practical impact on the guarantees asserted in the statement of  \cref{lem:restore-by-sampling}. Therefore, we omit the details.


\bigskip

Before proceeding, let us note a useful property of mappings that is shared by homomorphisms, subgraph isomorphisms, and induced subgraph isomorphisms, which essentially says that roots can be gadgeted using edges of unique colors. In the following, by a \emph{type of mappings} we mean a family of functions whose domain is the vertex set of one graph $H$ and the co-domain is the vertex set of another graph $G$.

\begin{definition}
	We say that a type of mappings $\F$ is \textit{root-gadgetable} if it satisfies the following property: Let $\rt{H}{x}$ and $\rt{G}{u}$ be $k$-colored graphs, for some $k\in \N$, each equipped with one root. Let $H'$ be the graph $H$ modified by adding pendants $x'$ and $x''$ adjacent to $x$, and  $G'$ be the graph $G$ modified by adding pendants $u', u''$ adjacent to $u$, so that edges $xx'$ and $uu'$ have color $k+1$ and edges $xx''$ and $uu''$ have color $k+2$. (Thus, $H'$ and $G'$ are $(k+2)$-colored.) Then
	\[\F(H',G') = \F(\rt{H'}{x,x',x''},\rt{G'}{u,u',u''}),\]
	and there is a natural bijection between $\F(H', G')$ and $\F(\rt{H}{x},\rt{G}{u})$ given by restricting any $\phi \in \F(H', G')$ to  $V(H)$. 
\end{definition}

\begin{lemma}\label{lem:root-gadget}
	Homomorphisms, subgraph isomorphisms, and induced subgraph isomorphisms are root-gadgetable types of mappings.
\end{lemma}
\begin{proof}
	Let $\F \in \{\Hom, \Sub, \ISub\}$. Let $\phi \in \F(H', G')$. As (induced) subgraph isomorphisms are also homomorphisms, we have that $\phi$ is a homomorphism from $H$ to $G$. As $x$ and $u$ are the only vertices simultaneously incident to edges of colors $k+1$ and $k+2$ in $H'$ and $G'$, respectively, we must have that $\phi(x)=u$, and consequently also $\phi(x') = u'$ and $\phi(x'') = u''$. It follows that $\F(H',G') = \F(\rt{H'}{x,x',x''},\rt{G}{u,u',u''})$. One can then readily verify that restricting any member of $\F(H', G')$ to $V(H)$ defines a bijection between $\F(H', G')$ and $\F(\rt{H}{x}, \rt{G}{u})$, for any $\F$ from $\{\Hom, \Sub, \ISub\}$.
\end{proof}

We now proceed to the proof of \cref{thm:dvorak-tuma-examples}, which will consist of two steps. First, we are going to design a data structure $\MapVer$ that can identify a possible image of a single vertex $x$ of $H$. Then, we are going to use it repeatedly to identify images of all the vertices of $H$ one by one.

Let us focus on designing $\MapVer$ first. The idea  is to apply \cref{lem:restore-by-sampling} for a vertex function $f$ defined as follows: $f(G, u)$ is the number of mappings from $\F$ that map $x$ to $u$. However, in order to follow this path we have to enrich \cref{thm:dvorak-tuma} with a capability of counting \emph{weighted} mappings. The following definition will be useful.

\begin{definition} \label{def:mapping-value}
	Let $H$ and $G$ be graphs and $\wei\colon V(H)\times V(G)$ be a weight function defined on pairs of vertices: one from $H$ and one from $G$. Let $\phi \colon V(H) \to V(G)$ be a mapping. We define the \textit{value} of $\phi$ (with respect to $\wei$) as follows: \[\val_\wei(\phi) = \prod_{x \in V(H)} \wei(x, \phi(x)).\] Consequently, for a type of mappings $\F$, we define $\val_\wei(\F(H, G))$ as $\sum_{\phi \in \F(H, G)} \val_\wei(\phi)$. This definition can be extended to mappings between rooted graphs in the expected way.
\end{definition}

With this definition, we state a weighted version of \cref{thm:dvorak-tuma} that will help us in applying \cref{lem:restore-by-sampling}. 

\begin{theorem} \label{lem:weighted-dvorak}
	Let $k\in \N$, $H$ be a fixed $k$-colored graph, $\G$ be a graph class of bounded expansion, and $\F$ be either $\Hom$, $\Sub$, or $\ISub$. Let $G$ be a dynamic $k$-colored graph on $n$ vertices that belongs to $\G$ at all times, and let $\wei \colon V(H) \times V(G) \to \Z$ be a fixed weight function. Then there is a data structure $\WeiMaps_{\F, H, \wei}[G]$ that is able to determine $\val_\wei(\F(H, G))$ after each update. The amortized update time is $\Oh_{\G, H, k}(\log^h n)$, where $h = {|H| \choose 2} - 1$, while the initialization time and the space complexity is~$\Oh_{\G, H, k}(n)$. 
\end{theorem}

\cref{lem:weighted-dvorak} can be proved by a rather straightforward modification of the proof of \cref{thm:dvorak-tuma} due to Dvo\v{r}\'ak and T\r{u}ma.
On a high level, \cref{thm:dvorak-tuma} is proven by means of consecutive reductions:
\begin{itemize}
	\item counting induced subgraph isomorphisms is reduced to counting subgraph isomorphisms;
	\item counting subgraph isomorphisms is reduced to counting homomorphisms; and
	\item counting homomorphisms is reduced to counting homomorphisms in a special case, where both $H$ and $G$ are additionally directed, but $H$ is a so-called \textit{elder graph}.
\end{itemize}
The first two reductions use the Inclusion--Exclusion Principle, the third reduction directly expresses general homomorphism counts as sums over homomorphism counts from elder graphs, while homomorphism counts from elder graphs are maintained by maintaining suitable dynamic programming tables.
As can be expected, it is a routine task to incorporate weighted counting into all these steps, in particular into the dynamic programming. For the sake of completeness, in \cref{sec:app-weighted-dvorak} we provide a detailed exposition of the modifications that need to be applied to the proof of \cref{thm:dvorak-tuma} in order to derive  \cref{lem:weighted-dvorak}.

With \cref{lem:weighted-dvorak} in hand, we are now able to design the promised $\MapVer$ data structure:

\begin{lemma} \label{lem:map-vertex}
	Let $\eps>0$, $k\in \N$, $\rt{H}{x}$ be a fixed $k$-colored graph with one root $x$, $\G$ be a graph class of bounded expansion, and $\F$ be either $\Hom$, $\Sub$, or $\ISub$. 
	Let $G$ be a dynamic $k$-colored graph on $n$ vertices that belongs to $\G$ at all times.
	Then, there is a randomized data structure $\MapVer_{\F, \rt{H}{x}, k,\eps}[G]$ which after every update is able to report that $\F(H,G)=\emptyset$, or provide a vertex $u$ of $G$ such that $\F(\rt{H}{x},\rt{G}{u})$ is non-empty. The amortized update time is $\Oh_{\G, \rt{H}{x}, k}(\log^{\Oh_{\rt{H}{x}}(1)} n \log{\frac{1}{\eps}})$, the initialization time is $\Oh_{\G, \rt{H}{x}, k}(n \log n \log {\tfrac{1}{\eps}})$, and the space complexity is $\Oh_{\G, \rt{H}{x}, k}(n \log n \log {\tfrac{1}{\eps}})$. The data structure never provides false positives, but may fail to provide an example vertex $u$ with probability at most $\eps$, against an oblivious adversary.
\end{lemma}
\begin{proof}
	Working towards an application of \cref{lem:restore-by-sampling}, we are going to construct the required data structures $\WeiSum_{f, \wei}[G]$ and $\Ver_f[G]$ such that the resulting data structure $\RetVer_{f,\eps}[G]$ fits the requirements for $\MapVer_{\F, \rt{H}{x}, k,\eps}[G]$.
		
	Let $f$ be a vertex function defined as follows: for a graph $G$ and a vertex $u$ of $G$, 
	\[f(G,u)\coloneqq |\F(\rt{H}{x},\rt{G}{u})|.\] 
	Consider any weight function $\wei\colon V(G)\to \Z$ and define a weight function $\wei' \colon V(H) \times V(G) \to \Z$ in the following way: for $(y,v)\in V(H)\times V(G)$,
	\[
	\wei'(y,v) =
	\begin{cases}
		\wei(v) & \text{if } y = x, \\
		1  & \text{otherwise.}
	\end{cases}
	\]
	Note that for $\wei'$ defined in this way, we have  \[\val_{\wei'}(\F(H, G)) =\sum_{\phi\in \F(H,G)} \wei(\phi(x)) = \sum_{u\in V(G)} |\F(\rt{H}{x},\rt{G}{u})|\cdot \wei(u)= \sum_{u \in V(G)} f(G, u)\cdot \wei(u).\]
	Hence, we can set $\WeiSum_{f, \wei}[G] \coloneqq \WeiMaps_{\F, H, \wei'}[G]$, provided by \cref{lem:weighted-dvorak}.
	
	As for designing $\Ver_f[G]$, let us first construct $H'$ from $H$ by adding two pendants $x',x''$ connected to $x$ by edges $xx'$ and $xx''$ of colors $k+1$ and $k+2$, respectively. Further, for any $u \in V(G)$, let $G_u$ be obtained from $G$ by adding two pendants $u',u''$ connected to $u$ by edges $uu'$ and $uu''$ of colors $k+1$ and $k+2$, respectively. As $\F$ is a root-gadgetable type of mappings by \cref{lem:root-gadget}, we have that $|\F(\rt{H}{x}, \rt{G}{u})|=|\F(H',G_u)|$. Hence, deciding whether $f(G, u)>0$ is equivalent to deciding whether $|\F(H', G_u)| > 0$. Consequently, the data structure $\Ver_f[G]$ can be implemented in the following way. Let us initialize the data structure $\UnMaps_{\F, H'}$, provided by \cref{thm:dvorak-tuma}, on the vertex set $V(G) \cup \{p_1, p_2\}$, where $p_1$ and $p_2$ are fresh vertices that will normally stay isolated. All updates to $G$ are relayed to that data structure. Whenever a query about a vertex $u$ is issued to $\Ver_f(G)$, we temporarily add edges $up_1$ and $up_2$ with colors $k+1$ and $k+2$ to $\UnMaps_{\F, H'}$, in order to transform $G$ to $G_u$, and answer that $f(G, u)>0$ if and only if $|\F(H', G_u)| > 0$. After answering the query, the edges $up_1$ and $up_2$ are removed, thus making $p_1$ and $p_2$ isolated again.
	
	If $\G'$ is the class of all the graphs obtainable by repeatedly adding pendants to a graph from $\G$, then $G_u \in \G'$ and $\G'$ has bounded expansion by \cref{obs:pendants-bnd-exp}. Hence, the amortized time complexity of both handling an update to $G$ or answering a single query by $\Ver_f[G]$ is $\Oh_{\G', H', k+2}(\log^{\Oh_{H'}(1)} |G'|) = \Oh_{\G, \rt{H}{x}, k}(\log^{\Oh_{\rt{H}{x}}(1)} n)$. Similarly, the space complexity and the initialization time is $\Oh_{\G, H, k}(n)$. 
	
	We may now
	apply \cref{lem:restore-by-sampling} by supplying it with functions $f, \wei$, data structures $\WeiSum_{f, \wei}[G]$ and $\Ver_f[G]$, and the error parameter $\eps$. This yields the data structure $\RetVer_{f,\eps}[G]$ that may serve as $\MapVer_{\F, \rt{H}{x}, k,\eps}[G]$ with the desired properties.
\end{proof}

%

We are now ready to conclude the proof of the \cref{thm:dvorak-tuma-examples}.

\begin{proof} [Proof of \cref{thm:dvorak-tuma-examples}]
	Let us enumerate $V(H)$ as $\{x_1, \ldots, x_{|H|}\}$ and create a sequence of graphs $H_0, \ldots, H_{|H|}$, where $H=H_0$ and $H_{i}$ is created from $H_{i-1}$ by adding pendants $p_{2i-1}$ and $p_{2i}$, connected to $v_i$ by edges of colors $k+2i-1$ and $k+2i$, respectively. Let also $\eps'\coloneqq \tfrac{\eps}{|H|}$ and $\G'$ be the class of graphs obtainable from graphs from $\G$ by repeatedly adding pendants. By \cref{obs:pendants-bnd-exp}, $\G'$ has bounded expansion.
	
	We maintain data structures $\MapVer_i \coloneqq \MapVer_{\F, \rt{H_i}{x_{i+1}}, k + 2i,\eps'}$ for $i=0, \ldots, |H|-1$, provided by \cref{lem:map-vertex} for the class $\G'$. Each data structure $\MapVer_i$ is initialized on $n+2i$ vertices: the $n$ vertices of $G$ and $2i$ additional vertices that will normally stay isolated. Any update to $G$ is passed to all of the data structures $\MapVer_0, \ldots, \MapVer_{|H|-1}$, so that each data structure $\MapVer_i$ maintains a graph consisting of $G$ plus $2i$ isolated vertices. 
	
	The data structure $\MapVer_i$ will be used to identify the image of $x_{i+1}$ after already fixing images of $x_1, \ldots, x_i$.
	That is, let us assume that we have already identified vertices $u_1, \ldots, u_i \in V(G)$ such that there exists a mapping $\phi \in \F(H, G)$ with $\phi(x_j)=u_j$ for all $j \le i$. Let $G_0, G_1, \ldots, G_i$ be the sequence of graphs such that $G=G_0$ and $G_i$ is created from $G_{i-1}$ by adding pendants $q_{2i-1}$ and $q_{2i}$ connected to $u_i$ by edges of  colors $k+2i-1$ and $k+2i$, respectively. By using the $2i$ isolated vertices in the graph stored in $\MapVer_i$ and temporarily adding $2i$ edges incident to them, we can turn the graph stored in $\MapVer_i$ into $G_i$. Since $\F$ is root-gadgetable (by \cref{lem:root-gadget}), a straightforward induction on $i$ shows that there is a natural bijection between mappings $\phi' \in \F(H_i, G_i)$ and mappings $\phi \in \F(H, G)$ satisfying $\phi(x_j)=u_j$ for all $j\leq i$, defined by $\phi'|_{V(H)} = \phi$. As we assumed that there exists $\phi \in \F(H, G)$ such that $\phi(x_i) = u_i$, we conclude that $\F(H_i, G_i)$ is nonempty. Hence, $\MapVer_i$ is able to provide a vertex $u_{i+1}$ such that there exists $\phi' \in \F(H_i, G_i)$ with $\phi'(x_{i+1})=u_{i+1}$, or conclude that no such vertex exists. Moreover, if such $u_{i+1}$ is found, then it has to be the case that $u_{i+1} \in V(G)$ (that is, it is not possible that $u_{i+1}=q_j$ for some $j\le i$), so by restricting $\phi'$ to $V(H)$ we get a mapping $\phi$ such that $\phi(x_j)=u_j$ for all $j\le i+1$. By repeating this reasoning for $i=0, \ldots, |H|-1$, we get a full mapping $\phi \in \F(H, G)$, as desired; or a conclusion at some step that no such mapping exists.
	
	Note that we have $|H_i|\leq |H|+2i\leq 3|H|$ and $|G_i| = n+2i\leq n + 2|H|$. Therefore, the promised guarantees on the amortized update time, initialization time, and space complexity follow directly from the guarantees for the data structures $\MapVer_i$, provided by \cref{lem:map-vertex}. Clearly the constructed data structure provides no false positives, and since we make $|H|$ queries to the $\MapVer_i$ data structures in total, the probability of a false negative is bounded by $|H|\eps'=\eps$.
\end{proof}

Finally, we remark that since $|\F(H, G)| \le |G|^{|H|}$ and $\val_\wei(\phi) \le |G|$ for any $\F$, $\wei$, and $\phi$ considered in the proofs presented above, all the integers involved in the computations are of bitlength $\Oh(|H|\log n)$, and hence they fit into $\Oh(|H|)$ words in the RAM model. Therefore, all the arithmetic operations can be performed in time $\Oh_H(1)$.

\subsection{Extension to rooted relational structures}

In the next section we will need a variant of \cref{thm:dvorak-tuma-examples} for rooted, colored directed graphs. While it is easy to add those features to the proof of \cref{thm:dvorak-tuma-examples}, for the sake of providing a robust citation interface for future works we choose to take an extra mile and lift \cref{lem:weighted-dvorak,thm:dvorak-tuma-examples} to the setting of rooted relational structures. The proof is a reduction by means of rather simple gadgeteering, similar to the arguments proposed by Dvo\v{r}\'ak and T\r{u}ma in~\cite[Section~4]{DvorakTumaArxiv}, where they argued that \cref{thm:dvorak-tuma} lifts to this setting.

Let us clarify the model. For a signature $\Sigma$, by a \emph{dynamic $\Sigma$-structure} we mean a $\Sigma$-structure $\Af$ whose universe stays fixed, but which is modified by updates of the following kind: given a relation $R\in \Sigma$ and a tuple $\tup u\in \Af^{\ar(R)}$, add/remove $\tup u$ from $R^\Af$. Data structures for dynamic $\Sigma$-structures are defined analogously to the graph setting.


\begin{theorem}\label{thm:homcounts-relational}
	Let $\Cc$ be a graph class of bounded expansion, $\Sigma$ be a signature, $\rt{\Bf}{\tup x}$ be a fixed $\Sigma$-structure with roots $\tup x$, and $\F$ be either $\Hom$, $\Sub$, or $\ISub$. Let $\Af$ be a dynamic $\Sigma$-structure with a universe of size $n$ whose Gaifman graph belongs at all times to $\Cc$. Further, let $\wei\colon U(\Bf)\times U(\Af)\to \Z$ be a fixed weight function. Then there is a deterministic data structure $\WeiMaps_{\F, \rt{\Bf}{\tup x},\wei}[\Af]$ for $\Af$ that supports queries:
	\begin{itemize}
		\item Given a tuple $\tup u\in U(\Af)^{|\tup x|}$, compute $\val_\wei(\F(\rt{\Bf}{\tup x},\rt{\Af}{\tup u}))$.
	\end{itemize}
	Moreover, for every $\eps>0$, there is a randomized data structure $\MapEx_{\F, \rt{\Bf}{\tup x},\eps}[\Af]$ for $\Af$ that supports queries:
	\begin{itemize}
		\item Given a tuple $\tup u\in U(\Af)^{|\tup x|}$, return any mapping $\varphi\in\F(\rt{\Bf}{\tup x},\rt{\Af}{\tup u})$, or $\bot$ if this set is empty.
	\end{itemize}
	The amortized update and query time of the data structure $\WeiMaps_{\F, \rt{\Bf}{\tup x},\wei}[\Af]$ is $\Oh_{\Cc}(\log^{\Oh_{\rt{\Bf}{\tup x},\Sigma}(1)} \log n)$, while the initialization time and the space complexity is $\Oh_{\Cc,\rt{\Bf}{\tup x},\Sigma}(n)$. 
	The amortized update and query time of the data structure $\MapEx_{\F, \rt{\Bf}{\tup x},\eps}[\Af]$ is $\Oh_{\Cc}(\log^{\Oh_{\rt{\Bf}{\tup x},\Sigma}(1)} \log n \log \tfrac{1}{\eps})$, while the initialization time and the space complexity is $\Oh_{\Cc,\rt{\Bf}{\tup x},\Sigma}(n\log n\log \tfrac{1}{\eps})$. The data structure $\WeiMaps_{\F, \rt{\Bf}{\tup x},\wei}[\Af]$ always provides correct answers to the queries. The data structure $\MapEx_{\F, \rt{\Bf}{\tup x},\eps}[\Af]$ never outputs false positives, but may output a false negative with probability at most $\eps$ against an oblivious adversary.
\end{theorem}

\begin{proof}
	We first argue that we may focus on the case where there are no roots ($\tup x$ and $\tup u$ are empty tuples). Indeed, if $|\tup x|=k$, then we may add $k$ new unary relations $Q_1,\ldots,Q_k$ to the signature and mark the $i$th root $x_i$ of $\rt{\Bf}{\tup x}$ using $Q_i$. Upon query about a tuple $\tup u$ in $\Af$ (either in $\WeiMaps_{\F, \rt{\Bf}{\tup x},\wei}[\Af]$ or in $\MapEx_{\F, \rt{\Bf}{\tup x},\eps}[\Af]$), we may temporarily mark each $u_i$ using $Q_i$, answer the query for unrooted but $\{Q_1,\ldots,Q_k\}$-decorated $\Bf$ and $\Af$, and restore the original state of $\Af$ by unmarking back the elements~$u_i$. Thus, from now on we focus on the case where there are no roots.
	
	By adding a unary relation $R$, such that every element of $\Af$ is a singleton tuple
	in $R$, and every element of $\Bf$ is a singleton tuple in $R$, we may assume that every element of $\Bf$ is in some relation.
	
	We first focus on designing the $\WeiMaps_{\F,\Bf,\wei}[\Af]$ data structure.
	Let $C=\Sigma\uplus 2^{[q]}$, where $q\coloneqq \max_{R\in \Sigma} \ar(R)$.
	Observe, that $q \in \Oh(\Cc)$ is a constant as $\Cc$ has bounded expansion and
	every non-empty $q$-ary relation induces a $q$-clique in the Gaifman graph.
	Let $M_\Bf$ be a $C$-colored undirected graph defined in the following way. First, we include the universe $U(\Bf)$ in the vertex set of $M_{\Bf}$. Then, for each $R\in \Sigma$ and $\tup v\in R^\Bf$, we create two new vertices $t^{R,1}_{\tup v}, t^{R,2}_{\tup v} \in V(M_\Bf)$ and add to $M_\Bf$ the following edges:
	\begin{itemize}
		\item edge $t^{R,1}_{\tup v}t^{R,2}_{\tup v}$ of color $R$; and
		\item for each vertex $w\in U(\Bf)$ appearing in $\tup v$, an edge $wt^{R,1}_{\tup v}$ of color $\{j\in [\ar(R)]\mid w=v_i\}$\footnote{In \cite[Section~4]{DvorakTumaArxiv} all edges of the form $wt^{R,1}_{\tup v}$ are of the same color. However, it is necessary to assign them different colors in order not lose information about the position(s) in the tuple $\tup v$ occupied by the vertex $w$.}.
	\end{itemize}
	We define the graph $M_\Af$ analogously. Note the following (an analogous statement is also proved in~\cite{DvorakTumaArxiv}).
	
	\begin{claim}\label{cl:still-bnd}
		There is a class of bounded expansion $\Cc'$, depending only on $\Cc$ and $\Sigma$, such that
		$M_\Af\in \Cc'$.
	\end{claim}
	\begin{claimproof}
		Let $G$ be the Gaifman graph of $\Af$. Since $G\in \Cc$, $G$ is $d$-degenerate where $d\coloneqq \nabla_0(\Cc)$, hence there exists an ordering $\preceq$ of the vertices of $G$ such that every vertex $u$ of $G$ has at most $d$ neighbors that are smaller in $\preceq$. (Such an ordering can be obtained by repeatedly removing from $G$ a vertex of the smallest degree and ordering the vertices by the reverse order of removal.) For any $R\in \Sigma$ and any tuple $\tup v\in R^{\Af}$, let $\mu(\tup v)$ be the $\preceq$-maximal vertex featured in $\tup v$. Note that for every $u\in V(G)$ we have $|\mu^{-1}(u)|\leq |\Sigma|\cdot (d+1)^q$,
		for every tuple $\tup v\in \mu^{-1}(u)$ is entirely contained in the set consisting $u$ and its $\preceq$-smaller neighbors --- which is of size at most $d+1$.
		Hence, it is easy to verify that the map $\eta$ defined~as
		\begin{align*}
			\eta(u)& \coloneqq \{u\} && \textrm{for }u\in U(\Af);\textrm{ and}\\
			\eta(t^{R,1}_\tup v)\coloneqq \eta(t^{R,2}_\tup v)& \coloneqq \{\mu(\tup v)\} && \textrm{for }\tup v\in R^\Af, R\in \Sigma;
		\end{align*} 
		is a congestion-$(2|\Sigma|\cdot (d+1)^q+1)$ depth-$0$ model of $M_\Af$ in $G$. It follows that $M_\Af\in \Minors^{2|\Sigma|\cdot (d+1)^q+1,0}(\Cc)$, which is a class of bounded expansion by \cref{thm:cong-minors}.
	\end{claimproof}
	
	It is now straightforward to see that subgraph isomorphisms from $M_\Bf$ to $M_\Af$ are in one-to-one correspondence to substructure isomorphisms from $\Bf$ to $\Af$. 
	
	\begin{claim}\label{cl:subgraph-works}
		Every substructure isomorphism $\varphi\in \Sub(\Bf,\Af)$ can be uniquely extended to a subgraph isomorphism $\varphi'\in \Sub(M_\Bf,M_\Af)$. Conversely, every subgraph isomorphism $\varphi'\in \Sub(M_\Bf,M_\Af)$ becomes a substructure isomorphism after restricting the domain to $U(\Bf)$.
	\end{claim}
	
	So if we extend the function $\wei \colon U(\Bf) \times U(\Af) \to \Z$ to a function $\wei' \colon V(M_H) \times V(M_G) 
	\to \Z$ by setting $\wei'(x, u) = \wei(x, u)$ for all $(x, u) \in U(\Bf) \times U(\Af)$, and $\wei'(x, u)=1$ for all other pairs $(x,u)$, then
	\[\val_{\wei'}(\Sub(M_\Bf,M_\Af))=\val_\wei(\Sub(\Bf,\Af)).\]
	Consequently, if we maintain the data structure $\WeiMaps_{\Sub, M_\Bf, \wei}(M_\Af)$ provided by \cref{lem:weighted-dvorak} for the class $\Cc'$ given by \cref{cl:still-bnd}, then this data structure can be directly used to answer queries to $\WeiMaps_{\Sub,\Bf,\wei}[\Af]$. Noting that one update to $\Af$ corresponds to at most $q+1$ updates to $M_\Af$, hence the claimed complexity guarantees about $\WeiMaps_{\Sub,\Bf,\wei}[\Af]$ follow directly from the guarantees provided by \cref{lem:weighted-dvorak}. 
	
	We note that each update to $\Af$, apart from altering the edge set of $M_\Af$, also alters its vertex set by either adding or removing two vertices specific to the toggled tuple, so it seems as if a method for creating or removing isolated vertices should be required. However, as argued in the \cref{cl:still-bnd}, the number of tuples in $\Af$ can be upper bounded by $n|\Sigma|\cdot (d+1)^q$, so we can deal with that by creating a stash of $2n|\Sigma|\cdot (d+1)^q$ isolated vertices during the initialization. Whenever we insert a tuple $\tup v$ to $R^\Af$, we take two vertices from the stash and name them $t^{R,1}_\tup v$ and $t^{R,2}_\tup v$ and whenever we remove $\tup v \in R^\Af$, we return these to the stash.
	
	
	
	While \cref{cl:subgraph-works} holds for subgraph isomorphisms, it unfortunately fails for induced subgraph isomorphisms and homomorphisms. Fortunately, it is well-known that weighted homomorphism counts and weighted induced subgraph isomorphism counts can be expressed as linear combinations of weighted subgraph isomorphism counts. Let us explain this argument in more detail, using the argumentation from \cite{DvorakTumaArxiv}, which is repeated in \cref{sec:app-weighted-dvorak}. In the following we assume that the reader is familiar with the material from this section.
	
	
	Let us tackle the case of homomorphisms first. Similarly as in \cref{def:projections} for any partition $\Pp$ of $U(\Bf)$, we may define the quotient $\Sigma$-structure $\Bf_\Pp$ by identifying every part $P\in\Pp$ into a single element $z_P$, and pushing the relations in $\Bf$ naturally to relations in $\Bf_\Pp$ through this identification mapping. We may also naturally define the quotient weight function $\wei_\Pp\colon U(\Bf)\times U(\Af)$ by setting $\wei_\Pp(z_P,u)\coloneqq \prod_{x\in P} \wei(x,u)$, for each $P\in \Pp$ and $u\in U(\Af)$. It then follows that 
	\[\val_\wei(\Hom(\Bf, \Af)) = \sum_{\Pp\colon \textrm{partition of }U(\Bf)} \val_{\wei_{\Pp}}(\Sub(\Bf_\Pp, \Af)).\]
	Hence, maintaining the data structures $\WeiMaps_{\Sub, \Bf_\Pp, \wei_\Pp}(\Af)$ for all partitions $\Pp$ pf $U(\Bf)$ allows us to answer the queries to $\WeiMaps_{\Hom, \Bf, \wei}[\Af]$ within the claimed complexity bounds.
	
	For the case of induced substructure isomorphisms we use the same Inclusion--Exclusion argument as in \cref{lem:weighted-isub-to-sub}, but instead of iterating over all supergraphs, it suffices to iterate over all superstructures, that is, all $\Sigma$-structures $\Bf'$ such that $U(\Bf') = U(\Bf)$ and $R^\Bf \subseteq R^{\Bf'}$ for all $R\in \Sigma$. Specifically, we have
	\[\val_\wei(\ISub(H, G)) = \sum_{\Bf'\colon \textrm{superstructure of } \Bf} (-1)^{\|\Bf'\| - \|\Bf\|}\cdot \val_\wei(\Sub(\Bf', \Af)),\]  where $\|\Bf\| = \sum_{R\in \Sigma} |R^\Bf|$, and similarly for $\Bf'$. Since all such superstructures $\Bf'$ are still of size $\Oh_{\Bf,\Sigma}(1)$ and there is $\Oh_{\Bf,\Sigma}(1)$ of them, maintaining the data structures $\WeiMaps_{\Sub, \Bf', \wei}(\Af)$ allows us to answer the queries to $\WeiMaps_{\ISub, \Bf, \wei}[\Af]$ within the claimed complexity bounds.
	
	Having designed a suitable data structure $\WeiMaps_{\F, \Bf, \wei}[\Af]$ for each $\F\in \{\Hom,\Sub,\ISub\}$, we may repeat the argumentation of \cref{lem:restore-by-sampling,lem:map-vertex,thm:dvorak-tuma-examples} to design also the data structure $\MapEx_{\F, \Bf,\eps}[\Af]$. The reasoning lifts essentially verbatim; we leave the details to the reader.
\end{proof}

	


\newcommand{\Labs}{\Lambda}
\newcommand{\lab}{\lambda}
\newcommand{\length}{\mathsf{len}}

\newcommand{\Mm}{\mathcal{M}}

\section{Implementation in classes of bounded expansion}\label{sec:implementation}

Here we implement the $\near{S,r}$ and $\far{S,r}$ queries for an arbitrary distance parameter $r$ in classes of bounded expansion. 
For that, we describe in this section the corresponding data structures $\nearVertexDS{G}{r,k,\eps}$
and $\farVertexDS{G}{r,k,\eps}$, generalized appropriately to be used for progressive exploration
algorithms for both $\drds{r }$ and $\dris{r }$.

\subsection{Detecting near vertices}

We start with the implementation of the $\nearVertexDS{G}{r,k,\eps}$ data structure, which will be an easy application of the toolbox presented in \cref{sec:dvorak}, particularly \cref{thm:homcounts-relational}.
As described in Section~\ref{sec:overview}, we will handle this by, upon every query,
creating a collection of graphs as in Figure~\ref{fig:homs} and finding their homomorphic rooted images in $G$.

\begin{theorem}\label{thm:NearVertex-new}
	Fix a graph class $\Cc$ of bounded expansion, $k,r\in \N$, and $\eps>0$.
	Let $G$ be a dynamic graph on $n$ vertices that belongs at all times to $\Cc$. Then there exists a randomized data structure $\nearVertexDS{G}{r,k,\eps}$ that provides access to the following query:
	\begin{itemize}
		\item $\near{S,\rho}$: For a given set of vertices $S\subseteq V(G)$ with $|S|\leq k$ and a function $\rho\colon S\to [0,r]$, return a vertex $v$ such that $\dist_G(s,v)\leq \rho(s)$ for every $s\in S$, or $\bot$ if no such vertex exists.
	\end{itemize}
	Every answer to the query is correct with probability at least $1-\eps$ against an oblivious adversary.
	The amortized time complexity of updates and queries is $\log^{\Oh_{\Cc,r,k}(1)} n\cdot \log \tfrac{1}{\eps}$. The data structure can be initialized for an edgeless $G$ in time $\Oh_{\Cc,r,k}(n\log n\log \tfrac{1}{\eps})$, and uses $\Oh_{\Cc,r,k}(n\log n\log \tfrac{1}{\eps})$ space at all times.
\end{theorem}
\begin{proof}
	Consider a query $\near{S,\rho}$. Enumerate all the vertices of $S$ as a $k$-tuple $\tup s$, possibly repeating some of them if necessary. Define $\rho'\colon [k]\to [0,r]$ by setting $\rho'(i)\coloneqq \rho(s_i)$, for all $i\in [k]$. 
	
	Let $\tup z$ be a tuple consisting of $k$ distinct vertices. 
	Consider a family of $\Xi_{\rho'}$ of all the rooted graphs $\rt{H}{\tup z,y}$ that can be obtained as follows (see Figure~\ref{fig:homs}):
	\begin{itemize}
		\item for every $i\in [k]$, attach to $z_i$ a path of some length between $0$ and $\rho'(i)$, and
		\item fuse all the other endpoints of those paths into a single vertex and call it $y$.
	\end{itemize}
	We also let $\Xi_r\coloneqq \Xi_{\rho_0}$ where $\rho_0(i)=r$ for all $i\in [k]$. 
	Note that $\Xi_{\rho'}\subseteq \Xi_r$ for every $\rho'\colon [i]\to [0,r]$, while $|\Xi_r|=(r+1)^{k}\leq \Oh_{r,k}(1)$ and $|H|\leq 1+(r-1)k\leq \Oh_{r,k}(1)$ for all $\rt{H}{\tup z,y}\in \Xi_r$.
	The following is clear.
	
	\begin{claim}\label{cl:distExample}
		For any vertex $v$ of $G$, the following conditions are equivalent:
		\begin{itemize}
			\item $\dist_G(s,v)\leq \rho(s)$ for all $s\in S$; and
			\item there exists $\rt{H}{\tup z,y}\in \Xi_{\rho'}$ such that $\Hom(\rt{H}{\tup z,y},\rt{G}{\tup s,v})\neq\emptyset$.
		\end{itemize}
	\end{claim}

	Let $\eps'\coloneqq \tfrac{\eps}{|\Xi_r|}\geq \tfrac{\eps}{\Oh_{r,k}(1)}$.
	For every rooted graph $\rt{H}{\tup z,y}\in \Xi_{r}$ we may apply \cref{thm:homcounts-relational} to construct the data structure $\mathsf{MappingExample}_{\Hom, \rt{H}{\tup z},\eps'}[G]$; note that here, we remove $y$ from the roots.
	Then, for every $\rt{H}{\tup z,y}\in \Xi_{\rho'}$, we ask the data structure $\mathsf{MappingExample}_{\Hom, \rt{H}{\tup z},\eps'}[G]$ for an example homomorphism $\varphi_{\rt{H}{\tup z,y}}\in \Hom(\rt{H}{\tup z},\rt{G}{\tup s})$. If any such homomorphism exists, then by \cref{cl:distExample} we have that $\dist_G(s,v)\leq \rho(s)$ for all $s\in S$, where $v\coloneqq \varphi_{\rt{H}{\tup z,y}}(y)$; so $v$ can be reported. And if no such homomorphism exists, for every $\rt{H}{\tup z,y}\in \Xi_{\rho'}$, then by \cref{cl:distExample} we may conclude that there is no vertex $v\in V(G)$ such that $\dist_G(s,v)\leq \rho(s)$ for all $s\in S$. By the union bound, the probability that any of the $|\Xi_{\rho}|\leq |\Xi_r|$ calls to  $\mathsf{MappingExample}_{\Hom, \rt{H}{\tup z},\eps'}[G]$ returns an incorrect answer is upper bounded by $\eps'\cdot |\Xi_\rho|\leq \eps$; this bounds the error probability of the algorithm.

	Since $|\Xi_{r}|\leq \Oh_{r,k}(1)$ and $|H|\leq \Oh_{r,k}(1)$ for all $\rt{H}{\tup z,y}\in \Xi_r$, the promised bounds on the amortized update and query time, initialization time, and space complexity follow immediately from the guarantees provided by \cref{thm:homcounts-relational}. 
\end{proof}

\subsection{Maintaining augmentations}\label{sec:augmentations}

Before proceeding to the specific description of the $\farVertexDS{G}{r,k,\eps}$ data structure, we explain how we dynamically maintain an augmented version of our graph.
Recall from Section~\ref{sec:overview}, that in order to
implement $\farVertexDS{G}{r,k,\eps}$ we need to be able to count patterns
of short paths in $G$, and we achieve this by maintaining the decorated orientation
of a supergraph of $G$. In this section we introduce the machinery that allows us to efficiently maintain
appropriately decorated orientations.
This is a key technical ingredient in the approach, which was already
exploited by Dvo\v{r}\'ak and T\r{u}ma~\cite{DvorakT13} in a similar
manner. Here we repeat their reasoning and adjust it to our setting.


The approach is based on maintaining an iterated fraternal augmentation of $G$, which is still sparse thanks to \cref{lem:fraternal}. The edges of this augmentation are decorated by some information memorized during the augmentation process. Formally, we will adhere to the following definitions.



\begin{definition}
	Let $\Lambda$ be a finite set of labels and $r\in \N$.
	A \emph{$(\Lambda,r)$-decorated graph} is an oriented graph $D$ together with a \emph{labeling function} $\lab_D\colon E(D)\to \Lambda$ and a \emph{length function} $\length_D\colon E(D)\to [r]$. The subscript can be omitted if $D$ is clear from the context. When speaking about just an \emph{$r$-decorated graph} we assume only the availability of the length function, and we allow $D$ to be also an undirected graph.
	
	We say that a $(\Lambda,r)$-decorated graph $D$ is \emph{unambiguous} if for every vertex $u$ of $D$, the edges of $D$ with tail $u$ receive pairwise different labels under $\lab_D$. Note that this in particular implies that the maximum outdegree in $D$ is at most $|\Lambda|$.
	
	For an undirected graph $G$, we say that a $(\Lambda,r)$-decorated graph $D$ is a \emph{faithful $(\Lambda,r)$-augmentation} of $G$ if the following conditions are satisfied:
	\begin{enumerate}[label={(A\arabic*)}]
		\item $V(D)=V(G)$ and the undirected graph underlying $D$ is a supergraph of $G$.
		\item $D$ is unambiguous.
		\item\label{c:short} For every edge $(u,v)$ of $D$, we have $\dist_G(u,v)\leq \length_D(u,v)$.
		\item\label{c:center} For every path $P=(u_0,u_1,\ldots,u_{r'})$ in $G$ of length $r'\leq r$, there exist indices 
		\[0=\alpha_0< \alpha_1< \ldots <\alpha_{s-1}<\alpha_s = \beta_t< \beta_{t-1}< \ldots< \beta_1< \beta_0=r',\]
		for some $s,t\in \N$, such that 
		\begin{itemize}
			\item for every $i\in [s]$, the edge $(u_{\alpha_{i-1}},u_{\alpha_i})$ is present in $D$ and $\length_D(u_{\alpha_{i-1}},u_{\alpha_i})\leq \alpha_i-\alpha_{i-1}$; and 
			\item for every $j\in [t]$, the edge $(u_{\beta_{j-1}},u_{\beta_j})$ is present in $D$ and $\length_D(u_{\beta_{j-1}},u_{\beta_j})\leq \beta_{j-1}-\beta_{j}$. 
		\end{itemize}
		The subgraph of $D$ consisting of the edges  $(u_{\alpha_{i-1}},u_{\alpha_i})$, for $i\in [s]$, and  $(u_{\beta_{j-1}},u_{\beta_j})$, for $j\in [t]$, will be called a \emph{shortcut} of $P$ (see Figure~\ref{fig:shortcut} for the reference).
	\end{enumerate}
\end{definition}

Homomorphisms of $(\Lambda,r)$-decorated graphs are defined naturally, as for general relational structures: for $\varphi$ to be a homomorphism we require that for every edge $(u,v)$ of the source graph, there is an edge $(\varphi(u),\varphi(v))$ of the target graph with exactly the same label and length.

We aim to show that in any class $\Cc$ of bounded expansion, faithful augmentations of any fixed depth can be maintained in polylogarithmic amortized update time. The proof is based on \cite[Theorem 4]{DvorakTumaArxiv}.

In preparation for the description of the data structure maintaining the faithful augmentation, let us describe a sequence of graph classes of bounded expansion. Let $\Cc$ be a graph class of bounded expansion. We define a sequence of graph classes $\Cc_1, \ldots, \Cc_r$ inductively in the following way:
\begin{itemize}
	\item $\Cc_1 \coloneqq \Cc$; and
	\item for $i\in [r-1]$, we set $\Cc_{i+1} \coloneqq \Fra(\Cc_i, 4 \cdot i \cdot d_i)$ for some chosen constants $d_i \geq \nabla_0(\Cc_i)$.
\end{itemize}
Note that a straightforward induction using \cref{lem:fraternal} shows that the consecutive classes $\Cc_1,\Cc_2,\Cc_3,\ldots$ have bounded expansion, hence $\nabla_0(\Cc_i)$ is finite and can be used to define $\Cc_{i+1}$.

Our data structure then needs access to parameters $d_1,\ldots , d_n$
to maintain orientations on $r$ levels of fraternal augmentations using
the data structure of Brodal and Fagerberg (Theorem~\ref{thm:bf}).
We note that given parameters
$\nabla_0(\Cc),\nabla_1(\Cc),\nabla_2(\Cc),\ldots$, we can compute the
upper bounds $d_1\le d_2\leq \ldots\le d_r$ on
$\nabla_0(\Cc_1)\le \nabla_0(\Cc_2)\leq \ldots \leq \nabla_0(\Cc_r)$.
This is because
the proof of \cref{lem:fraternal} is effective and
gives computable upper bounds on the parameters $\nabla_r(\Fra(\Cc_i,d))$ based on $d$ and the parameters $\nabla_r(\Cc)$.

With that, let us state the main lemma of this section.

\begin{lemma}\label{lem:dyn-augmentation}
	Fix a graph class $\Cc$ of bounded expansion and $r\in \N$. Then there exists a finite set of labels $\Lambda$ such that every graph $G\in \Cc$ has a faithful $(\Lambda,r)$-augmentation $D$. 
	
	Also, there is a data structure that for a dynamic $n$-vertex graph $G\in \Cc$, maintains its faithful $(\Lambda,r)$-orientation $D$ with amortized update time $\Oh_{\Cc,r}(\log^{r-1} n)$. Further, every update to $G$ triggers $\Oh_{\Cc,r}(\log^{r-1} n)$ updates to $D$ in the amortized sense. The data structure can be initialized for an edgeless $G$ in time $\Oh_{\Cc,r}(n)$ and occupies $\Oh_{\Cc,r}(n)$ space at all times. 
\end{lemma}






\begin{proof} 
	Since $d_i \geq \nabla_0(\Cc_i)$ for $i \in [r]$, every graph in $\Cc_i$ admits an orientation with maximum outdegree at most $d_i$.
	Moreover, knowing $d_i$, we can efficiently maintain this orientation.
	
	For each $i\in [r]$,
	let us define dynamic graphs $G_i, G_i', D_i$ and $D_i'$, all on the same vertex set $V(G)$. All of them are going to be $r$-decorated, except for $D_r$, which will be $(\Lambda, r)$-decorated. The graphs $G_i$ and $G_i'$ will be undirected, while $D_i$ and $D_i'$ will be orientations of $G_i$ and $G_i'$, respectively. For convenience, we also define $G_0$ and $D_0$ to be edgeless graphs on the same vertex set too. We set $G_1' \coloneqq G$, where $\length_{G_1'}(e) \coloneqq 1$ for each $e \in E(G_1')$. The graphs $G_1', \ldots, G_r'$ will be edge-disjoint, so we can inductively define $G_1,\ldots,G_r$ and $D_1,\ldots,D_r$ from $G_1', \ldots, G_r'$ and $D_1', \ldots, D_r'$ by setting \[G_i \coloneqq G_{i-1} \cup G_i'\qquad\textrm{and}\qquad D_i \coloneqq D_{i-1} \cup D_i'.\] We want to maintain that $D_i'$ is a $4d_i$-orientation of $G_i'$. Consequently, $D_i$ will have maximum outdegree at most $\Delta_i \coloneqq 4(d_1+\ldots+d_i)$. Note that $\Delta_i \le 4 \cdot i \cdot d_i$. From the construction it will become apparent that $G_i' \in \Cc_i$, which is why a $4d_i$-orientation of $G_i'$ exists and can be effectively maintained using the data structure of \cref{thm:bf}.
	
	We now inductively define the graphs
	$G_2',\ldots,G'_r,D_1',\ldots,D_r'$.
	To define $G_{i+1}'$ for $i\in [r-1]$, consider a triple of distinct
	vertices $u, v, w \in V(G)$ such that
	$(w, u), (w, v) \in E(D_i)$, $uv \not\in G_i$ and
	$\length_{D_i}(w, u) + \length_{D_i}(w, v) =i+1$. For each such triple,
	we add the edge $uv$ to $G_{i+1}'$ with $\length_{G_{i+1}'}(u,v)=i+1$ and there are no other edges in $G_{i+1}'$. Note that this specification ensures that $G_{i+1}'$ will be edge-disjoint from $G_1', \ldots, G_i'$, as desired. Also note that $G_{i+1}$ defined in this way belongs to $\Fra(G_i, \Delta_i)$, so if $G_i \in \Cc_i$ and $\Delta_i \le 4 \cdot i \cdot d_i$, we have that $G_{i+1} \in \Cc_{i+1}$, as desired. Then $D_{i+1}'$ is a $4d_{i+1}$-orientation of $G_{i+1}'$, which will be eventually maintained using the data structure of \cref{thm:bf}. 
	Note that such definition ensures that $\dist_{G_{i+1}}(u,v)\leq \length_{D_{i+1}}(u,v)$, which is easily seen through an induction argument. That concludes the definitions of $G_i, G_i', D_i$ and $D_i'$. The resulting $D_r$ will be our faithful $(\Lambda, r)$-augmentation of $G$ (for some $\Lambda$ to be defined later).
	
%
%
%

	We now show how to effectively maintain the constructed graphs.
	
	\begin{claim}\label{cl:maintain}
		The graphs $G_i, G_i', D_i$ and $D_i'$ can all be maintained
		with amortized update time bounded by $\Oh_r(\log^{r-1}n \cdot \Delta_1 \Delta_2 \ldots \Delta_r)$ per edge addition or removal in $G$. Moreover, every update to $G$ triggers that many updates to $G_i, G_i', D_i$, and $D_i'$ in the amortized sense.
	\end{claim}
	
	\begin{claimproof}
		The graphs $D_i'$, which are orientations of $G_i'$, are maintained using the Brodal--Fagerberg data structure of \cref{thm:bf}. That is, assuming that $d_i$ is chosen so that $G_i$ is $d_i$-degenerate (which follows from $G_i\in \Cc_i$ and $d_i\ge \nabla_0(\Cc_i)$), the data structure of \cref{thm:bf} indeed maintains a $4d_i$-orientation of $G_i'$ (a subgraph of $G_i$), which we call $D_i'$. Note that by \cref{thm:bf}, these data structures take time $\Oh_{\Cc,r}(n)$ to initialize and use $\Oh_{\Cc,r}(n)$ space at all times.
		
		Whenever an update to the edge set of $G$ occurs, this starts a series of updates to all of $G_i, G_i', D_i$ and $D_i'$, for $i=1,2,\ldots,r$. Whenever an edge is supposed to change its orientation or its length in any of these graphs, this is modeled by one edge removal and one edge addition. Each change to the edge set of $D_i$ causes $\Oh(\Delta_i)$ changes to the edge set of $G_{i+1}'$, since for a given edge $(w, u)$ there are only at most $\Delta_i$ edges of the form $(w, v)$, and consequently, at most $\Delta_i$ edges $uv$ that need updating.
		Each change to $G_i'$ incurs $\Oh(\log n)$ reorientations in $D_i'$ in the amortized sense, which incurs the same number of reorientations in $D_i$, which then incurs $\Oh(\log n \cdot \Delta_i)$ changes in $G_{i+1}$. Hence, the amortized number of changes to all maintained graphs caused by one change to $G$ is $\Oh_r(\log^{r-1} n \cdot \Delta_1 \Delta_2 \ldots \Delta_{r-1})$ (one can easily argue that amortized costs multiply in the same way that worst-case costs do). Whenever we perform a series of reorientations, we should first perform all corresponding removals and only then all corresponding additions, to make sure that we stay within the relevant graph classes during the intermediate computations as well.
		
		An additional aspect requiring careful bookkeeping is the maintenance of labels and lengths. If we remove an edge $(w, u)$ from $D_i$ that along with some edge $(w, v)$  gave rise to an edge $(u, v) \in G_{i+1}'$, then $(u, v)$ may still belong to $G_{i+1}'$ if $w$ was not a unique vertex with edges to $u$ and $v$ in $D_i$ such that $\length_{D_i}(w, u)+\length_{D_i}(w, v)=i+1$. In order to correctly and effectively update the existence of this edge, for every edge $e$ in $G_{i+1}'$ we additionally keep an auxiliary counter $\mathrm{cnt}_{e, i}$, denoting the number of vertices $w \in V(G_i)$ such that $(w, u), (w, v) \in E(D_i)$ and $\length_{D_i}(w, u) + \length_{D_i}(w, v) = i+1$. Maintaining these counters lets us effectively (with $\Oh(1)$ overhead) determine the correctly updated edge set of $G_{i+1}'$ under any additions or removals of edges from $D_i$. As for labels in $D_r$, it is enough to set $\Lambda \coloneqq \{1, 2, \ldots, \Delta_r\}$ and for each vertex $v \in V(G)$ keep a boolean array $\mathrm{used}_v[1,\ldots,\Delta_r]$, where $\mathrm{used}_v[j]$ denotes whether there exists an edge outgoing of $v$ with the label $j$. When a new outgoing edge is added to $v$, we search for an unused label for it. Navigating through the $\mathrm{used}_v[1,\ldots,\Delta_r]$ array incurs an $\Oh(\Delta_r)$ overhead to the  $\Oh_r(\log^{r-1}(n) \Delta_1 \Delta_2 \ldots \Delta_r)$ total time of processing all necessary updates to all intermediate structures for one edge change in $G$.
	\end{claimproof}
	
	We  shall now prove that $D_r$ indeed contains the desired shortcuts.
	
	\begin{claim}\label{cl:shortcuts}
		Let $P=(u_0, \ldots, u_{r'})$ be a path in $G$ of some length $r' \le r$. Then, $D_r$ contains a shortcut of path $P$. 
	\end{claim}
%
	\begin{proof}
		Let us say that an edge $e \in G_r$ is \emph{useful} if it is of the form $u_{i}u_{j}$ and $\length_{G_r}(e) \le |j-i|$. 
		
		Let us take a subset
		$S=\{s_0, \ldots, s_c\} \subseteq \{0, \ldots, r'\}$ of smallest
		possible size such that $s_0=0, s_c=r'$ and for each $0 \le i \le c-1$ we have that $u_{s_i}u_{s_{i+1}}$ is a useful edge. There exists at least one feasible subset $S$ with that property as $\{0, \ldots, r'\}$ clearly fulfills the condition. We claim that these edges form a shortcut of $P$ in $D_r$.
		
		Let us assume that this is not the case. Then, there exists $1 \le i \le c-1$ such that both $s_{i}s_{i-1}$ and $s_is_{i+1}$ are directed outwards from $s_i$. We have that $\length_{D_r}(s_i, s_{i+1}) \le s_{i+1}-s_i$ and $\length_{D_r}(s_i, s_{i-1}) \le s_i - s_{i-1}$, so $L \coloneqq \length_{D_r}(s_i, s_{i+1})+\length_{D_r}(s_i, s_{i-1}) \le s_{i+1}-s_{i-1} \le r$.
	Hence, the graph $G_L$ is defined and at the moment of constructing it, we put the edge $s_{i-1}s_{i+1}$ into it, unless it was already present in $G_{L-1}$. In either case, we have that $\length_{G_r}(s_{i-1}s_{i+1}) \le L \le s_{i+1}-s_{i-1}$, so it is a useful edge as well, contradicting the minimality of $S$. 
	\end{proof}
	
	\cref{cl:shortcuts} proves that $D_r$ is a faithful $(\Lambda,r)$-augmentation of $G$, which by \cref{cl:maintain} can be maintained within the claimed complexity bounds. This concludes the proof.
\end{proof}

\subsection{Detecting far vertices}\label{sec:far-be}

We now move on to the implementation of the $\farVertexDS{G}{r,k,\eps}$ data structure. Our implementation is encapsulated in the following statement.

\begin{theorem}\label{thm:FarVertex-new}
	Fix a graph class $\Cc$ of bounded expansion, $r,k\in \N$, and $\eps>0$.
	Let $G$ be a dynamic graph on $n$ vertices that belongs at all times to $\Cc$. Then there exists a randomized data structure $\farVertexDS{G}{r,k,\eps}$ that provides access to the following query:
	\begin{itemize}
		\item $\far{S,\rho}$: For a given set of vertices $S\subseteq V(G)$ with $|S|\leq k$ and a function $\rho\colon S\to [0,r]$, return a vertex $v$ such that $\dist_G(u,S)>\rho(s)$ for all $s\in S$, or $\bot$ if no such vertex exists.
	\end{itemize}
	Every answer to the query is correct with probability at least $1-\eps$ against an oblivious adversary.
	The amortized time complexity of updates and queries is $\log^{\Oh_{\Cc,r,k}(1)} n\cdot \log \tfrac{1}{\eps}$. The data structure can be initialized for an edgeless $G$ in time $\Oh_{\Cc,r,k}(n\log n\log \tfrac{1}{\eps})$, and it uses $\Oh_{\Cc,r,k}(n\log n\log \tfrac{1}{\eps})$ space at all times.
\end{theorem}

The following statement follows from \cref{thm:homcounts-relational}.

\begin{lemma}\label{lem:dyn-rooted-homs}
	Fix a graph class $\Cc$ of bounded expansion, a finite set of labels $\Lambda$, $r\in \N$, and a rooted $(\Lambda,r)$-oriented graph $\rt{H}{x,y}$. Let $D$ be a dynamic $(\Lambda,r)$-oriented graph on $n$ vertices whose underlying undirected graph belongs at all times to $\Cc$. Further, let $u$ be a fixed vertex of $D$ and $\wei\colon V(D)\to \Z$ be a weight function; these do not change over time. Then there exists a data structure that maintains the value
	\[\sum_{v\in V(G)} \wei(v)\cdot |\Hom(\rt{H}{x,y},\rt{D}{u,v})|\]
	with amortized update time $\log^{\Oh_{\Cc,H,\Lambda,r}(1)} n$. The data structure can be initialized for an edgeless $D$ in time $\Oh_{\Cc,H,\Lambda,r}(n)$ and uses $\Oh_{\Cc,H,\Lambda,r}(n)$ space at all times.
\end{lemma}
\begin{proof}
	Follows immediately from \cref{thm:homcounts-relational} applied to the rooted $(\Lambda,r)$-oriented graph $\rt{H}{x}$ (understood as a relational structure in the obvious way) and the weight function $\wei'\colon V(H)\times V(D)$ defined as follows: for $(z,v)\in V(H)\times V(D)$,
	\[\wei'(z,v)\coloneqq \begin{cases}\wei(v) &\textrm{if }z=x,\\ 1 & \textrm{otherwise.}\end{cases}\qedhere\]
\end{proof}

The next lemma is the combinatorial core of our approach. We show that using the Inclusion--Exclusion Principle, we may express a weighted sum over vertices $u$ that are \emph{far} from a fixed vertex $v$ as a linear combination of weighted sums over vertices that are \emph{close}.

\begin{lemma}\label{lem:IE}
	Let $G$ be a graph, $\wei\colon V(G)\to \Z$ be a weight function on the vertices of $G$, $u$ be a vertex of~$G$, and $r\in \N$. Suppose $D$ is a faithful $(\Lambda,r)$-augmentation of $G$, for some finite set of labels $\Lambda$. Then there exists a set $\Mm$, depending only on $\Lambda$ and $r$, with the following properties:
	\begin{itemize}
		\item every element of $\Mm$ is a pair $(s, \rt{H}{x,y})$, where $s\in \{-1,+1\}$ and $\rt{H}{x,y}$ is a $(\Lambda,r)$-oriented graph with $\Oh_{\Lambda,r}(1)$ vertices and two roots $x,y$; and
		\item we have
		\[\sum_{\substack{v\in V(G)\\ \dist_G(u,v)>r}} \wei(v) =\sum_{(s,\rt{H}{x,y})\in \Mm}\ s\cdot \sum_{v\in V(G)} \wei(v)\cdot |\Hom(\rt{H}{x,y},\rt{D}{u,v})|.\]
	\end{itemize}
\end{lemma}
\begin{proof}
	Let a \emph{pattern} be a pair of finite sequences $\sigma,\tau$, each with entries in $\Lambda\times [r]$ such that the total sum of the second coordinates of the entries of $\sigma$ and $\tau$ is at most $r$. Note that this in particular means that $|\sigma|+|\tau|\leq r$. Let $\Pi$ be the set of all possible patterns; note that $|\Pi|\leq \Oh_{\Lambda,r}(1)$.
	
	For a pattern $\pi=(\sigma,\tau)\in \Pi$ and a vertex $v$ of $G$, we say that $\pi$ is \emph{realized} in $\rt{D}{u,v}$ if in $D$ there exists oriented walks $P_\sigma=(a_0,a_1,\ldots,a_{|\sigma|})$ and $P_\tau=(b_0,b_1,\ldots,b_{|\tau|})$ such that
	\begin{itemize}
		\item for each $i\in [|\sigma|]$, the pair $(\lab_D((a_{i-1},a_i)),\length_D((a_{i-1},a_i)))$ is equal to the $i$th entry of $\sigma$;
		\item for each $j\in [|\sigma|]$, the pair $(\lab_D((b_{j-1},b_j)),\length_D((b_{j-1},b_j)))$ is equal to the $j$th entry of $\tau$; and
		\item $a_0=u$, $b_0=v$, and $a_{|\sigma|}=b_{|\tau|}$. 
	\end{itemize}
	The subgraph of $\rt{D}{u,v}$ consisting of the union of the walks $P_\sigma$ and $P_\tau$ will be called the \emph{realization} of $\pi$ in $\rt{D}{u,v}$. Note that since $D$ is unambiguous, the realization of $\pi$ in $\rt{D}{u,v}$, if existent, is unique.
	
	We note the following.
	
	\begin{claim}\label{cl:dist-pattern}
		For every vertex $v$ of $G$, we have the following:
		\[\dist_G(u,v)\leq r \qquad\textrm{if and only if}\qquad \textrm{there exists }\pi\in \Pi\textrm{ that is realized in }\rt{D}{u,v}.\]
	\end{claim}
	\begin{claimproof}
		Suppose first that $\dist_G(u,v)\leq r$ and let $Q$ be a $u$-$v$ path of length at most $r$. Since $D$ is a faithful $(\Lambda,r)$-augmentation of $G$, by \ref{c:center} there is a shortcut of $Q$, and this shortcut witnesses that some pattern $\pi\in \Pi$ is realized in $\rt{D}{u,v}$.
		
		Suppose next that some pattern $\pi\in \Pi$ is realized in $\rt{D}{u,v}$, hence it has some realization. Since the total sum of lengths of the edges in this realization is at most $r$, we can repeatedly use \ref{c:short} with triangle inequality to infer that $\dist_G(u,v)\leq r$.
	\end{claimproof}
	
	Next,
	call a rooted $(\Lambda,r)$-oriented graph $\rt{H}{x,y}$ \emph{anchored} if every vertex of $H$ is reachable from $x$ or $y$ in $H$. Consider any vertex $v$ of $G$. We note the following.
	
	\begin{claim}\label{cl:unambigious}
		For any anchored $(\Lambda,r)$-oriented graph $\rt{H}{x,y}$, we have \[|\Hom(\rt{H}{x,y},\rt{D}{u,v})|\in \{0,1\}.\]
	\end{claim}
	\begin{claimproof}
		Consider any $\varphi\in \Hom(\rt{H}{x,y},\rt{G}{u,v})$. Observe that for any edge $(s,t)\in E(H)$, if $\varphi(s)=w$, then $\varphi(t)$ must be an out-neighbor of $w$ in $D$ such that the edge $(w,\varphi(t))$ has label $\lab_H((s,t))$. By the unambiguity of $D$, such an outneighbor, if existent, is unique. Since we require that $\varphi(x)=u$ and $\varphi(y)=v$, the assumption that $\rt{H}{x,y}$ is anchored implies that all the images of all the vertices of $H$ under $\varphi$ are determined, or no such $\varphi$ exists.
	\end{claimproof}
	 
	\begin{claim}\label{cl:testing-subset}
		For any set of patterns $\Gamma\subseteq \Pi$ there exists an anchored $(\Lambda,r)$-oriented graph $\rt{H_\Gamma}{x,y}$ on at most $2+(r-1)|\Gamma|$ vertices such that we have:
			\[\textrm{all the patterns of }\Gamma\textrm{ are realized in } \rt{D}{u,v}\qquad \textrm{if and only if}\qquad \Hom(\rt{H_\Gamma}{x,y},\rt{D}{u,v})\neq \emptyset.\]
	\end{claim}
	\begin{claimproof}
		For every pattern $\pi=(\sigma,\tau)\in \Gamma$, build a $(\Lambda,r)$-oriented graph $\rt{H_\pi}{x,y}$ as follows:
		\begin{itemize} 
			\item construct an oriented path $P_\sigma$ of length $|\sigma|$ that starts at $x$ and whose consecutive edges have lengths and labels equal to the consecutive terms of $\sigma$;
			\item construct an oriented path $P_\tau$ analogously for $\tau$, where the start vertex of $P_\tau$ is $y$; and
			\item identify the end vertex of $P_\sigma$ with the end vertex of $P_\tau$.
		\end{itemize}
		Note that $\rt{H_\pi}{x,y}$ is anchored and has at most $r-1$ vertices apart from $x$ and $y$. 
		
		Now, $\rt{H_\Gamma}{x,y}$ can be obtained by taking the disjoint union of graphs $\rt{H_\pi}{x,y}$ for all $\pi\in \Gamma$ and identifying all their vertices $x$ into a single vertex $x$, and similarly for $y$ (see Figure~\ref{fig:HGamma} for the reference). It is straightforward to see that $\rt{H_\Gamma}{x,y}$ constructed in this way has the desired properties.
	\end{claimproof}
	
	For $\Gamma\subseteq \Pi$, let $V_\Gamma$ be the set of those vertices $v$ of $G$ for which every pattern of $\Gamma$ is realized in $\rt{D}{u,v}$. Also, let $V_\mathsf{far}$ be the set of vertices $v\in V(G)$ satisfying $\dist_G(u,v)>r$. By \cref{cl:dist-pattern}, $V_\mathsf{far}$ can be equivalently defined the set of those vertices $v$ for which none of the patterns of $\Pi$ is realized in $\rt{D}{u,v}$. Therefore, by the Inclusion--Exclusion Principle,
	\begin{equation}\label{eq:bobr}
		\sum_{\substack{v\in V(G)\\ \dist_G(u,v)>r}} \wei(v)=\sum_{v\in V_\mathsf{far}} \wei(v) = \sum_{\Gamma\subseteq \Pi} (-1)^{|\Gamma|}\sum_{v\in V_\Gamma} \wei(v).
	\end{equation}
	Next, by \cref{cl:unambigious,cl:testing-subset}, for every $\Gamma\subseteq \Pi$ and $v\in V(G)$ we have
	\begin{equation}\label{eq:wydra}
		|\Hom(\rt{H_\Gamma}{x,y},\rt{D}{u,v})|=\begin{cases}1 & \textrm{if }v\in V_\Gamma,\\ 0 & \textrm{otherwise.}\end{cases}
	\end{equation}
	By combining \eqref{eq:bobr} and \eqref{eq:wydra} we conclude that
	\begin{equation*}
		\sum_{\substack{v\in V(G)\\ \dist_G(u,v)>r}} \wei(v)= \sum_{\Gamma\subseteq \Pi} (-1)^{|\Gamma|}\sum_{v\in V(G)} \wei(v)\cdot |\Hom(\rt{H_\Gamma}{x,y},\rt{D}{u,v})|.
	\end{equation*}	
	Hence, we may take
	\[\Mm\coloneqq \left\{\,((-1)^{|\Gamma|},\rt{H_\Gamma}{x,y})\colon \Gamma\subseteq \Pi\,\right\}.\qedhere\]
\end{proof}

By combining \cref{lem:dyn-augmentation,lem:dyn-rooted-homs,lem:IE}, we easily obtain the following.

\begin{lemma}\label{lem:wei-sums-far}
	Fix a graph class $\Cc$ of bounded expansion and $r\in \N$.
	Let $G$ be a dynamic graph on $n$ vertices that belongs at all times to $\Cc$, $\wei\colon V(G)\to \Z$ be a weight function fixed upon initialization, and $u$ be a fixed vertex of $G$ (which does not change over time). Then there exists a data structure that maintains the value
	\[\sum_{\substack{v\in V(G)\\ \dist_{G}(u,v)>r}} \wei(v)\]
	with amortized update time $\log^{\Oh_{\Cc,r}(1)} n$. The data structure can be initialized for an edgeless $G$ in time $\Oh_{\Cc,r}(n)$ and uses $\Oh_{\Cc,r}(n)$ space at all times.
\end{lemma}
\begin{proof}
	Let $\Lambda$ be the label set provided for $\Cc$ and $r$ by \cref{lem:dyn-augmentation}; note that $|\Lambda|\leq \Oh_{\Cc,r}(1)$. Further, let $\Mm$ be the set provided by \cref{lem:IE} for $\Lambda$ and $r$.
	By \cref{lem:dyn-augmentation}, we may maintain a faithful $(\Lambda,r)$-augmentation $D$ of $G$ so that every update to $G$ triggers an amortized number of $\log^{\Oh_{\Cc,r}(1)} n$ updates to $D$. Therefore, by \cref{lem:dyn-rooted-homs}, for every $(s,\rt{H}{x,y})\in \Mm$ we may construct a data structure that maintains the value $\sum_{v\in V(G)} \wei(v)\cdot |\Hom(\rt{H}{x,y},\rt{D}{u,v})|$ with amortized update time $\log^{\Oh_{\Cc,H,\Lambda,r}(1)} n$ (per update in $D$). Since $|\Mm|\leq \Oh_{\Lambda,r}(1)\leq \Oh_{\Cc,r}(1)$ and each graph $\rt{H}{x,y}$ featured in $\Mm$ has $\Oh_{\Lambda,r}(1)\leq \Oh_{\Cc,r}(1)$ vertices, we conclude that updating all those data structures takes amortized time $\log^{\Oh_{\Cc,r}(1)} n$ per update in $D$, so also amortized time $\log^{\Oh_{\Cc,r}(1)} n$ per update in $G$. Finally, the value $\sum_{\substack{v\in V(G)~\mid~ \dist_{G}(u,v)>r}} \wei(v)$ can be recomputed in time $\Oh_{\Cc,r}(1)$ upon every update using the formula provided by \cref{lem:IE}. The claimed bounds on the initialization time and the space usage follow directly from the bounds provided by \cref{lem:dyn-augmentation,lem:dyn-rooted-homs}.
\end{proof}

With all the tools prepared, we may prove \cref{thm:FarVertex-new} and thereby provide a suitable data structure $\farVertexDS{G}{r,k}$ by combining \cref{lem:wei-sums-far,lem:restore-by-sampling}.

\begin{proof}[Proof of \cref{thm:FarVertex-new}]
	We first argue that without loss of generality we may assume that in every query $\far{S,\rho}$, we have $S=\{u\}$ for a fixed vertex $u$ that does not change over time and $\rho(u)=r+1$. We do it as follows. Instead of maintaining the original dynamic graph $G$, we maintain the dynamic graph $G'$ obtained from $G$ by $1+kr$ fresh vertices that remain isolated. Distinguish one of those vertices and call it $u$. Upon query $\far{S,\rho}$ in $G$, we temporarily use the additional isolated vertices to connect $u$ with each $s\in S$ by a path of length $(r+1)-\rho(s)$; this requires at most $k(r+1)$ edge additions. Once the paths are prepared, we ask the query $\far{\{u\},(u\mapsto r+1)}$ in $G'$. It is clear that this query in $G'$ is equivalent to the query $\far{S,\rho}$ in $G$ in terms of correct answers. Having obtained the answer to the query in $G'$, we may remove the temporary edges, thus making $u$ and all the other additional vertices again isolated. Note that we have $G'\in \Cc'$ where $\Cc'$ is the class of all subdivisions of graphs in which one vertex can be removed to obtain a graph from $\Cc$; this class also has bounded expansion, since we have $\nabla_d(\Cc')\leq \nabla_d(\Cc)+1$ for all $d\in \N$. So the result for $\Cc'$ with the assumption that $S=\{u\}$ and $\rho=(u\mapsto r+1)$ in all the queries implies the general result for $\Cc$.
	
	We proceed under the assumption that every query asks about $S=\{u\}$ and $\rho=(u\mapsto r+1)$, for a fixed vertex $u$.
	For the purpose of applying \cref{lem:restore-by-sampling}, for a vertex $v$ of $G$ we define
	\[f(G,v)\coloneqq \begin{cases} 1 &\textrm{if }\dist_G(u,v)>r+1,\\ 0 & \textrm{otherwise.}\end{cases}\]
	A suitable data structure $\WeiSum_{f, \wei}[G]$ for any fixed weight function $\wei\colon V(G)\to \Z$ is provided by \cref{lem:wei-sums-far} (applied with parameter $r+1$ instead of $r$). The data structure for $\Ver_f[G]$ can be obtained by maintaining the data structures $\WeiMaps_{\Hom,\rt{H}{x,y},1}[G]$ provided by \cref{thm:homcounts-relational} for $\rt{H}{x,y}$ ranging over paths of length $0,1,2,\ldots,r+1$ with endpoints $x$ and $y$. Therefore, \cref{lem:restore-by-sampling} applies and the data structure $\mathsf{RetrieveVertex}_{f,\eps}[G]$ provided by it meets all the requirements for $\farVertexDS{G}{r,k,\eps}$.
\end{proof}

\subsection{Dynamic domination and independence in classes of bounded expansion}

We may finish the proofs of our main two results, \cref{thm:main-domset-be,thm:main-indset-be}, which we recall for convenience.

\mainDomBE*
\begin{proof}
	Follows immediately by plugging \cref{thm:NearVertex-new,thm:FarVertex-new} into \cref{lem:summary-dom} and using \cref{thm:sli-be} to infer the finiteness of $\sli_r(\Cc)$.
\end{proof}

\mainIndBE*
\begin{proof}
	Follows immediately by plugging \cref{thm:NearVertex-new,thm:FarVertex-new} into \cref{lem:summary-ind}.
\end{proof}

\section{Distance-\texorpdfstring{$1$}{1} domination in degenerate graphs}\label{sec:distone}

In this section we give a simpler data structure for \textsc{Dominating Set} for the distance-$1$ case, that is, we prove \cref{thm:main-domset-deg}. For this purpose, we give a simpler and more efficient implementation of data structures supporting queries $\near{S,1}$ and $\far{S,1}$ under the assumption that the maintained graph is $d$-degenerate. These data structures are new; they are not merely specializations of the bounded-expansion constructions from \cref{sec:implementation}. 

Throughout this section we assume that the maintained dynamic graph $G$ has a fixed vertex set $V=\{1,\ldots,n\}$ and is $d$-degenerate at all times, for a fixed parameter $d\in \N$. We also assume, without loss of generality, that $n\geq 2$ and $d\geq 1$.




Using the Brodal--Fagerberg data structure (\cref{thm:bf}), we may assume that we maintain an orientation $\vec{G}$ of $G$ with maximum outdegree $\Delta \coloneqq 4d$. 
We write $N(v)$ for the open neighborhood of $v$ in $G$ (consisting of all the neighbors of $v$ in $G$), $N[v]\coloneqq \{v\}\cup N(v)$ for the closed neighborhood, and $N[X]\coloneqq \bigcup_{v\in X}N[v]$ for any $X\subseteq V$. We also write $N^+(v)\coloneqq \{w\in V : (v,w)\in E(\vec G)\}$ for the open outneighborhood of $v$ in $\vec G$.

\subsection{Toolkit}\label{sec:distone-infra}



We start with describing a helper data structure that we call \emph{toolkit}, which will be later used in the data structures for $\near{S,1}$ and $\far{S,1}$ queries. We stress that the toolkit is fully deterministic; randomization will enter the scene only later, through an application of \cref{thm:dvorak-tuma-examples}.


Let us first explain the intuition.
The toolkit architecture replaces searching through neighborhoods in $G$ (which may be large) with certain bucket arithmetic plus inspection of a few outneighborhoods in~$\vec G$; these are small due to the bound on the maximum outdegree of $\vec G$.  
Buckets reduce the query ``given $Z$, find $u$ such that $Z\subseteq N^+(u)$'' to a single dictionary lookup: instead of scanning through all the vertices of~$V$, we consult the single list $B(Z)\coloneqq \{u\in V \mid Z\subseteq N^+(u)\}$, maintained in the dictionary.

To facilitate later uses for weighted counting,
we shall assume that the maintained graph is equipped with a weight function $\wei\colon V\to \Z$ that is fixed upon initialization and does not change over time. In the toolkit data structure, we then explicitly maintain the following objects at all times.
\begin{itemize}
\item The bucket family $B(Z)\coloneqq \{u\in V \mid Z\subseteq N^+(u)\}$ for all sets $Z\subseteq V$ with $B(Z)$ being nonempty (note that this implies $|Z|\leq \Delta$). Each bucket $B(Z)$ is stored as a doubly linked list, while the whole bucket family is stored as a dictionary keyed by the sorted representation of $Z$. Observe that every vertex $u\in V$ is in $2^{|N^+(u)|}\leq 2^\Delta$ buckets, hence the total sum of the lengths of the lists stored for the buckets $B(Z)$ is bounded by $2^\Delta\cdot n$. In particular, the dictionary stores at most this many entries. Additionally, with every bucket $B(Z)$ we store its total weight $\wei(B(Z))\coloneqq \sum_{v\in B(Z)}\wei(v)$.
\item A registration list for every vertex $u \in V$, tracking its $2^{|N^+(u)|} \le 2^\Delta$ bucket memberships (pointers to list elements on the bucket lists).
\end{itemize}
We remark that if we have a pointer to an element of a doubly linked list, then this element may be deleted in $\Oh(1)$ time from that list, hence maintaining registration lists lets us delete all occurrences of $u$ on bucket memberships lists in $\Oh(2^{\Delta})$ time complexity.

We recall that one of the features of the \cref{thm:bf} is the ability to perform the membership tests $v\in N(u)$ and $v\in N^+(u)$ in $\Oh(d)$ time.

Every bucket key $Z$ is stored explicitly as the sorted list of its elements, so it occupies $\Oh(|Z|)\le \Oh(\Delta)=\Oh(d)$ space. It stores a pointer to its first element, its current cardinality, and its current weight.
Consequently, one bucket lookup, insertion, or deletion by key requires $\Oh(\log (2^\Delta n))=\Oh(\Delta+\log n)$ comparisons on lists of length at most $\Delta$, hence it takes $\Oh(d^2 + d \log n)$ time complexity (more sophisticated data structures would allow for a more efficient lookup, but the improvement is considered negligible). Whereas after the bucket is found, accessing its first element, its cardinality or its weight costs $\Oh(1)$ time.
Note that the bucket $B(\emptyset)$ consists of the whole vertex set $V$ at all times. 


We now prove a lemma providing guarantees on the complexity of maintaining a toolkit.

\begin{lemma}[Toolkit]\label{lem:infra-update}
One instance of the toolkit described above can be initialized on the edgeless graph on $V$ provided with a weight function $\wei\colon V\to \Z$ in $\Oh(n)$ time.
The toolkit uses $\Oh(2^{4d}\cdot dn)$ space at all times.
Upon edge insertion or deletion, the toolkit can be updated in amortized time $\Oh(2^{4d}\cdot d^2\cdot\log^2 n)$. 
\end{lemma}
\begin{proof}
Initialization takes $\Oh(n)$ time: we set up all the relevant dictionaries to be empty, and then register all the vertices in the bucket $B(\emptyset)$, which becomes the only element in the dictionary of buckets.


Consider one edge insertion or deletion.
First, we update the two affected neighbor dictionaries $N(\cdot)$; this takes $\Oh(\log n)$ time.
Second, we update the Brodal--Fagerberg orientation $\vec G$.
By \cref{thm:bf}, this takes amortized time $\Oh(d+\log n)$, the amortized number of reoriented edges is $\Oh(\log n)$, and we obtain the whole list of those edges, hence also the set of vertices whose out-neighborhood changed.
Each reoriented edge contributes at most its two endpoints.
In addition, the updated edge contributes only its two endpoints. 
Therefore, the amortized number of vertices $u$ for which the set $N^+(u)$ may change is~$\Oh(\log n)$.

Fix one such vertex $u$.
Let $A_{\mathrm{old}}\coloneqq N^+_{\mathrm{old}}(u)$ and $A_{\mathrm{new}}\coloneqq N^+_{\mathrm{new}}(u)$, where we follow the convention that $\mathrm{old}$ and $\mathrm{new}$ signifies whether an object in $G$ refers to the state before or after the update.
The old registrations of $u$ form the family $\mathcal{R}_{\mathrm{old}}(u)\coloneqq \{Z\mid Z\subseteq A_{\mathrm{old}}\}$, while the new registrations form the family  $\mathcal{R}_{\mathrm{new}}(u)\coloneqq \{Z\mid Z\subseteq A_{\mathrm{new}}\}$.
Producing one subset key and performing the corresponding bucket lookup, insertion or deletion, and update to the bucket's cardinality and weight, costs $\Oh(d^2 + d\log n)$, because the key has size at most $\Delta$. Note here that if the bucket $B(Z)$ did not exist before ($B_{\textrm{old}}(Z)=\emptyset$), it must be added to the dictionary of buckets, and if its cardinality drops to $0$, then it must be removed from the dictionary of buckets; all of this can be done within the same complexity bounds.
So the whole work caused by the vertex $u$ is bounded by \[\Oh((|\mathcal{R}_{\mathrm{old}}(u)|+|\mathcal{R}_{\mathrm{new}}(u)|)\cdot (d^2+d\log n))=\Oh(2^\Delta\cdot (d^2+d\log n)).\]

Multiplying by the amortized number $\Oh(\log n)$ of changed vertices per graph update yields amortized update time $\Oh(2^\Delta\cdot (d^2+d\log n)\cdot \log n)\leq \Oh(2^{4d}\cdot d^2\cdot \log^2 n)$, because $\Delta=4d$.
Note here that the factor $2^{4d}$ comes from enumerating all subsets of an outneighborhood of size at most $\Delta=4d$.

For the space bound, the orientation $\vec G$ and the outneighbor dictionaries use $\Oh(nd)$ space because every $d$-degenerate graph on $V$ has at most $d(n-1)$ edges.
For every $u\in V$, its registration list contains at most $2^{|N^+(u)|}\le 2^\Delta$ entries.
Hence all registrations and bucket incidences together use at most $\Oh\left(\sum_{u\in V} 2^{|N^+(u)|}\right)=\Oh(2^\Delta\cdot n)$ words, which also gives an upper bound on the total size of the lists stored for the buckets. 
Thus, as $\Delta=4d$, the whole toolkit uses $\Oh(2^{4d}\cdot n)$ words of space.
\end{proof}

In further preparation for the queries $\near{S,1}$ and $\far{S,1}$, we introduce the following definitions. Let $S\subseteq V(G)$. We define:
\begin{align*}
	U(S)&\coloneqq \{u\in V\setminus S\mid S\subseteq N(u)\};\\
	U_2(S) & \coloneqq \{u\in V\setminus S\mid S\subseteq N^+(u)\};\\
	U_1(S) & \coloneqq U(S)\setminus U_2(S)=\{u\in U(S)\mid u\in \textstyle\bigcup_{s\in S} N^+(s)\}.
\end{align*}
Clearly, $(U_1(S),U_2(S))$ is a partition of $U(S)$.
Note also that $|U_1(S)|\leq \Delta|S|$, because $|N^+(s)|\leq \Delta$ for each $s\in S$.
Finally, observe the following.


\begin{lemma}[Bucket characterization of $U_2(S)$]\label{lem:distone-U2-bucket}
For any $S\subseteq V$, we have $U_2(S)=B(S)$.
\end{lemma}
\begin{proof}
	Follows from the definitions and the observation that $S\subseteq N^+(u)$ entails $u\notin S$, due the fact that $G$ is a simple graph.
\end{proof}

\cref{lem:distone-U2-bucket} provides the following combinatorial characterization of vertices that are candidates for the output in the query $\near{S,1}$.

\begin{lemma}\label{cl:distone-near-candidates}
Let $S\subseteq V$.
If a vertex $v$ satisfies $S\subseteq N[v]$, then exactly one of the following holds:
\begin{itemize}
\item $v\in S$; or
\item $v \not\in S$ and $v\in \bigcup_{s\in S}N^+(s)$; or
\item $|S|\le \Delta$ and $v\in B(S)$.
\end{itemize}
In particular, apart from the single bucket $B(S)$, the list above includes at most $|S|+\Delta\cdot |S|=\Oh(d|S|)$ candidates.
\end{lemma}
\begin{proof}
If $v\in S$, we are in the first case.
Assume then that $v\notin S$, so $v\in U(S)$.
If $v\in U_1(S)$, then $v\in \bigcup_{s\in S}N^+(s)$.
And if $v\in U_2(S)$, then $|S|\leq \Delta$ due to $S\subseteq N^+(v)$, and \cref{lem:distone-U2-bucket} yields $v\in B(S)$.
\end{proof}

\subsection{Detecting near vertices}\label{sec:distone-near}


In this subsection we provide the desired implementation of the $\near{S,1}$ queries. Note that this time, the data structure is entirely deterministic and can handle any given set $S$, not necessarily of size stipulated by a fixed constant. Also, since we are going to reuse this data structure in the implementation of $\far{S,1}$, we equip it with the capability of weighted counting.

\begin{theorem}\label{thm:NearVertex-new-deg}
	Fix $d\in \N$.
	Let $G$ be a dynamic graph on vertex set $V=\{1,\ldots,n\}$ that is $d$-degenerate at all times, and let $\wei\colon V\to \Z$ be a weight function fixed upon initialization. Then there exists a data structure $\nearVertexDS{G}{}$ that provides access to the following query:
	\begin{itemize}
		\item $\near{S,1}$: For a given set of vertices $S\subseteq V(G)$, return a vertex $v$ such that $S\subseteq N[v]$, or $\bot$ if no such vertex exists. Also, return the value $\sum_{v\colon S\subseteq N[v]} \wei(v)$.
	\end{itemize}
	The amortized time of an update is $\Oh(2^{4d}\cdot d^2\cdot \log^2 n)$. The queries take worst-case time $\Oh(d^2|S|^2\cdot\log n)$. The data structure can be initialized for an edgeless $G$ in $\Oh(n)$ time and uses $\Oh(2^{4d}\cdot n)$ space at all times.
\end{theorem}

\begin{proof}
The data structure consists of one instance of the toolkit of \cref{lem:infra-update} and an auxiliary Boolean array indexed by $V$ used for temporary marking. Therefore, the guarantees on the update time, initialization time, and space usage follows directly from \cref{lem:infra-update}. We are left with implementing the query $\near{S,1}$.


Upon query $\near{S,1}$, we execute the following steps:
\begin{enumerate}
	\item Initialize $L$ to be an empty list. The intention is that on $L$ we shall gather the vertices $v\in S\cup \bigcup_{s\in S}N^+(s)$ satisfying $S\subseteq N[v]$.
	\item For every $s\in S$, test whether $S\setminus\{s\}\subseteq N(s)$. If this is the case, add $s$ to $L$ and mark it in the auxiliary Boolean array as added.
	\item For every $s\in S$ and every vertex $v\in N^+(s)$, test whether $S\subseteq N[v]$. If this is the case and $v$ is not marked as already added, add $v$ to $L$ and mark it as added.
	\item Iterate through $L$ and remove all the markings from the auxiliary Boolean array. Thus it becomes empty (all-false) again and $L$ contains no duplicates.
	\item If $|S|\le \Delta$, find the bucket $B(S)$ in the dictionary of buckets, if existent.
	\item If $L$ is nonempty or $B(S)$ exists, return the first vertex of $L$ or the first vertex of $B(S)$; otherwise return $\bot$. For $\sum_{v\colon S\subseteq N[v]} \wei(v)$, return $\wei(B(S))+\wei(L)$, where $\wei(B(S))$ is replaced by $0$ if $B(S)$ is not present in the dictionary of buckets. Note that $L$ and $B(S)$ are disjoint according to \cref{cl:distone-near-candidates}.
\end{enumerate}
Step~1 takes constant time. Step~2 takes $\Oh(|S|^2 d)$ time, because it boils down to $\Oh(|S|^2)$ adjacency checks. Step~3 takes $\Oh(d^2|S|^2)$ time, because there are at most $d|S|$ vertices $v$ to consider, and for each of them we test $|S|$ vertices of $S$ for membership in $N[v]$, each in $\Oh(d)$ time. Step~4 takes $\Oh(|L|)\leq \Oh(d|S|)$ time.
Step~5 takes $\Oh(d^2+d\log n)$ time, because this is the complexity of retrieving the bucket $B(S)$ from the dictionary of buckets. Finally, Step~6 takes constant time, because retrieving the weight and the first element of a bucket takes constant time. Thus, executing the query takes $\Oh(d^2|S|^2 + d\log n) = \Oh(d^2 |S|^2 \log n)$ time in the worst case.

The correctness of the algorithm presented above follows immediately from \cref{cl:distone-near-candidates} and the observation that the contribution to the sum $\sum_{v\colon S\subseteq N[v]} \wei(v)$ of the vertices $v$ that do not satisfy $S\subseteq N^+(v)$ is exactly counted in the variable $\gamma$.
\end{proof}

\subsection{Detecting far vertices}\label{sec:distone-far}

In this section we give a data structure for the $\far{S,1}$ queries. Similarly to the reasoning presented in \cref{sec:far-be}, the idea is to first give a data structure that maintains weighted sums over vertices outside of $N[S]$ (this part is entirely deterministic), and then lift it to an example-reporting data structure using \cref{lem:restore-by-sampling}. Therefore, our first goal is to prove the following.

\newcommand{\WeiFar}{\mathsf{WeightedFarSum}}
\newcommand{\sumFar}{\mathtt{sumFar}}

\begin{lemma}\label{lem:wei-sums-far-deg}
	Fix $d\in \N$.
	Let $G$ be a dynamic graph on the vertex set $V=\{1,\ldots,n\}$ that is $d$-degenerate at all times, and let $\wei\colon V\to \Z$ be a weight function fixed upon initialization. Then there exists a data structure $\WeiFar_{\wei}[G]$ for $G$ that supports the following query:
	\begin{itemize}
		\item $\sumFar(S)$: Given $S\subseteq V$, return $\sum_{v\in V\setminus N[S]} \wei(v)$.
	\end{itemize}
	The amortized update time is $\Oh(2^{4d}\cdot d^2\cdot \log^2 n)$, and every query is answered in worst-case time $\Oh(2^{|S|}\cdot d^2|S|^2\cdot \log n)$. The data structure can be initialized for an edgeless $G$ in time $\Oh(n)$ and uses $\Oh(2^{4d}\cdot n)$ space at all times.
\end{lemma}
\begin{proof}
	Observe that by the Inclusion-Exclusion Principle,
	\[\sum_{v\in V\setminus N[S]} \wei(v) = \sum_{v\colon N[v]\cap S=\emptyset} \wei(v)= \sum_{S'\subseteq S} (-1)^{|S'|}\cdot \sum_{v\colon S'\subseteq N[v]} \wei(v).\]
	Therefore, if we maintain one instance of the $\nearVertexDS{G}{}$ data structure provided by \cref{thm:NearVertex-new-deg}, then the value $\sum_{v\in V\setminus N[S]} \wei(v)$ can be computed using the formula above from the answers to $2^{|S|}$ queries to $\nearVertexDS{G}{}$: one for each subset $S'\subseteq S$. The complexity guarantees follow directly from the guarantees provided by \cref{thm:NearVertex-new-deg}.
\end{proof}

We are now in position to implement the data structure for $\far{S,1}$ queries.

\begin{theorem}\label{thm:FarVertex-new-deg}
	Fix $d\in \N$ and $\eps>0$.
	Let $G$ be a dynamic graph on vertex set $V=\{1,\ldots,n\}$ that is $d$-degenerate at all times. Then there exists a randomized data structure $\farVertexDS{G}{\eps}$ that provides access to the following query:
	\begin{itemize}
		\item $\far{S,1}$: For a given set of vertices $S\subseteq V$, return a vertex $v$ that belongs to $V-N[S]$, or $\bot$ if no such vertex exists.
	\end{itemize}
	Every answer to the query is correct with probability at least $1-\eps$ against an oblivious adversary.
	The amortized time complexity of updates and queries is $\Oh(2^{4d+|S|}\cdot d^2|S|^2\cdot \log^3 n \log \tfrac{1}{\eps})$. The data structure can be initialized for an edgeless $G$ in time $\Oh(n\log n\log \tfrac{1}{\eps})$, and uses $\Oh(2^{4d}\cdot n\log n\log \tfrac{1}{\eps})$ space at all times.
\end{theorem}
\begin{proof}

	We first observe that the data structure of \cref{lem:restore-by-sampling} can be extended to handle graphs with a given set of vertices $S$, in the following sense:
	\begin{itemize}
		\item The vertex function $f$ takes three arguments: a graph $G$, a vertex subset $S\subseteq V(G)$, and a vertex $v$.
		\item The data structure $\WeiSum_{f, \wei}[G]$ instead of reporting $\sum_{v\in V(G)}f(G,v)\cdot \wei(v)$ after every update, can be queried for $\sum_{v\in V(G)}f(G,S,v)\cdot \wei(v)$ for a set $S$ given as input to the query.
		\item Similarly, the query of the data structure $\Ver_{f}[G]$ takes also $S$ on input and reports whether $f(G,S,v)>0$.
		\item The query of the constructed data structure $\RetVer_{f,\eps}[G]$ also takes  $S$ on input and returns any vertex $v$ such that $f(G,S,v)>0$, or $\bot$ if there is no such vertex.
		\item The promised guarantee $T\cdot \log \tfrac{1}{\eps}$ on the update/query complexity of $\WeiSum_{f, \wei}[G]$ and $\Ver_{f}[G]$ may depend on $|S|$. The guarantee on the update/query complexity of $\RetVer_{f,\eps}[G]$ becomes $\Oh(T\cdot |S|\cdot \log \tfrac{1}{\eps})$.
	\end{itemize}
	To argue this, one may either readily verify that the presented proof goes through without any changes, or perform the following easy gadgeteering.
	Add to $G$ a fresh special vertex $v^\star$ with $\wei(v^\star)=0$ that will normally stay isolated; let $G^\star$ be $G$ with $v^\star$ added.
	We define a (standard) vertex function $f^\star$ as $f^\star(G^\star,v)\coloneqq f(G,N(v^\star),v)$ for $v\in V(G)$, and $f^\star(G^\star,v^\star)=0$; that is, we interpret $S$ to be the neighborhood of $v^\star$ and put value $0$ on the special vertex $v^\star$. To implement the extended data structure $\RetVer_{f,\eps}[G]$ of \cref{lem:restore-by-sampling}, we set up the data structure of \cref{lem:restore-by-sampling} without extension, but for $f^\star$ and $G^\star$; that is, $\RetVer_{f^\star,\eps}[G^\star]$. This requires suitable implementations of $\WeiSum_{f^\star, \wei}[G^\star]$ and $\Ver_{f^\star}[G^\star]$, which can be easily emulated using the assumed data structures $\WeiSum_{f, \wei}[G]$ and $\Ver_f[G]$ that take $S$ as input to the query. 
	Upon query to $\RetVer_{f,\eps}[G]$ with a set $S\subseteq V(G)$, we temporarily add all the edges between $v^\star$ and all the vertices of $S$, query $\RetVer_{f^\star,\eps}[G^\star]$, and remove the added edges. We remark here that adding a single vertex $v^\star$ to the graph can increase its degeneracy by at most $1$.

	With \cref{lem:restore-by-sampling} extended, we may proceed to the proof. Towards an application of this lemma, we define a vertex function
	\[f(G,S,v)\coloneqq \begin{cases}1 &\textrm{if } v\notin N[S],\\
	0 &\textrm{otherwise.}\end{cases}\]
	A suitable implementation of a data structure $\WeiSum_{f,\wei}[G]$, for any weight function $\wei\colon V\to \Z$ fixed upon initialization, is provided by \cref{lem:wei-sums-far-deg}.
	Note that $f$ is binary, so there is no need to provide a data structure $\Ver_{f}[G]$.
	
	We may therefore apply the extended \cref{lem:restore-by-sampling} and obtain the data structure $\RetVer_{f,\eps}[G]$, which may serve as the output data structure $\farVertexDS{G}{\eps}$. The complexity guarantees follow directly from the extended \cref{lem:restore-by-sampling} combined with the guarantees provided by \cref{lem:wei-sums-far-deg}.
\end{proof}

%
%
%

\subsection{Dynamic domination in degenerate graphs}

We may now conclude the proof of the main result of this section.

\mainDomDeg*
\begin{proof}
	We may assume that $k\geq 2$, because for $k=1$ we may simply keep track of the degrees of vertices and report any vertex of degree $n-1$, if existent.
	
	Note that if $\Cc$ is the class of $d$-degenerate graphs, then every graph from $\Cc$ excludes the biclique $K_{d+1,d+1}$ as a subgraph. Hence by \cref{thm:sli-biclique} we get $\sli_1(\Cc)<3d+3$. Therefore, we may apply \cref{lem:summary-dom} with  \cref{thm:NearVertex-new-deg,thm:FarVertex-new-deg} plugged in. Note that in applications of \cref{thm:NearVertex-new-deg} we always have $|S|\leq k^{\sli_1(\Cc)}\leq k^{\Oh(d)}$ and in applications of \cref{thm:FarVertex-new-deg} we always have $|S|\leq k$, so thanks to the guarantees provided by guarantees provided by \cref{thm:NearVertex-new-deg,thm:FarVertex-new-deg}, the quantities $T,I,M$ mentioned in \cref{lem:summary-dom} can be set as $2^{\Oh(k)}\cdot k^{\Oh(d)}\cdot \log^3 n$, $2^{\Oh(d)}\cdot n\log n$, and $2^{\Oh(d)}\cdot n\log n$, respectively. The claimed complexity guarantees then follow directly from \cref{lem:summary-dom}.
\end{proof}

\section{Dynamic approximation of distance-\texorpdfstring{$1$}{1} dominating sets}\label{sec:apx}

Here we give a dynamic constant-factor approximation for the domination number in graphs of bounded degeneracy.
We first explain the general weak-reachability argument behind the approximation and only afterwards specialize it to the distance-$1$ representation used by the dynamic implementation. The reason for this is that the general argument applies to an arbitrary distance, but we are able to dynamically maintain the relevant structures only for distance $1$.
Throughout this section, for a graph $G$ we denote by $\mathsf{dom}_1(G)$ the minimum size of a (distance-$1$) dominating set in $G$. A \emph{set system} is a family of subsets of some universe $U$, and a \emph{packing} in a set system $\mathcal{F}$ is a subfamily $\mathcal{M}\subseteq \mathcal{F}$ consisting of pairwise disjoint members of $\mathcal{F}$. A packing $\mathcal{M}$ is \emph{maximal} if no superset of $\mathcal{M}$ is a packing; equivalently, every set $F\in \mathcal{F}\setminus \mathcal{M}$ intersects a member of $\mathcal{M}$. The \emph{arity} of a set system $\mathcal{F}$ is the maximum cardinality of a member of $\mathcal{F}$.

\subsection{Combinatorial core}

Fix a graph $G$, an ordering $\sigma$ of $V(G)$, and a radius $r\in\N$.
For $v\in V(G)$, let $\WReach_r[G,\sigma,v]$ be the set of vertices $u\in V(G)$ for which there exists a path $u=p_0,\ldots,p_\ell=v$ of length $\ell\le r$ such that $\sigma(u)\le \sigma(p_i)$ for every $i\in\set{1,\ldots,\ell}$ (in particular $v \in \WReach_r[G, \sigma, v]$).
Further, let $\wcol_r(G,\sigma)\coloneqq\max_{v\in V(G)} |\WReach_r[G,\sigma,v]|$ and $\mathcal{F}_r(\sigma)\coloneqq\setof{\WReach_r[G,\sigma,v]}{v\in V(G)}$.
The set system $\mathcal{F}_r(\sigma)$ is the general combinatorial source of the approximation: it works for every radius $r$, whereas the dynamic construction below will retain only the case $r=1$.
The next two lemmas formalize the hitting-packing argument showing that weak reachable sets yield a constant-factor approximation of minimum distance-$r$ domination.

\begin{lemma}\label{lem:apx-wreach-domination}
Let $\mathcal{M}\subseteq \mathcal{F}_r(\sigma)$ be a maximal packing in $\mathcal{F}_r(\sigma)$.
Then $D\coloneqq\bigcup_{X\in \mathcal{M}} X$ is a distance-$r$ dominating set in $G$.
\end{lemma}
\begin{proof}
Take any vertex $x\in V(G)$.
If $x\in D$, then $x$ is distance-$r$ dominated by itself.
So suppose that $x\notin D$.
As $\mathcal{M}$ is maximal, the set $\WReach_r[G,\sigma,x]\in \mathcal{F}_r(\sigma)$ intersects some member of $\mathcal{M}$.
Choose $X\in \mathcal{M}$ and $d\in \WReach_r[G,\sigma,x]\cap X$.
By the definition of $\WReach_r[G,\sigma,x]$, there is a $d$-$x$ path of length at most $r$.
Since $d\in X\subseteq D$, this shows that $x$ is distance-$r$ dominated by a vertex of~$D$.
\end{proof}

\begin{lemma}\label{lem:apx-wreach-transversal}
Let $D^*$ be a distance-$r$ dominating set in $G$. Define $T\coloneqq\bigcup_{d\in D^*} \WReach_r[G,\sigma,d]$.
Then $T$ intersects every set in $\mathcal{F}_r(\sigma)$.
\end{lemma}
\begin{proof}
Fix any vertex $v\in V(G)$.
Since $D^*$ is a distance-$r$ dominating set, there exists $d\in D^*$ with $\dist_G(v,d)\le r$.
Let $P$ be a shortest $v$-$d$ path and let $u$ be the vertex of $P$ with the minimum position in $\sigma$.
The subpath of $P$ from $u$ to $v$ has length at most $r$, and every internal vertex of this subpath has position at least $\sigma(u)$.
Hence $u\in \WReach_r[G,\sigma,v]$.
By the same argument applied to the subpath of $P$ from $u$ to~$d$, we obtain $u\in \WReach_r[G,\sigma,d]\subseteq T$.
Therefore $\WReach_r[G,\sigma,v]\cap T\neq \emptyset$.
\end{proof}

\begin{lemma}\label{lem:apx-wreach-static}
Let $\mathcal{M}\subseteq \mathcal{F}_r(\sigma)$ be a maximal packing in $\mathcal{F}_r(\sigma)$, $D\coloneqq\bigcup_{X\in \mathcal{M}} X$, and $D^*$ be a minimum distance-$r$ dominating set in $G$.
Then $D$ is a distance-$r$ dominating set in $G$ satisfying \mbox{$|D|\le \wcol_r(G,\sigma)^2\cdot |D^*|$.}
Consequently, whenever $\wcol_r(G,\sigma)\le c$, the set $D$ is a distance-$r$ dominating set of size at most $c^2$ times the optimum.
\end{lemma}
\begin{proof}
By \cref{lem:apx-wreach-domination}, the set $D$ is a distance-$r$ dominating set.

Let $T\coloneqq\bigcup_{d\in D^*} \WReach_r[G,\sigma,d]$.
By \cref{lem:apx-wreach-transversal}, every set in $\mathcal{M}$ intersects $T$.
Since $\mathcal{M}$ is a packing, its sets are pairwise disjoint, hence $|\mathcal{M}|\le |T|$.
Using the definition of $\wcol_r(G,\sigma)$, we infer that
\[|\mathcal{M}|\le |T|\le \sum_{d\in D^*} |\WReach_r[G,\sigma,d]|\le \wcol_r(G,\sigma)\cdot |D^*|.\]
Again by the definition of $\wcol_r(G,\sigma)$, every set in $\mathcal{M}$ has size at most $\wcol_r(G,\sigma)$, so
\[|D| = \sum_{X\in \mathcal{M}} |X| \le \wcol_r(G,\sigma)\cdot |\mathcal{M}| \le \wcol_r(G,\sigma)^2\cdot |D^*|.\qedhere\]
\end{proof}

Thus, \cref{lem:apx-wreach-static} already contains the whole approximation argument.
Note that if $\vec{G}_\sigma$ is the orientation obtained by directing every edge towards the smaller endpoint in $\sigma$, then \[\WReach_1[G,\sigma,v]=\set{v}\cup N^+_{\vec{G}_\sigma}(v)\qquad\textrm{for every vertex }v.\]
In our dynamic data structure, we do not maintain an order $\sigma$ explicitly. Instead, we maintain an orientation of bounded outdegree using the data structure of Brodal and Fagerberg~\cite{BrodalF99}.
Therefore, we now adjust the result of \cref{lem:apx-wreach-static} to the setting of distance-$1$ dominating sets and orientations of bounded outdegree.
The sets $S(v)$ introduced below are not a substitute for the preceding $\WReach_r$ argument; they are exactly its distance-$1$ implementation-level representation.

Let $G$ be a graph and let $\vec{G}$ be an orientation of $G$.
Recall that for every vertex $v\in V(G)$, by $N^+_{\vec{G}}(v)$ we denote the out-neighborhood of $v$ in $\vec{G}$. Define \[S(v)\coloneqq\set{v}\cup N^+_{\vec{G}}(v).\]
We regard the family ${\cal S}\coloneqq \setof{S(v)}{v\in V(G)}$ as a set system over the vertex set $V(G)$, with member per each vertex of $G$.

\begin{lemma}\label{lem:apx-domination}
Let $M\subseteq V(G)$ be such that $\setof{S(v)}{v\in M}$ is a maximal packing in the set system ${\cal S}$.
Then $D\coloneqq\bigcup_{v\in M} S(v)$ is a dominating set in $G$.
\end{lemma}
\begin{proof}
Take any vertex $x\in V(G)$.
If $x\in D$, then $x$ is dominated.
Assume now that $x\notin D$.
By maximality, the set $S(x)$ is not disjoint from at least one set $S(v)$ with $v\in M$.
Choose $v\in M$ and $y\in S(x)\cap S(v)$.
Since $y\in S(v)\subseteq D$, we have $y\neq x$.
Therefore $y\in N^+_{\vec{G}}(x)$, which implies that $xy\in E(G)$.
Hence $x$ is dominated by the vertex $y\in D$.
\end{proof}

\begin{lemma}\label{lem:apx-transversal}
Let $D^*$ be a dominating set in $G$ and define $T\coloneqq\bigcup_{d\in D^*} S(d)$.
Then $T$ intersects every set $S(v)$ with $v\in V(G)$.
\end{lemma}
\begin{proof}
Fix any vertex $v\in V(G)$.
Since $D^*$ is a dominating set, there exists $d\in D^*$ with $d\in N[v]$.
If $d=v$, then $v\in S(v)\cap T$.
Assume therefore that $d\neq v$.
If the edge $vd$ is oriented in $\vec{G}$ from $v$ to $d$, then $d\in S(v)\cap S(d)\subseteq S(v)\cap T$.
Otherwise the edge $vd$ is oriented from $d$ to $v$, and then $v\in S(v)\cap S(d)\subseteq S(v)\cap T$.
Thus $S(v)\cap T\neq \emptyset$.
\end{proof}

\begin{lemma}\label{lem:apx-static}
Assume that $\vec{G}$ has maximum out-degree at most $\Delta$.
Let $M\subseteq V(G)$ be such that the family $\setof{S(v)}{v\in M}$ is a maximal packing in ${\cal S}$.
Then $D\coloneqq\bigcup_{v\in M} S(v)$ is a dominating set in $G$ and satisfies $|D|\le (\Delta+1)^2\cdot \mathsf{dom}_1(G)$.
\end{lemma}
\begin{proof}
By \cref{lem:apx-domination}, the set $D$ is a dominating set.

Let $D^*$ be a minimum dominating set in $G$ and let $T\coloneqq\bigcup_{d\in D^*} S(d)$.
By \cref{lem:apx-transversal}, every set $S(v)$ with $v\in M$ intersects $T$.
Since the family $\setof{S(v)}{v\in M}$ is a packing in $\cal S$, these sets are pairwise disjoint, and hence $|M|\le |T|$.
Every set $S(d)$ has size at most $\Delta+1$, so \[|M|\le |T|\le \sum_{d\in D^*} |S(d)|\le (\Delta+1)\cdot |D^*| = (\Delta+1)\cdot \mathsf{dom}_1(G).\]
Using again that $|S(v)|\le \Delta+1$ for every $v\in V(G)$, we obtain \[|D| = \sum_{v\in M} |S(v)|\le (\Delta+1)\cdot |M| \le (\Delta+1)^2\cdot \mathsf{dom}_1(G).\qedhere\]
\end{proof}

\subsection{Dynamic data structure}

We now propose a dynamic data structure that maintains an approximate domination number of a graph of bounded degeneracy based on the approximation method yielded by \cref{lem:apx-static}. We will exploit the orientation data structure of \cref{thm:bf} and the following result by Assadi and Solomon.


\begin{theorem}[Assadi--Solomon maximal packing data structure,~\cite{AssadiSolomon21}]\label{thm:AS}
	For every $\Delta\in \N$, there is a randomized dynamic data structure that, against an oblivious adversary, maintains a maximal packing in a fully dynamic labelled set system over a universe $U$ and of arity at most $\Delta$. It supports insertions and deletions of sets in expected amortized time $\Oh(\Delta^2)$, with the same bound holding with high probability, and can be implemented using $\Oh(|U|+\Delta N)$ words of space, where $N$ is the number of maintained labelled sets.
\end{theorem}

Assadi and Solomon phrase their result in the language of maximal matchings in hypergraphs, which is equivalent to our terminology of maximal packings in set systems. 

Throughout this subsection, we assume $n\ge 2$ and $d\ge 1$, and we consider a fully dynamic $d$-degenerate graph $G$ on the fixed vertex set $V=\set{1,\ldots,n}$.
Hence, we may invoke the Brodal--Fagerberg orientation data structure (\cref{thm:bf}) on $G$, and thus maintain also an orientation $\vec{G}$ of $G$ with maximum out-degree at most $4d$.
We follow the convention that $G_t$ is the graph $G$ at time step $t$, and consequently $\vec{G}_t$ is the maintained orientation of $G_t$.
For every vertex $v\in V$, we let $S_t(v)\coloneqq\set{v}\cup N^+_{\vec{G}_t}(v)$.

\paragraph{State of the data structure.}
At time step $t$, the complete state of the data structure consists of the following objects.
\begin{itemize}
\item \emph{Orientation component}: The data structure provided by \cref{thm:bf}, storing the current graph $G_t$ together with its orientation $\vec{G}_t$ and, on every graph update, reporting the updated edge together with the list of all edges reoriented during this update.
\item Set system $\mathcal{S}_t\coloneqq\setof{S_t(v)}{v\in V}$ on the universe $V$, stored explicitly. Every set $S_t(v)\in \mathcal{S}_t$ is labelled with a permanent label $v$. For every $v\in V$, we store the current members of $S_t(v)$ explicitly, and we also store the size of $S_t(v)$. Equivalently, we store the permanent element $v$ together with the current out-neighborhood $N^+_{\vec{G}_t}(v)$ and the cardinality of $\{v\}\cup N^+_{\vec{G}_t}(v)$.
\item \emph{Packing component}: The data structure provided by \cref{thm:AS} for the set system $\mathcal{S}_t$, storing a maximal packing $\mathcal{M}_t$ in $\mathcal{S}_t$ as an iterable list of the labels of the members of $\mathcal{M}_t$. Note that by \cref{lem:apx-static}, the set $D_t\coloneqq \bigcup_{X\in \mathcal{M}_t} X$ is a dominating set in $G_t$ of size at most $(4d+1)^2\cdot \mathsf{dom}_1(G_t)$. Therefore, together with $\mathcal{M}_t$ we also store the cardinality of $D_t$. This can be easily updated within the data structure of \cref{thm:AS} upon updates to $\mathcal{M}_t$, as with every set in $\mathcal{S}_t$ we explicitly store also its cardinality.
\end{itemize}
No further objects are needed. In particular, the current dominating set is not stored explicitly: it is represented implicitly by the packing $\mathcal{M}_t$.

\paragraph{Interface.}
As announced in \cref{thm:apx-dynamic}, our data structure supports five methods.
\begin{itemize}
\item $\init{n}$: initialize the data structure on $G_0$ being the edgeless graph on vertex set $V$. 
\item $\add{xy}$: insert the edge $xy$ and update all maintained components.
\item $\remove{xy}$: delete the edge $xy$ and update all maintained components.
\item $\mathtt{SizeDominatingSet()}$: return the cardinality of the current dominating set represented by the packing component.
\item $\mathtt{DominatingSet()}$: returns the current dominating set represented by the packing component.
\end{itemize}
Their guarantees and complexities are summarized in \cref{thm:apx-dynamic}. We now describe the implementation.

\paragraph{Synchronization of the orientation and packing components.}
Whenever a graph update touches an edge $xy$, the set system $\mathcal{S}_t=\{S_t(v)\colon v\in V\}$ needs to be modified only by updating the sets indexed by $x$ and $y$.
More precisely, creating an oriented edge $(x,y)$ inserts $y$ into the $S_t(x)$, deleting an oriented edge $(x,y)$ deletes $y$ from $S_t(x)$, and flipping the orientation of $xy$ from $(x,y)$ to $(y,x)$ is implemented as deleting $y$ from $S_t(x)$ and inserting $x$ into $S_t(y)$.
Thus, every changed oriented edge yields at most two updates to $\mathcal{S}_t$, each also relayed immediately to the packing component.
In the packing component, each modification of a labelled set is implemented by deleting its old version and inserting its updated version.
By \cref{thm:bf}, one update to $G_t$ produces $\Oh(\log n)$ updates to $\vec{G}_t$ in the amortized sense: the updated graph edge contributes one update to $\mathcal{S}_t$, and every reorientation performed during the update contributes two more.

\paragraph{Method $\init{n}$} is implemented as follows:
\begin{enumerate}
\item Initialize the orientation component on the edgeless graph on $V=\{1,\ldots,n\}$.
\item For every vertex $v\in V$, create the singleton set $S(v)=\set{v}$ labelled with $v$.
\item Initialize the packing component on the set system consisting of these $n$ labelled sets.
\end{enumerate}
After these steps the orientation is correct, the stored set system is exactly $\mathcal{S}_0=\setof{\set{v}}{v\in V}$, and the packing component stores a maximal packing in $\mathcal{S}_0$.
The initialization time is $\Oh(n)$, because we create $n$ singleton sets.

\paragraph{Method $\add{xy}$} is implemented as follows:
\begin{enumerate}
\item Execute the insertion of $xy$ in the orientation component.
\item Read the newly oriented copy of $xy$ and the list of edges reoriented during this update.
\item For every changed oriented edge (inserted, removed, or reoriented), relay the modification to appropriately update $\mathcal{S}_t$ and the packing component, as described in the Synchronization paragraph above.
\end{enumerate}
Only the sets labelled with endpoints of changed oriented edges are modified, hence after Step~3 every set labelled with $v\in V$ stored in the set system again equals the current set $S_t(v)$, hence the set system is indeed equal to $\mathcal{S}_t$.
Since the packing component is updated after every modification of $\mathcal{S}_t$, at the end it stores a maximal packing in the updated set system.
The amortized number of updates to $\mathcal{S}_t$ is $\Oh(\log n)$.
Each of them touches a set of size at most $4d+1$, hence, by \cref{thm:AS}, it costs an expected amortized time of $\Oh((4d+1)^2)=\Oh(d^2)$.
Therefore $\add{xy}$ runs in expected amortized time $d^{\O(1)}\cdot \log n$.

\paragraph{Method $\remove{xy}$} is implemented as follows:
\begin{enumerate}
\item Execute the deletion of $xy$ in the orientation component.
\item Read the deleted oriented copy of $xy$ and the list of edges reoriented during this update.
\item For every changed oriented edge, relay the modification to appropriately update $\mathcal{S}_t$ and the packing component, as described in the Synchronization paragraph above.
\end{enumerate}
The same argument as in the $\add{xy}$ method shows that after Step~3 the stored set system is exactly $\mathcal{S}_t$, and the packing component stores a maximal packing in it.
Again there are amortized $\Oh(\log n)$ updates to $\mathcal{S}_t$, each affecting a set of size at most $4d+1$, so $\remove{xy}$ also runs in expected amortized time $d^{\Oh(1)}\cdot \log n$.

\paragraph{Method $\mathtt{SizeDominatingSet()}$} is implemented by just returning the cardinality of $D_t=\bigcup_{X \in \mathcal{M}_t} X$, stored explicitly within the packing component.

\paragraph{Method $\mathtt{DominatingSet()}$} is implemented as follows.
We traverse the set system $\mathcal{M}_t$ (stored as an iterable list of labels of the members of $\mathcal{M}_t$) and we output the union $D_t=\bigcup_{X\in \mathcal{M}_t} X$.
No duplicate elimination is needed, because the members of $\mathcal{M}_t$ are pairwise disjoint.
Since the query outputs exactly the vertices contained in the members of $\mathcal{M}_t$, its running time is $\Oh(|D_t|)$.

\paragraph{Memory usage.}
Let $m_t\coloneqq|E(G_t)|$. By \cref{thm:bf},
the orientation component uses $\Oh(n+m_t)$ words.
The explicitly stored set system $\mathcal{S}_t$ also uses $\Oh(n+m_t)$ words, because the total sum of cardinalities of the members of $\mathcal{S}_t$ is exactly $n+m_t$. 
The packing component is maintained over the universe $V$, with $n$ labelled sets of arity at most $4d+1$, so by \cref{thm:AS}, it uses $\Oh(n+(4d+1)n)=\Oh(dn)$ words.
The dominating set returned by $\mathtt{DominatingSet()}$ is not stored explicitly, so it contributes no additional space. As we have that $m_t \le dn$, we conclude that the whole data structure occupies $\Oh(n+m_t)+\Oh(n+m_t)+\Oh(dn)=\Oh(dn)$ words.


The description provided above amounts to the proof of the result promised in \cref{sec:intro}, which concludes this section.

\mainApx*

\section{Conclusions}\label{sec:conclusions}

Let us conclude by discussing a handful of open questions.
\begin{itemize}
	\item The central open question is that of Dvo\v{r}\'ak and T\r{u}ma~\cite{DvorakT13}: Can one design a dynamic data structure for $\FO$ model-checking on any fixed  class of bounded expansion with polylogarithmic amortized update time? A natural approach to this question is to try to develop dynamic counterparts for various techniques of Sparsity underlying the known static algorithms~\cite{DvorakKT13,grohe2011methods,PilipczukST18}, such as transitive--fraternal augmentations, weak coloring numbers, or low treedepth colorings. After a few years of attempts, we consider this route hopeless. Instead, we believe that developing a new model-checking algorithm for $\FO$ on classes of bounded expansion, which would be more in the spirit of progressive exploration, could be a way to go.
	\item While the data structure of Dvo\v{r}\'ak and T\r{u}ma~\cite{DvorakT13} can be also deployed on any nowhere dense graph class, with the resulting amortized update time becoming $\Oh_{\Cc,H,\delta}(n^\delta)$ for any $\delta>0$, this is not the case for our data structures of \cref{thm:main-domset-be,thm:main-indset-be}. The reason is that in the nowhere dense setting, the fraternal augmentations discussed in \cref{sec:augmentations} have maximum outdegree $\Oh_{\Cc,r,\delta}(n^\delta)$ instead of a constant, and hence the Inclusion--Exclusion formula postulated in \cref{lem:IE} may run over $2^{\Oh_{\Cc,r,\delta}(n^\delta)}$ elements of $\cal M$, which is superpolynomial. We leave it open whether \drds{r} and \dris{r} on nowhere dense classes admit dynamic data structures with amortized update time $\Oh_{\Cc,r,k,\delta}(n^\delta)$, for any $\delta>0$. However, we remark that the \cref{lem:IE} is the only reason why this argument does not lift to nowhere dense classes and multiple other results, in particular \cref{thm:dvorak-tuma-examples}, can be lift verbatim.
	\item Can our data structures be derandomized? At this point, we crucially rely on randomization in the fingerprint retrieval technique, to turn counting data structures into example-reporting data structures.
	\item It remains open whether the result of \cref{thm:apx-dynamic} can be extended to distance-$r$ domination for $r>1$ under the assumption that $G$ belongs at all times to a fixed class of bounded expansion $\Cc$. As discussed in \cref{sec:overview-other}, this is connected to the dynamic maintenance of the set system of weak $r$-reachability sets.
\end{itemize}

\bibliographystyle{plain}
\bibliography{references}

\appendix
\section{The proof of \cref{lem:weighted-dvorak}} \label{sec:app-weighted-dvorak}


The proof of \cref{lem:weighted-dvorak} is heavily based on the proof of its original unweighted version from \cite{DvorakTumaArxiv}. (Note that~\cite{DvorakT13} is the conference version of this work, throughout this section we will mostly rely on the full version available on arXiv~\cite{DvorakTumaArxiv}.) Incorporating weights into it requires closely following the original proof and plugging weights wherever required. The original proof goes through a number of reductions --- reducing counting induced subgraph isomorphisms to counting subgraph isomorphisms, reducing counting subgraph isomorphisms to counting homomorphisms, reducing counting homomorphisms to counting homomorphisms from an elder graph, and finally solving that last problem in a dynamic setting. In this section we are going to outline that process in reverse, explaining how to add weights in all necessary~places.

\subsection{Counting weighted augmented homomorphisms}
 
A directed graph $H$ is called an \textit{elder graph} if it has the property that $(u, w), (v, w) \in E(H)$ imply that $u$ and $v$ are adjacent as well, that is, either $(u, v) \in E(H)$ or $(v, u) \in E(H)$.
In \cite[Theorem~14]{DvorakTumaArxiv}, Dvo\v{r}\'ak and T\r{u}ma prove that for a fixed directed connected elder graph $H$ and dynamic directed graph $G$ with maximum indegree\footnote{In \cite{DvorakTumaArxiv}, orientations have bounded maximum indegree instead of outdegree. In this paper we work with bounded outdegree orientations, as this is the prevalent convention in the literature.} at most $D$, where both have edges colored by $\{0, 1, \ldots, k\}$, it is possible to maintain the number of homomorphisms from $H$ to $G$ in $\Oh(f(|H|, D))$ time complexity per update, for some function $f$. It will turn out that the value of $D$ in our application will be $\Oh_\Cc(1)$, so $f(|H|, D) = \Oh_{\Cc, H}(1)$. 

We generalize their statement to the following weighted version:

\begin{lemma} \label{lem:weighted-ahom}
	Let $D$ be an integer, $H$ be a fixed directed connected elder graph, $G$ be a dynamic directed graph with maximum indegree at most $D$, where both have edges colored by $\{0, 1, \ldots, k\}$ and $\wei \colon V(H) \times V(G) \to \Z$ is a fixed weight function. There exists a data structure $\mathsf{WeightedAHom}_{H, k, D, \wei}(G)$ that maintains $\val_\wei(\Hom(H, G))$. Each update to the data structure takes $f(|H|, D)$ time complexity, it is initialized for an edgeless graph in $\Oh(n)$ time and the space complexity of it is $\Oh(f(|H|, D) \cdot |G|)$, for some function $f$.
\end{lemma}

\begin{proof}
	The generalization to the weighted version requires just a few minor modifications to the proof of the original Theorem 14 of \cite{DvorakTumaArxiv}, therefore we are only going to highlight the differences between them and assume familiarity of the reader with that proof. We are also going to use the terminology of that proof (that is, notions like \textit{vineyard, clan, extended clan, ghost}).
	
	During the initialization of the data structure we fix any vineyard $T$ for $H$. For a clan $C$ and vertices $v, w_1, \ldots, w_m \in V(G)$, by $\Hom_{H, T}(C, v, w_1, \ldots, w_m, G)$ we are going to denote the set of homomorphisms from an extended clan $C^*$ to $G$ such that $r(C)$ is mapped to $v$ and ghosts $g_1, \ldots, g_m$ of $C$ are mapped to $w_1, \ldots, w_m$. A value of a partial homomorphism $\phi$ from $H$ to $G$ belonging to that set is defined naturally as the product of $\wei(x_i, y_i)$ over all $x_i \in V(C^*)$, where $\phi(x_i) = y_i$, and the value of a whole set is defined as the sum of values of its elements.
	
	For each clan $C \neq V(H)$ with $m$ ghosts and each $m$-tuple of vertices $w_1, \ldots, w_m$ of $G$, in the unweighted version the number $S(C, w_1, \ldots, w_m) = \sum_{v \in N_1^+(w_1)} |\Hom_{H, T}(C, v, w_1, \ldots, w_m, G)|$ was recorded. However, we will record the number \[S_\wei(C, w_1, \ldots, w_m) \coloneqq \sum_{v \in N_1^+(w_1)} \val_\wei(\Hom_{H, T}(C, v, w_1, \ldots, w_m, G))\] instead. 
			
	The original proof expresses the difference of the sets $\Hom_{H, T}(C, v, w_1, \ldots, w_m, G)$ before and after an addition of an edge $e$ to $G$ as a disjoint sum of some sets, where each of them results from a different guess of which edges of $H$ will be mapped to $e$. In each of these cases, a partial homomorphism $\phi$ is determined and the problem of extending it to a homomorphism of full $C^*$ decomposes to a few independent subproblems, hence the set of ways to extend the fixed partial homomorphisms resulting from that guess to a full homomorphisms of $C^*$ can be expressed as a cartesian product of sets of homomorphisms for these smaller instances. If $C_i$ and $g_1^i, \ldots, g_{m_i}^i$ are the clan and the ghosts of the $i$-th subproblem, then the number of ways to extend $\phi$ to a full homomorphisms of $C^*$ can be expressed as $\prod_{i=1}^{t} S(C_i, \phi(g_1^i), \ldots, \phi(g_{m_i}^i))$, whereas the incurred contribution to the value of all accounted homomorphisms will be $\val_\wei(\phi) \cdot \prod_{i=1}^{t}S_\wei(C_i, \phi(g_1^i), \ldots, \phi(g_{m_i}^i))$.
	
	We remark that the calculation above crucially relies on the properties of $\val_\wei$ that $\val_\wei(X \uplus Y) = \val_\wei(X) + \val_\wei(Y)$ and $\val_\wei(P \times Q) = \val_\wei(P) \cdot \val_\wei(Q)$, where $P$ and $Q$ are acting on disjoint subsets $V_P$ and $V_Q$ of $V(H)$ and each element $\phi = (\phi_P, \phi_Q) \in P \times Q$ is understood as a unique mapping on $V_P \uplus V_Q$ such that $\phi|_{V_P} = \phi_P$ and $\phi|_{V_Q} = \phi_Q$. (Here, $\uplus$ denotes the disjoint union.) 
\end{proof}

\subsection{Counting weighted homomorphisms}
Before proceeding with the rest of the proof, we are going to recall an important definition of a \textit{$0$-contraction} from \cite{DvorakTumaArxiv}.

\begin{definition}[\cite{DvorakTumaArxiv}]
	Let $F'$ be a directed graph with edges colored by $\{0, 1, \ldots, k\}$. Let $P$ be a partition of vertices of $F'$ such that:
	\begin{itemize}
		\item for every $p \in P$, the subgraph of $F'$ induced by $p$ is connected and contains only edges colored $0$; and
		\item if $p_1, p_2 \in P$ are distinct, $u, u' \in p_1, v, v' \in p_2$ and $(u, v)$ is an edge, then $(v', u')$ is not an edge, and if $(u', v')$ is an edge, then it has the same color as $(u, v)$.
	\end{itemize}
	Let $F''$ be the directed graph with edges colored by $\{0, 1, \ldots, k\}$, such that $V(F'') = P$ and $(p_1, p_2) \in E(F'')$ if and only if $(v_1, v_2) \in E(F')$ for some $v_1 \in p_1$ and $v_2 \in p_2$; and in this case, $(p_1, p_2)$ and $(v_1, v_2)$ have the same color. That is, $F''$ is obtained from $F'$ by identifying the vertices in each part of $P$ and suppressing the parallel edges and loops, and we also remember which vertices of $F'$ correspond to each vertex of $F''$. We say that $F''$ is a \textit{$0$-contraction} of $F'$.
\end{definition}

We also recall the rephrased version of maintaining iterated fraternal augmentations from~\cite{DvorakTumaArxiv}:
\begin{theorem}[{\cite[Theorem 4]{DvorakTumaArxiv}}] \label{thm:maintaining-augmentation}
	Let $\Cc$ be a graph class of bounded expansion and $h \in \N$. For a graph $G$, let its \emph{$h$-th augmentation} be any directed graph obtained by taking any orientation $\vec G$ of $G$ and iterating fraternal augmentation on it $h$ times. There exists $C \in \N$ dependent only on $\Cc$ and $h$ such that there is a data structure maintaining some $h$-th augmentation $\tilde{G}_h$ of a dynamic graph $G \in \Cc$ such that:
	\begin{itemize}
		\item the maximum outdegree of $\tilde{G}_h$ is at most $C$;
		\item an edge can be added to $G$ in $\Oh(C \log^{h+1} n)$ amortized time; and
		\item an edge can be removed from $G$ in $\Oh(C)$ amortized time.
	\end{itemize}
\end{theorem}

We assume that the edges of non-augmented graphs are colored with colors $\{1, \ldots, k\}$ and edges introduced during the augmentations are colored with a new color $0$.

We say that a graph $G_2$ can be obtained from $G_1$ by \emph{recoloring zeros} if and only if $V(G_1) = V(G_2)$, $E(G_1) = E(G_2)$, however for every edge $(u, v) \in E(G_1)$ such that its colors in $G_1$ and $G_2$ are different, we have that its color in $G_1$ is $0$.

Moreover, for a graph $H$ with edges colored by $\{1, \ldots, k\}$, let $\mathcal{H}^e$ denote the set of all possible $0$-contractions of graphs obtained by recoloring zeros from all possible $h$-th augmentations of $H$, where $h = {|H| \choose 2} - 2$.

Dvo\v{r}\'ak and T\r{u}ma prove the following:

\begin{lemma}[{\cite[Lemma 11]{DvorakTumaArxiv}}] \label{lem:unweighted-hom-proj}
	Let $H$ and $G$ be graphs with edges colored by $\{1, \ldots, k\}$ and let $h = {|H| \choose 2} - 2$. If $G'$ is an $h$-th augmentation of $G$, then $|\Hom(H, G)| = \sum_{H' \in \mathcal{H}^e} |\Hom(H', G')|$.
\end{lemma}

They show this equality by showing a natural bijection $\Phi$ between $\Hom(H, G)$ and pairs $(H', \phi')$, where $H' \in \mathcal{H}^e$ and $\phi' \in \Hom(H', G')$. As each $H' \in \mathcal{H}^e$ can be shown to be elder (\cite[Lemma 12]{DvorakTumaArxiv}), a clear consequence of the original unweighted variant of \cref{lem:weighted-dvorak}, \cref{lem:unweighted-hom-proj} and \cref{thm:maintaining-augmentation} was that there exists a data structure that efficiently maintains the number of homomorphisms for a dynamic graph $G$ of bounded expansion. Our goal now will be to adjust \cref{lem:unweighted-hom-proj} to the weighted setting.

Let $\wei\colon V(H) \times V(G) \to \Z$ be a fixed weight function and let $H' \in \mathcal{H}^e$. We recall that $V(H')$ is a partition of $V(H)$. Let us now define a weight function $\wei_{H'} \colon V(H') \times V(G') \to \Z$ by the following formula: $\wei_{H'}(p, u) \coloneqq \prod_{v \in p} \wei(v, u)$. Let $\phi' \in \Hom(H', G')$ and let $\phi \in \Hom(H, G)$ be such that $\Phi(\phi) = (H', \phi')$. The bijection $\Phi$ satisfies that for all $v\in V(H)$ we have $\phi(v) = \phi'(p_v)$, where $p_v$ is the part of the partition $V(H')$ that contains $v$.
Therefore, we have that \begin{align*}\val_\wei(\phi) & = \prod_{v \in V(H)} \wei(v, \phi(v)) = \prod_{p \in V(H')} \prod_{v \in p} \wei(v, \phi(v))\\ &= \prod_{p \in V(H')} \prod_{v \in p} \wei(v, \phi'(p)) = \prod_{p \in V(H')} \wei_{H'}(p, \phi'(p)) = \val_{\wei_{H'}}(\phi').
	\end{align*} As a consequence we observe the following:

\begin{lemma} \label{lem:weighted-hom-proj}
	Let $H$ and $G$ be graphs with edges colored by $\{1, \ldots, k\}$, $h = {|H| \choose 2} - 2$, and $\wei \colon V(H) \times V(G) \to \Z$ be a weight function. If $G'$ is an $h$-th augmentation of $G$, then \[\val_\wei(\Hom(H, G)) = \sum_{H' \in \mathcal{H}^e} \val_{\wei_{H'}}(\Hom(H', G')).\]
\end{lemma}

With the above statement, we are able to conclude the following lemma:

\begin{lemma} \label{lem:weighted-hom}
	Let $H$ be a fixed graph, $\G$ be a class of graphs of bounded expansion, and $G \in \G$ be a dynamic graph, where edges of $H$ and $G$ are colored with colors $\{1, \ldots, k\}$. Let also $\wei \colon V(H) \times V(G) \to \Z$ be a fixed weight function. Then, there exists a data structure $\mathsf{WeightedHom}_{H, k, \wei}(G)$ that is able to determine $\val_\wei(\Hom(H, G))$ after each update. 
	The data structure processes each update in $\Oh_{\G, H, k}(\log^h |G|)$,  where $h = {|H| \choose 2} - 1$, it can be initialized for an edgeless graph in $\Oh_{\G, H, k}(|G|)$ time and its space complexity is~$\Oh_{\G, H, k}(|G|)$.
\end{lemma}

\begin{proof}
	This proof combines \cref{lem:weighted-hom-proj}, \cref{lem:weighted-ahom} and \cref{thm:maintaining-augmentation} in an identical way as the analogous unweighted data structure from \cite{DvorakTumaArxiv}.
	
	We note that if $H$ is not connected, then $\val_\wei(\Hom(H, G))$ is just equal to the product \[\prod_{H' \in \mathsf{CC}(H)} \val_\wei(\Hom(H', G))\] taken over the set $\mathsf{CC}(H)$ of connected components of $H$. Hence, we may assume that $H$ is connected.
	
	The data structure $\mathsf{WeightedHom}_{H, k, \wei}(G)$ maintains an $h$-th augmentation $G'$ of $G$, as described in \cref{thm:maintaining-augmentation} and for each $H' \in \mathcal{H}^e$ it additionally maintains a data structure $\mathsf{WeightedAHom}_{H', k, D, \wei}(G)$, where $C$ is the integer from \cref{thm:maintaining-augmentation}. As $|\mathcal{H}^e|$ is bounded by a function of $H$ and $k$ only, its size is constant. Each addition of an edge to $G$ results in $\Oh(C \log^{h+1}|G|)$ changes in $G'$, each removal results in $\Oh(C)$ such changes (amortized).
	 All changes in $G'$ are relayed to all $\mathsf{WeightedAHom}$ structures. The value $\val_\wei(\Hom(H, G))$ is derived as the sum of all values $\val_{\wei_{H'}}(\Hom(H', G'))$.
\end{proof}

\subsection{Counting weighted subgraph isomorphisms}

Next, we proceed to analyzing $\val_\wei(\Sub(H, G))$.  First, we  recall the definition of \textit{projections} from \cite{DvorakTumaArxiv}:
\begin{definition} \label{def:projections}
	Consider a graph $H$ with colored edges, and let $P$ be a partition of $V(H)$ such that
	\begin{itemize}
		\item each element of $P$ induces an independent set in $H$, and
		\item for all $p_1, p_2 \in P, u, u' \in p_1$ and $v, v' \in p_2$, if both $uv$ and $u'v'$ are edges of $H$, then they have the same color.
	\end{itemize}
	Let $H'$ be the graph obtained from $H$ by identifying the vertices in each part of $P$ and suppressing the parallel edges. We say that $H'$ is a \textit{projection} of $H$. 
\end{definition}
Let $\mathcal{H}^p$ denote the set of all projections $H'$ of $H$.

Dvo\v{r}\'ak and T\r{u}ma connect counting the number of subgraph isomorphisms with counting the number of homomorphisms through the following statement:

\begin{lemma}[{\cite[Lemma 7]{DvorakTumaArxiv}}] \label{lem:unweighted-sub-to-homo}
	For every graph $H$ with colored edges, there exist integer coefficients $\alpha_{H'}$ such that for every graph $G$ with colored edges, $|\Sub(H, G)| = \sum_{H \in \mathcal{H}^p} \alpha_{H'} |\Hom(H', G)|$.
\end{lemma}

For a projection $H' \in \mathcal{H}^p$ and a homomorphism $\phi' \colon V(H') \to V(G)$ we can define its \emph{uncontracted version} $\Phi(\phi') = \phi$ by $\phi(v) \coloneqq \phi'(p_v)$, where $v \in p_v$. Similarly as in the case of $0$-contractions, we can define the projected weight function $\wei_{H'}(p, u) = \prod_{v \in p} \wei(v, u)$ and show that $\val_{\wei_{H'}}(\phi') = \val_\wei(\phi)$.

We generalize \cref{lem:unweighted-sub-to-homo} to the weighted setting in the following way:
\begin{lemma} \label{lem:weighted-sub-to-homo}
	For every graph $H$ with colored edges, there exist integer coefficients $\alpha_{H'}$ such that for every graph $G$ with colored edges and every weight function $\wei\colon V(H) \times V(G) \to \Z$,
	$\val_\wei(\Sub(H, G)) = \sum_{H \in \mathcal{H}^p} \alpha_{H'} \val_{\wei_{H'}}(\Hom(H', G))$.
\end{lemma}

\begin{proof}
	The coefficients $\alpha_{H'}$ are unsurprisingly going to be the same as in \cref{lem:unweighted-sub-to-homo}. For these coefficients, the proof of this lemma basically says that if we reformulate the supposed equality $|\Sub(H, G)| = \sum_{H \in \mathcal{H}^p} \alpha_{H'} |\Hom(H', G)|$ as $|\Sub(H, G)| = \sum_{H \in \mathcal{H}^p} \alpha_{H'} |\Phi(\Hom(H', G))|$, then each $\phi \in \Hom(H, G)$ is counted on both sides of the equality the same number of times. From that, it easily follows that actually $\val_\wei(\Sub(H, G)) = \sum_{H \in \mathcal{H}^p} \alpha_{H'} \val_\wei(\Phi(\Hom(H', G))) = \sum_{H \in \mathcal{H}^p} \alpha_{H'} \val_{\wei_{H'}}(\Hom(H', G))$.
\end{proof}

Armed with this statement, we are able to prove the following:

\begin{lemma}
	Let $H$ be a fixed graph, $\G$ be a class of graphs of bounded expansion, and $G\in \Cc$ be a dynamic graph, where edges of $H$ and $G$ are colored with colors $\{1, \ldots, k\}$. Let also $\wei \colon V(H) \times V(G) \to \Z$ be a fixed weight function. Then, there exists a data structure $\mathsf{WeightedSub}_{H, k, \wei}(G)$, which is able to determine $\val_\wei(\Sub(H, G))$ after each update. 
	The data structure processes each update in $\Oh_{\G, H, k}(\log^h |G|)$, where $h = {|H| \choose 2} - 1$, it is initialized for an edgeless graph in $\Oh_{\G, H, k}(|G|)$ time and its space complexity is $\Oh_{\G, H, k}(|G|)$.
\end{lemma}

\begin{proof}
As based on \cref{lem:weighted-sub-to-homo} we have that
$\val_\wei(\Sub(H, G)) = \sum_{H \in \mathcal{H}^p} \alpha_{H'} \val_{\wei_{H'}}(\Hom(H', G))$,
it suffices to maintain $\val_{\wei_{H'}}(\Hom(H', G))$ for each $H' \in \mathcal{H}^p$. However, the size of $\mathcal{H}^p$
is bounded by a function of $|H|$, so we consider it to be a constant. For each such $H'$ we use one instance of $\mathsf{WeightedHom}_{H', k, \wei_{H'}}(G)$ to track that value.
\end{proof}

\subsection{Counting weighted induced subgraph isomorphisms}
And finally, let us analyze $\val_\wei(\ISub(H, G))$. 
A direct application of the Inclusion-Exclusion principle shows that \[\val_\wei(\ISub(H, G)) = \sum_{i=0}^{{|H| \choose 2} - \|H\|} (-1)^i \sum_{H' \in H(+, i, k)} \val_\wei(\Sub(H', G)),\] where $H(+, i, k)$ is the set of all supergraphs of $H$ obtained by adding exactly $i$ new edges and assigning them colors from~$\{1, \ldots, k\}$.

Based on this, we are able to conclude the last required statement:

\begin{lemma} \label{lem:weighted-isub-to-sub}
	Let $H$ be a fixed graph, $\G$ be a class of graphs of bounded expansion, and $G \in \G$ be a dynamic graph, where edges of $H$ and $G$ are colored with colors $\{1, \ldots, k\}$. Let also $\wei \colon V(H) \times V(G) \to \Z$ be a fixed weight function. Then, there exists a data structure $\mathsf{WeightedIndSub}_{H, k, \wei}(G)$, which is able to determine $\val_\wei(\ISub(H, G))$ after each update.
	The data structure processes each update in $\Oh_{\G, H, k}(\log^h |G|)$, where $h = {|H| \choose 2} - 1$, it can be initialized for an edgeless graph in $\Oh_{\G, H, k}(|G|)$ time, and its space complexity is $\Oh_{\G, H, k}(|G|)$.
\end{lemma}
\begin{proof}
	Based on the mentioned equality, the problem of maintaining $\val_\wei(\ISub(H, G))$ easily reduces to the problem of maintaining $\val_\wei(\ISub(H', G))$ for all supergraphs $H'$ of $H$. As there is a constant number of them, the statement follows from \cref{lem:weighted-sub-to-homo}.
\end{proof}

Now, the \cref{lem:weighted-dvorak} is just a combination of \cref{lem:weighted-hom}, \cref{lem:weighted-sub-to-homo} and \cref{lem:weighted-isub-to-sub}. 

\end{document}